\documentclass[12pt]{article}

\usepackage{amsmath,amssymb,amsthm,mathtools}
\usepackage[margin=0.93in]{geometry}
\usepackage{natbib}
\usepackage{xr-hyper}
\usepackage[hidelinks]{hyperref}
\usepackage{booktabs}
\usepackage{graphicx}
\usepackage{float}
\usepackage{array}
\usepackage{xcolor}
\usepackage{colortbl}
\usepackage{bm}
\usepackage{mathrsfs}
\usepackage{enumitem}
\usepackage{makecell}
\usepackage{setspace}
\usepackage[mathlines]{lineno}
\usepackage{commands}
\usepackage{titlesec}

\titlespacing*{\paragraph}
{0pt}
{0.25\baselineskip}
{0.75em}
\titleformat{\paragraph}[runin]
{\normalfont\normalsize\bfseries}
{}
{0pt}
{}
[.]
\providecommand{\depIndAnonymous}{0}
\makeatletter
\newcommand{\suppfallback}[2]{%
  \expandafter\gdef\csname suppfallback@#1\endcsname{#2}}
\suppfallback{suppeffOu.ssbp_support}{S1.2}
\suppfallback{suppeffOu.orthogonalSpecialization}{S1.3}
\suppfallback{suppeffOu.sec:alternativeParametrizationsForTwodimensionalBlocks}{S2}
\suppfallback{suppeffOu.sec:alternativeParametrizationWithinHssbpFamily}{S2.1}
\suppfallback{suppeffOu.sec:alternativeParametrizationOutsideHssbpFamily}{S2.2}
\suppfallback{suppeffOu.sec:forward_exp}{S3.1}
\suppfallback{suppeffOu.sec:forward_lyap}{S3.2}
\suppfallback{suppeffOu.forward_frechet}{S4.1}
\suppfallback{suppeffOu.adjoint_exp}{S4.2}
\suppfallback{suppeffOu.adjoint_lyap}{S4.3}
\suppfallback{suppeffOu.sec:numerical_robustness}{S5}
\suppfallback{suppeffOu.sec:simulation_material}{S6}
\suppfallback{suppeffOu.sec:empirical_details}{S7}
\newcommand{\suppsectionnumber}[1]{%
  \@ifundefined{r@#1}{%
    \@ifundefined{suppfallback@#1}{??}{\csname suppfallback@#1\endcsname}%
  }{\ref*{#1}}}
\newcommand{\suppref}[1]{%
  \@ifundefined{r@#1}{%
    Supp. Material, Section~\suppsectionnumber{#1}%
  }{%
    \hyperref[#1]{Supp. Material, Section~\suppsectionnumber{#1}}%
  }}
\newcommand{\supprefsrange}[2]{%
  \@ifundefined{r@#1}{%
    Supp. Material, Sections~\suppsectionnumber{#1} and~\suppsectionnumber{#2}%
  }{%
    \hyperref[#1]{Supp. Material, Sections~\suppsectionnumber{#1} and~\suppsectionnumber{#2}}%
  }}
\makeatother

\setlist[itemize]{itemsep=2pt, topsep=4pt, parsep=0pt, partopsep=0pt}
\author{}
\usepackage{caption}
\begin{document}
\captionsetup{
  labelfont=bf
}
\doublespacing
\setlength{\abovedisplayskip}{9pt plus 0pt minus 0pt}%
\setlength{\belowdisplayskip}{5pt plus 0pt minus 0pt}%
\setlength{\abovedisplayshortskip}{0pt plus 0pt}%
\setlength{\belowdisplayshortskip}{0pt plus 0pt minus 0pt}%
\renewcommand*{\setdisplayskipstretch}[1]{}
\if\depIndAnonymous1

\vspace{-0.8in}

\else
\begin{flushright}
Version dated: August 17, 2026\\
\end{flushright}

\bigskip
\medskip

\fi
\begin{center}

\noindent{\Large \bf
\begin{singlespace}
Stable Matrix Parametrizations \\
and Structured Adjoints\\
for Ornstein--Uhlenbeck Processes
\end{singlespace}
}
\bigskip

\if\depIndAnonymous1

\vspace{-0.1in}

\else

\noindent{\normalsize \sc
	Filippo Monti$^{1}$ \\
  Andrew Holbrook$^{1}$ \\
  Nathan E. Glatt-Holtz$^{2, 3}$ \\
  Marc A.~Suchard$^{1,4,5}$} \\

\bigskip
\noindent {\small
  \it $^1$ Department of Biostatistics, Jonathan and Karin Fielding School of Public Health, University of California Los Angeles, Los Angeles, CA, USA \\
  \it $^{2}$ Department of Statistics, Indiana University Bloomington, Bloomington, Indiana, USA \\
\it $^{3}$ Department of Mathematics, Indiana University Bloomington, Bloomington, Indiana, USA \\
  \it $^{4}$ Department of Biomathematics, David Geffen School of Medicine at UCLA, University of California Los Angeles, Los Angeles, CA, USA \\
  \it $^{5}$ Department of Human Genetics, David Geffen School of Medicine at UCLA, University of California Los Angeles, Los Angeles, CA, USA
} \\

\fi
\end{center}

\if\depIndAnonymous1
\else
\medskip
\vspace{0.5in}
\noindent{\bf Corresponding author:} Marc A.~Suchard, Departments of Biostatistics, Biomathematics, and Human Genetics,
University of California Los Angeles, 695 Charles E.~Young Dr., South,
Los Angeles, CA 90095-7088, USA; E-mail: \url{msuchard@ucla.edu}

\vspace{0.3in}


\newpage
\fi
\begin{abstract}
\noindent
Ornstein--Uhlenbeck processes with flexible multivariate drift matrices are powerful models for capturing coupled, asymmetric, and damped-oscillatory mean reversion.
However, likelihood-based inference is challenging because the drift matrix must remain Hurwitz stable, while likelihood and gradient evaluations require repeated, costly computation and differentiation of the drift matrix exponential and of solutions to the associated Lyapunov equation.


\noindent We introduce the \emph{Hurwitz smooth spectral block parametrization} (H-SSBP), which represents the drift as a change of basis applied to independent one- and two-dimensional stable blocks. 
Simple scalar constraints enforce Hurwitz stability, while the two-dimensional blocks vary smoothly between real- and complex-eigenvalue regimes, avoiding discrete model selection. 
The H-SSBP represents every real Hurwitz matrix diagonalizable over the complex numbers and has dense, full-measure support within the Hurwitz cone. 
Its block structure reduces the OU transition matrix, stationary and innovation covariances, and their reverse-mode derivatives to constant-size block or block-pair computations plus change-of-basis multiplications. 
Numerical experiments demonstrate dramatic speed-ups over alternatives, particularly for matrix-exponential adjoints and Lyapunov-equation kernels.
Real-data analyses of asynchronous financial data and multivariate phylogenetic traits illustrate the proposed framework in practice.

\end{abstract}


\clearpage

\section{Introduction}\label{effOu.sec:intro}

Ornstein--Uhlenbeck (OU) processes \citep{uhlenbeck_theory_1930} are fundamental building blocks in statistics and machine learning \citep{rasmussen_gaussian_2006}, signal processing and state-space modeling \citep{kalman_new_1960, sarkka_bayesian_2013}, and phylogenetics \citep{hansen_stabilizing_1997, bastide_efficient_2021}.
Multivariate OU processes are powerful models for capturing mean reversion, dependence across coordinates, and damped-oscillatory dynamics.
Specifically, let $\effOuDataRv{}(\effOuTime)\in\mathbb R^{\effOuDim}$ follow a time-homogeneous OU process with equilibrium mean $\effOuEquilMean\in\mathbb R^{\effOuDim}$, drift matrix $\effOuDriftMatrix\in\mathbb R^{\effOuDim\times\effOuDim}$, and positive-definite diffusion covariance $\effOuDiffMat\in\mathbb R^{\effOuDim\times\effOuDim}$ governed by the stochastic differential equation \citep{sarkka_applied_2019}:
\begin{align}\label{effOu.eq:sde}
\dd\effOuDataRv{}(\effOuTime)
=
\effOuDriftMatrix
\left(
\effOuDataRv{}(\effOuTime)
-
\effOuEquilMean
\right)
\dd\effOuTime
+
\effOuDiffMat^{\frac{1}{2}}
\dd\effOuWienerProcess{}(\effOuTime).
\end{align}
Here $\effOuWienerProcess{}(\effOuTime)$ is a $\effOuDim$-dimensional Wiener process, and $\effOuDiffMat^{1/2}(\effOuDiffMat^{1/2})^{\effOuTranspose}=\effOuDiffMat$.
Mean reversion requires $\effOuDriftMatrix$ to be Hurwitz, meaning that every eigenvalue of $\effOuDriftMatrix$ has negative real part.
Many applications place latent and observed OU processes along graphical structures such as time series and trees.
The process then evolves continuously along graph edges, induces Gaussian dependence among node variables, and can be coupled with observation models for measurement noise or partial observations \citep{kalman_new_1960, sarkka_bayesian_2013, felsenstein_phylogenies_1985, bastide_efficient_2021}.

The drift matrix $\effOuDriftMatrix$ is the main source of flexibility in a multivariate OU model, but also the main obstacle to scalable, gradient-based inference.
Its eigenvalues determine the modal mean-reversion behavior: real eigenvalues produce non-oscillatory exponential relaxation, while complex-conjugate eigenvalues are essential for capturing coupled trajectories with damped oscillatory decay toward equilibrium.
On the computational side, the drift appears inside expensive matrix functions such as matrix exponentials and Lyapunov solves, which need to be repeatedly evaluated and differentiated during inferential routines.
A useful parametrization of $\effOuDriftMatrix$ should therefore enforce stability without ruling out asymmetric, coupled, or damped-oscillatory dynamics, and should expose enough structure to make likelihood and gradient calculations computationally efficient.

\paragraph{Existing parametrizations} Existing drift parametrizations trade off \emph{expressiveness}, \emph{stability enforcement}, and \emph{computational tractability}. 
Diagonal drift matrices are stable and computationally simple, but remove cross-dimensional mean-reverting coupling.
Symmetric negative-definite drift matrices allow coupling, but rule out rotational drift behavior and damped oscillations from complex eigenvalues \citep{horn_matrix_2012}.
Real Schur decompositions and Lyapunov-based constructions provide more expressive stable representations \citep{schur_ueber_1909, golub_matrix_2013, lancaster_algebraic_1995}, but they retain dense matrix-function and linear-solve operations inside likelihood and gradient calculations \citep{moler_nineteen_1978, higham_scaling_2005, higham_scaling_2009, al-mohy_computing_2011}.
As a parametrization, a fixed Schur block layout also fixes the number and placement of real and complex spectral blocks in advance, while Lyapunov-based constructions introduce additional matrix inversions or dense linear solves that can worsen the likelihood geometry and introduce numerical fragility.

\paragraph{Proposed parametrization: H-SSBP} We propose the \emph{Hurwitz smooth spectral block parametrization} (H-SSBP), which represents the drift as a change of basis applied to a block-diagonal matrix with scalar and two-dimensional stable blocks,
$\effOuDriftMatrix=\effOuRotMat \effOuBlkDiag \effOuRotMat^{-1}$ (Def.~\ref{effOu.def:SSBP} and \ref{effOu.def:H-SSBP}). 
The $2\times2$ blocks are smooth functions of three parameters $(\effOuCartBlockRate_{\effOuIdxB}, \effOuCartNormalizingRate_{\effOuIdxB}, \effOuBlkT{\effOuIdxB})$:
\begin{align}\label{effOu.eq:block_def_intro}
    \effOuBlk{\effOuIdxB}(\effOuCartBlockRate_{\effOuIdxB},\effOuCartNormalizingRate_{\effOuIdxB},\effOuBlkT{\effOuIdxB})
    :=
    \begin{pmatrix}
      \effOuCartBlockRate_{\effOuIdxB}
        &
      \effOuCartBlockRate_{\effOuIdxB}\effOuCartNormalizingRate_{\effOuIdxB}
      +
      \effOuBlkT{\effOuIdxB}
      \\
      \effOuCartBlockRate_{\effOuIdxB}\effOuCartNormalizingRate_{\effOuIdxB}
      -
      \effOuBlkT{\effOuIdxB}
        &
      \effOuCartBlockRate_{\effOuIdxB}
    \end{pmatrix}.
\end{align}
Hurwitz stability is enforced by the independent constraints $\effOuCartBlockRate_{\effOuIdxB}<0$ and $|\effOuCartNormalizingRate_{\effOuIdxB}|<1$ (Prop.~\ref{effOu.prop:stability}), while varying the unrestricted skew-symmetric parameter $\effOuBlkT{\effOuIdxB}\in\mathbb{R}$ allows smooth transitions between real and complex spectral regimes.
The block structure in $\effOuBlkDiag$ also organizes computation and is designed to improve numerical robustness over alternative dense parametrizations.
In particular, the expensive matrix-function operations reduce to independent constant-size block and block-pair computations, leaving change-of-basis multiplications as the dominant dense operations.

When the change-of-basis matrix is allowed to range over all invertible matrices, the H-SSBP contains every real Hurwitz matrix that is diagonalizable over $\mathbb{C}$, as well as certain defective matrices. 
Consequently, its image is dense in and has full Lebesgue measure within the real Hurwitz cone (Prop.~\ref{prop:smbp_support}).
We also discuss an \emph{orthogonal variant}
$(\effOuRotMat^{\effOuTranspose}\effOuRotMat=\effOuIdentity)$
that improves conditioning and interpretability but restricts the class of matrices represented (Cor.~\ref{cor:orthogonal_smbp_support}).

\paragraph{Contributions} This paper makes three main contributions.
\begin{itemize}
\item[1] We introduce the H-SSBP as a stable drift parametrization for likelihood-based inference in multivariate OU models (Sec.~\ref{effOu.sec:SSBP}) and characterize its properties relative to competing stable-matrix parametrizations (Secs.~\ref{effOu.sec:landscape} and~\ref{effOu.sec:SSBP}).
\item[2] We leverage the H-SSBP block structure to improve the scalability of exact Gaussian message passing likelihood evaluations for OU processes on rooted directed trees.
We derive reverse-mode gradients in which the matrix-exponential and Lyapunov pullbacks reduce, after a change of basis, to constant-size block or block-pair kernels (Sec.~\ref{effOu.sec:adjoints}).
\item[3] We show that, when the drift is homogeneous across edges, working directly in the block basis substantially reduces repeated dense matrix operations (Sec.~\ref{effOu.sec:homogeneous_block_basis}).
When $\effOuDriftMatrix$ and $\effOuDiffMat$ additionally share a common orthogonal block basis, the process fully decouples in the block basis into independent one- and two-dimensional components (Sec.~\ref{effOu.sec:homogeneous_block_basis_simultaneous}).
\end{itemize}

\noindent We then support the theoretical results with numerical benchmarks (Sec.~\ref{effOu.sec:results}), simulations (Sec.~\ref{effOu.sec:simulations}), and data examples in finance (Sec.~\ref{effOu.sec:bitmex_trade_data}) and phylogenetics (Sec.~\ref{effOu.sec:phylogeneticExample}).

\section{OU Processes on Graphs}
\label{effOu.sec:model}

This section formalizes the OU process model on rooted directed trees, specifies the transition model, and identifies the two computational targets---the matrix exponential and the Lyapunov equation---that motivate the development in Sections~\ref{effOu.sec:SSBP}--\ref{effOu.sec:adjoints}.

\paragraph{Graph structure}
Let $\effOuDag=(\effOuSetNodes,\effOuEdgeSet)$ be a rooted directed tree, with all edges oriented away from the root.
Thus, every non-root node $\effOuIdx$ has a unique parent $\effOuParentNodeOperator{\effOuIdx}$, connected to it by an edge of length $\effOuEdgeLength{\effOuIdx}>0$.
This class of graphs includes, for example, chains and branching phylogenetic trees.
We focus on this structure because unique parentage associates each non-root state with a single ancestral state and avoids the need to specify how information from multiple parents is combined.
However, many of the results developed below can be extended to more general directed acyclic graphs.

\paragraph{Transition model} Here and throughout, a real square matrix $\effOuDriftMatrix{}$ is called \emph{Hurwitz stable} if every eigenvalue of $\effOuDriftMatrix{}$ has strictly negative real part.
For an edge ending at node $\effOuIdx$, consider
the $\effOuDim$-dimensional OU process in Equation~\eqref{effOu.eq:sde}, with equilibrium mean $\effOuEquilMeanEdge{\effOuIdx} \in\mathbb{R}^{\effOuDim}$, Hurwitz stable drift matrix $\effOuDriftMatrixEdge{\effOuIdx} \in\mathbb{R}^{\effOuDim \times \effOuDim}$, and symmetric positive-definite diffusion covariance $\effOuDiffMatEdge{\effOuIdx} \in \mathbb{R}^{\effOuDim \times \effOuDim}$.
Let $\effOuStatVarEdge{\effOuIdx}$ denote the edge-level stationary covariance, namely the unique symmetric positive-definite solution of the continuous Lyapunov equation \citep{lyapunov_general_1992, hammarling_numerical_1982, lancaster_algebraic_1995}
\begin{align}\label{effOu.eq:lyapunovEquation}
\effOuDriftMatrixEdge{\effOuIdx} \effOuStatVarEdge{\effOuIdx}
+
\effOuStatVarEdge{\effOuIdx} \effOuDriftMatrixEdge{\effOuIdx}^{\effOuTranspose}
+
\effOuDiffMatEdge{\effOuIdx}
=
0.
\end{align}
The OU process evolves along the graph $\effOuDag$ according to the graph's conditional independence structure.
Writing $\effOuDataRv{\effOuIdx}$ for the process value at vertex $\effOuIdx$, the transition kernel along the edge ending at $\effOuIdx$ is

\vspace{-1.2cm}

\begin{align}\label{effOu.eq:transition}
  \effOuDataRv{\effOuIdx}
  \mid
  \effOuDataRv{\effOuParentNodeOperator{\effOuIdx}}
  \sim
  \effOuNorm \Bigl(
    \effOuActual{\effOuIdx}\,
    \effOuDataRv{\effOuParentNodeOperator{\effOuIdx}}
    +
    (\effOuIdentity-\effOuActual{\effOuIdx})
    \effOuEquilMeanEdge{\effOuIdx},\;
    \effOuEdgeCov{\effOuIdx}
  \Bigr),
\end{align}
where $\effOuActual{\effOuIdx}
=
\effOuExpMatrix{\effOuDriftMatrixEdge{\effOuIdx} \effOuEdgeLength{\effOuIdx}}$
is the edge-specific actualization matrix and
$\effOuEdgeCov{\effOuIdx}
=
\effOuStatVarEdge{\effOuIdx}
-
\effOuActual{\effOuIdx}
\effOuStatVarEdge{\effOuIdx}
\effOuActual{\effOuIdx}^{\effOuTranspose}$
is the edge innovation covariance.
Thus $\effOuDriftMatrixEdge{\effOuIdx}$ governs both deterministic propagation through $\effOuActual{\effOuIdx}$ and stochastic innovation through the stationary covariance $\effOuStatVarEdge{\effOuIdx}$ solving Equation~\eqref{effOu.eq:lyapunovEquation}.
To complete the model, we also define a root distribution $\effOurootDistribution$.

\paragraph{Observation model}
For the main development, we consider \emph{direct} observations at a subset of vertices.
Gaussian measurement error can be incorporated without altering any of the structural computations presented below.
Missing and partially missing data can also be analytically integrated within this framework \citep{bastide_efficient_2021, hassler_inferring_2022}.

\paragraph{Primary parametrization}
Suppressing the edge indices, the matrices $\effOuDriftMatrix$, $\effOuDiffMat$, and $\effOuStatVar$ are linked through the Lyapunov equation~\eqref{effOu.eq:lyapunovEquation} and therefore cannot be specified independently.
This work takes $\effOuDriftMatrix$ and $\effOuDiffMat$ as primary parameters and obtains $\effOuStatVar$ as the solution of~\eqref{effOu.eq:lyapunovEquation}.
Alternatively, one may specify $\effOuDiffMat$ and $\effOuStatVar$ and parametrize the drift through the Lyapunov identity as 
$\effOuDriftMatrix=\left(-\tfrac{1}{2}\effOuDiffMat+\effOuSkewSymmatrixMat\right)\effOuStatVar^{-1}$ where $\effOuSkewSymmatrixMat$ is a skew-symmetric matrix.
We discuss this case in more detail in Section~\ref{effOu.sec:landscape}.
A third possibility is to take $\effOuDriftMatrix$ and $\effOuStatVar$ as primary parameters and define
$\effOuDiffMat=-(\effOuDriftMatrix\effOuStatVar+\effOuStatVar\effOuDriftMatrix^{\effOuTranspose})$.
Since $\effOuDiffMat$ must be positive definite, this choice introduces a coupled feasibility condition on $\effOuDriftMatrix$ and $\effOuStatVar$, which is difficult to enforce directly.

\paragraph{Inference and computation}
Exact likelihood-based inference requires repeated evaluation of $\effOuActual{}$, $\effOuEdgeCov{}$, and their derivatives with respect to $\effOuDriftMatrixEdge{}$ and $\effOuDiffMatEdge{}$.
For an unstructured drift, the matrix exponentials and Lyapunov solves, as well as their gradients, require expensive general dense algorithms; the parametrization proposed in this work (H-SSBP; Def.~\ref{effOu.def:H-SSBP}) replaces these primitives with cheap structured block computations.

\section{The Landscape of Stable Matrix Parametrizations}
\label{effOu.sec:landscape}

We frame the parametrization problem in terms of three axes: \emph{expressiveness}, \emph{stability enforcement}, and \emph{computational complexity}.
Expressiveness measures the size of the set of stable matrices reachable by the parametrization; stability enforcement refers to how easily the Hurwitz constraint---that all eigenvalues have negative real parts---is maintained during optimization; computational complexity concerns the per-iteration cost of computing likelihood and gradients through the matrix exponential and the Lyapunov equation.
No existing approach dominates on all three axes simultaneously. 
Table~\ref{tab:landscape} provides a compact summary of the approaches presented below and shows how H-SSBP is designed to occupy a favorable region of this landscape.
\definecolor{proposedrow}{RGB}{244,250,235}

\begin{table}[ht]
\centering
\small
\setlength{\tabcolsep}{5pt}
\renewcommand{\arraystretch}{1.35}
\begin{tabular}{
  L{3.5cm}
  L{3.7cm}
  L{3.0cm}
  L{4.5cm}
}
\toprule
\textbf{Parametrization}
  & \textbf{Expressiveness}
  & \textbf{Stability}
  & \textbf{Dominant computations} \\
\midrule
\rowcolor{rowalt}
Diagonal
  & \bad{Real $\lambda$; no coupling}
  & \good{$s_i < 0$}
  & 
  \good{scalar arithmetic}  \\

Symmetric
  & \med{Real $\lambda$ coupling} \med{via rotation}
  & \good{$d_i < 0$}
  &  \med{rotations only} \\


\rowcolor{rowalt}
Real Schur $\effOuRotMat \effOuUpperQuasiTriangular \effOuRotMat^{\effOuTranspose}$ ($\effOuUpperQuasiTriangular$ upper quasi-tri.)
  & \med{Fixed number} 
  \med{of complex pairs}
  &  \med{Per-block}
  \med{conditions}
  & 
  \bad{Schur–Parlett expm}
  \bad{Bartels--Stewart Lyap.} \\

\rowcolor{proposedrow}
H-SSBP 
$\effOuRotMat \,\effOuBlkDiag\,\effOuRotMat^{-1}$
(this work)
  & \good{Diagonaliz. over $\mathbb{C}$}$^*$ {\footnotesize (and defective with $2\times 2$ Jordan blocks)}
  & 
    \good{$\effOuCartBlockRate_{\effOuIdxB}{<}0$, $|\effOuCartNormalizingRate_{\effOuIdxB}|{<}1$}
  & 
  \med{change-of-basis only \hspace{0.3cm}}
  \\

\rowcolor{rowalt}
Lyapunov-based $\left(-\tfrac{1}{2}\effOuDiffMat + \effOuSkewSymmatrixMat\right) \effOuStatVar^{-1}$
  & \good{All stable matrices}
  &  \good{By construction}
  & 
  \bad{Dense Schur/Pad\'e expm}
   \\
   Generic SVD/QR
$\mathbf U\mathbf D\mathbf V^{\effOuTranspose}$ or $\mathbf Q\mathbf R$
  & \good{All stable matrices}
  & \bad{Not possible}
  \bad{directly}
  & \bad{Dense Schur/Pad\'e expm}
  \bad{Bartels--Stewart Lyap.} \\
\bottomrule
\noalign{\vskip 2pt}
\multicolumn{4}{@{}p{14.8cm}@{}}{%
\scriptsize
$^{*}$
In the Hurwitz cone, matrices diagonalizable over $\mathbb{C}$ are dense and have full Lebesgue measure (Prop.~\ref{prop:smbp_support}); nondiagonalizable cases have measure $0$.
} \\
\end{tabular}
\caption{Comparison of stable matrix parametrizations on three axes.
\textit{Expressiveness}: set of Hurwitz matrices reachable.
\textit{Stability}: ease of enforcing $\mathrm{Re}(\lambda_i)<0$ for drift eigenvalues.
\textit{Dominant computations}: operations contributing to the leading computational cost.
} \label{tab:landscape}
\end{table}

\paragraph{Diagonal parametrization} The simplest choice writes the drift as a diagonal matrix
$\effOuDriftMatrix = \mathrm{diag}(s_1,\ldots,s_p)$ with negative diagonal elements $s_i < 0$.
This reduces matrix exponentials to scalar exponentials and Lyapunov equations to independent entrywise calculations.
The price is the loss of off-diagonal coupling in the mean-reverting dynamics.

\paragraph{Symmetric negative-definite parametrization}
Restricting the drift to a symmetric negative definite matrix introduces inter-component coupling while keeping all drift eigenvalues real and negative.
Under a spectral decomposition representation $\effOuDriftMatrix = \effOuRotMat \effOuDiagonalMatrix \effOuRotMat^{\effOuTranspose}$ with $\effOuRotMat$ orthogonal and $\effOuDiagonalMatrix = \mathrm{diag}(\effOuDiagonalMatrixElement{i})$, $\effOuDiagonalMatrixElement{i}<0$,
the matrix exponential is $\effOuRotMat e^{\effOuDiagonalMatrix \effOuEdgeLength{}}\effOuRotMat^{\effOuTranspose}$ and the Lyapunov equation decouples in the eigenbasis.
The dominant cost becomes rotating into and out of the eigenbasis.
The symmetry restriction is itself a substantive model assumption: it rules out oscillatory dynamics in which the equilibrating force has a rotational component.

\paragraph{Real Schur decomposition}
The real Schur decomposition $\effOuDriftMatrix=\effOuRotMat\effOuUpperQuasiTriangular\effOuRotMat^{\effOuTranspose}$, with $\effOuRotMat$ orthogonal and $\effOuUpperQuasiTriangular$ upper quasi-triangular, is a standard real representation of matrices with possibly complex eigenvalues \citep{schur_ueber_1909,francis_qr_1961}.
As a parametrization, a fixed Schur layout determines the number of complex-conjugate eigenvalue pairs in advance, introducing a discrete choice at the real--complex boundary.
Computationally, standard Schur-based matrix-exponential and Lyapunov algorithms use Parlett- and Bartels--Stewart-type recursions, respectively, which propagate through the upper quasi-triangular couplings in $\effOuUpperQuasiTriangular$ and incur cubic costs with substantial constants \citep{bartels_algorithm_1972,higham_functions_2008}.
Restricting to real normal matrices eliminates these couplings by reducing $\effOuUpperQuasiTriangular$ to independent $1\times1$ and $2\times2$ spectral blocks, but the block partition must still change when a real eigenvalue pair becomes complex.
Stability conditions for the Schur blocks are discussed in \suppref{suppeffOu.sec:alternativeParametrizationOutsideHssbpFamily}.


\paragraph{Generic dense, SVD, and QR parametrizations}
General nonsymmetric drift matrices can be parametrized directly through their entries or represented using generic decompositions such as SVD or QR factorizations. 
These parametrizations can represent every real matrix, and hence every real Hurwitz matrix.
However, Hurwitz stability cannot be enforced directly,
and neither the entrywise nor the factorized representations generally simplify matrix exponentials or Lyapunov solves.

\paragraph{Lyapunov-based parametrization}
An alternative way to enforce stability is to parametrize the drift through the Lyapunov identity.
Writing $\effOuLyapProd=\effOuDriftMatrix\effOuStatVar$, the relation $\effOuDriftMatrix\effOuStatVar+\effOuStatVar\effOuDriftMatrix^{\effOuTranspose}=-\effOuDiffMat$ becomes $\effOuLyapProd+\effOuLyapProd^{\effOuTranspose}=-\effOuDiffMat$.
Hence, introducing a skew-symmetric matrix $\effOuSkewSymmatrixMat^{\effOuTranspose}=-\effOuSkewSymmatrixMat$, we can express $\effOuLyapProd=-\tfrac{1}{2}\effOuDiffMat+\effOuSkewSymmatrixMat$ and then
$\effOuDriftMatrix=\left(-\tfrac{1}{2}\effOuDiffMat+\effOuSkewSymmatrixMat\right)\effOuStatVar^{-1}$.
This spans all OU processes with stationary covariance $\effOuStatVar$ and diffusion $\effOuDiffMat$, with $\effOuSkewSymmatrixMat$ encoding the non-reversible rotational component of the dynamics \citep{lelievre_optimal_2013,duncan_variance_2016}.
Although this construction enforces stability and satisfies the Lyapunov identity by construction, the resulting $\effOuDriftMatrix$ is generally nonnormal, so its matrix exponential and adjoint require costly dense matrix-function algorithms.
The representation also requires solves involving $\effOuStatVar$, which can become poorly conditioned when the prescribed stationary covariance is nearly singular.

\paragraph{Machine learning approaches} Related structured state-space parametrizations in machine learning, such as HiPPO and S4, also exploit spectral, diagonal, or low-rank structure for fast linear dynamics \citep{gu_hippo_2020, gu_efficiently_2022}. 
These methods, however, are primarily designed for efficient long-sequence convolution or recurrence rather than likelihood-based inference for continuous-time stochastic processes.

\section{The Hurwitz Smooth Spectral Block Parametrization}
\label{effOu.sec:SSBP}
Section~\ref{effOu.sec:landscape} highlights a central tension: computationally tractable stability-preserving parametrizations often restrict the class of admissible drift matrices, whereas more expressive representations retain expensive dense matrix operations.
We now introduce a new parametrization of the drift matrix that combines broad spectral flexibility with tractable stability constraints and matrix-function computations.
We first define a general representation based on a global change of basis and structured scalar or two-dimensional blocks, and then derive conditions ensuring Hurwitz stability.
We subsequently characterize its support and show how the block structure simplifies matrix exponentials and Lyapunov equations.

\subsection{Construction and definition}
\label{effOu.sec:smbp_def}
The construction rests on two observations.
First, $2\times2$ real blocks are the smallest components that can accommodate either two real eigenvalues or a complex-conjugate pair (\suppref{suppeffOu.sec:alternativeParametrizationsForTwodimensionalBlocks}).
Second, every real matrix diagonalizable over $\mathbb{C}$ admits a real invariant-subspace decomposition with blocks of size at most two, and such matrices form a dense, full-measure subset of the square matrices \citep{horn_matrix_2012,golub_matrix_2013}.
These observations motivate representing the drift as a global change of basis applied to a block-diagonal matrix.
Since we allow each block to have dimension at most two, its matrix exponential and the associated block-pair Lyapunov equations reduce to constant-size computations \citep{lancaster_algebraic_1995,higham_functions_2008}.

\begin{definition}[Smooth Spectral Block Parametrization]\label{effOu.def:SSBP}
A matrix $\effOuDriftMatrix \in \mathbb{R}^{\effOuDim \times \effOuDim}$ is represented under the \emph{smooth spectral block parametrization} (SSBP) if it can be written as

\vspace{-0.5in}

\begin{align}\label{effOu.eq:SSBP}
  \effOuDriftMatrix
  =
  \effOuRotMat\,
  \effOuBlkDiag\,
  \effOuRotMat^{-1},
\end{align}

\vspace{-0.1in}

\noindent where $\effOuRotMat \in GL_{\effOuDim}(\mathbb{R})$ is an invertible change-of-basis matrix and $\effOuBlkDiag=\operatorname{bdiag}(\effOuBlk{1},\ldots,\effOuBlk{\effNBlock})$ is block diagonal.
If $\effOuDim$ is even, all blocks $\effOuBlk{\effOuIdxB}$ are $2\times2$ matrices of the form
%
%
\begin{align}\label{effOu.eq:block_def}
    \effOuBlk{\effOuIdxB}
    =
    \effOuBlkDiag(\effOuCartBlockRate_{\effOuIdxB},\effOuCartNormalizingRate_{\effOuIdxB},\effOuBlkT{\effOuIdxB})
    :=
    \begin{pmatrix}
      \effOuCartBlockRate_{\effOuIdxB}
        &
      \effOuCartBlockRate_{\effOuIdxB}\effOuCartNormalizingRate_{\effOuIdxB}
      +
      \effOuBlkT{\effOuIdxB}
      \\
      \effOuCartBlockRate_{\effOuIdxB}\effOuCartNormalizingRate_{\effOuIdxB}
      -
      \effOuBlkT{\effOuIdxB}
        &
      \effOuCartBlockRate_{\effOuIdxB}
    \end{pmatrix},
\end{align}
for real $(\effOuCartBlockRate_{\effOuIdxB},\effOuCartNormalizingRate_{\effOuIdxB}, \effOuBlkT{\effOuIdxB})$.
If $\effOuDim$ is odd, one block is scalar, $\effOuBlk{\effOuIdxB}=\effOuCartBlockRate_{\effOuIdxB}$, while all others are $2\times2$ blocks of the form~\eqref{effOu.eq:block_def}.
\end{definition}

\paragraph{Smoothness at the real--complex boundary} For a $2\times2$ block $\effOuBlk{\effOuIdxB}$ parametrized as in Equation~\eqref{effOu.eq:block_def}, denote the product of the off-diagonal elements by
\[
  \effOuBlkDelta{\effOuIdxB}
  :=
  \effOuCartBlockRate_{\effOuIdxB}^{2}\effOuCartNormalizingRate_{\effOuIdxB}^{2}
  -
  \effOuBlkT{\effOuIdxB}^{2}.
\]
Then, the eigenvalues of $\effOuBlk{\effOuIdxB}$ are
\begin{align}\label{effOu.eq:eigenvalues}
  \lambda_{\effOuIdxB}^{\pm}
  =
  \effOuCartBlockRate_{\effOuIdxB}
  \pm
  \begin{cases}
    \sqrt{\effOuBlkDelta{\effOuIdxB}},
    & \effOuBlkDelta{\effOuIdxB}>0,
    \\
    0,
    & \effOuBlkDelta{\effOuIdxB}=0,
    \\
    \mathrm{i}\sqrt{-\effOuBlkDelta{\effOuIdxB}},
    & \effOuBlkDelta{\effOuIdxB}<0.
  \end{cases}
\end{align}
The sign of $\effOuBlkDelta{\effOuIdxB}$ determines whether $\effOuBlk{\effOuIdxB}$ has two real eigenvalues, a repeated real eigenvalue, or a complex-conjugate pair.
Although the eigenvalues themselves, viewed as functions of $(\effOuCartBlockRate_{\effOuIdxB},\effOuCartNormalizingRate_{\effOuIdxB},\effOuBlkT{\effOuIdxB})$, are not differentiable at $\effOuBlkDelta{\effOuIdxB}=0$, the entries of $\effOuBlk{\effOuIdxB}$ and analytic matrix functions of the block remain smooth in $(\effOuCartBlockRate_{\effOuIdxB},\effOuCartNormalizingRate_{\effOuIdxB},\effOuBlkT{\effOuIdxB})$.
The parametrization therefore \emph{crosses the real--complex boundary} without a discrete change of coordinates.

\subsection{Stability and support}
We now impose sufficient conditions on the block parameters to ensure Hurwitz stability and characterize the resulting matrix-space support.

\begin{proposition}
  \label{effOu.prop:stability}
A matrix $\effOuDriftMatrix$ represented under the SSBP is \textbf{Hurwitz stable} if

\vspace{-1.1cm}

\begin{align} \label{effOu.eq:stability}
  \effOuCartBlockRate_{\effOuIdxB}<0 \quad \text{and} \quad \effOuCartNormalizingRate_{\effOuIdxB}\in(-1,1),
\end{align}

\vspace{-0.4cm}

\noindent for each $2\times 2$ block $\effOuIdxB$.
When $\effOuDim$ is odd, we also require the scalar block to be negative.
\end{proposition}
\noindent
\begin{proof}
For a two-dimensional block,
\begin{align}
\frac{\effOuBlk{\effOuIdxB}+\effOuBlk{\effOuIdxB}^{\effOuTranspose}}{2}
=
\effOuCartBlockRate_{\effOuIdxB}
\begin{pmatrix}
1 & \effOuCartNormalizingRate_{\effOuIdxB}\\
\effOuCartNormalizingRate_{\effOuIdxB} & 1
\end{pmatrix}.
\end{align}
Its eigenvalues are $\effOuCartBlockRate_{\effOuIdxB}(1+\effOuCartNormalizingRate_{\effOuIdxB})$ and $\effOuCartBlockRate_{\effOuIdxB}(1-\effOuCartNormalizingRate_{\effOuIdxB})$.
Under $\effOuCartBlockRate_{\effOuIdxB}<0$ and $|\effOuCartNormalizingRate_{\effOuIdxB}|<1$, both are negative, so the symmetric part of $\effOuBlk{\effOuIdxB}$ is negative definite.
Every eigenvalue of $\effOuBlk{\effOuIdxB}$ therefore has negative real part.
The scalar block is Hurwitz when it is negative, block-diagonal assembly preserves stability, and similarity by $\effOuRotMat$ preserves eigenvalues.
\end{proof}
%
\begin{remark}[Controlling oscillations]
  The unconstrained parameter $\effOuBlkT{\effOuIdxB}$ controls transitions between real and complex regimes without approaching the stability boundary.
\end{remark}

\noindent We can now define the new stable matrix parametrization and study its support over the set of stable matrices.

\begin{definition}[Hurwitz Smooth Spectral Block Parametrization]\label{effOu.def:H-SSBP}
A matrix $\effOuDriftMatrix \in \mathbb{R}^{\effOuDim \times \effOuDim}$ is represented under the \textbf{Hurwitz smooth spectral block parametrization} (\textbf{H-SSBP}) if it belongs to the SSBP family and its block parameters satisfy the conditions from Proposition~\ref{effOu.prop:stability}.
\end{definition}

\begin{proposition}\label{prop:smbp_support}
The following hold for H-SSBP matrices.
\begin{itemize}
        \item[i.] (Full spectral support) Scalar blocks can realize arbitrary negative real eigenvalues; $2\times2$ blocks can realize any stable real or complex-conjugate eigenvalue pair.
    \item[ii.] (Dense matrix-space support) Allowing $\effOuRotMat$ to range over all invertible matrices, the H-SSBP contains every real Hurwitz matrix diagonalizable over $\mathbb C$.
    It also contains defective structures whose nontrivial Jordan blocks have length two and are associated with real eigenvalues.
    Consequently, its image is \emph{dense} in the real Hurwitz cone in the Euclidean topology inherited from $\mathbb{R}^{\effOuDim\times\effOuDim}$ and has full Lebesgue measure within that cone.
\end{itemize}
\end{proposition}
\begin{corollary}[Orthogonal H-SSBP support]\label{cor:orthogonal_smbp_support}
If $\effOuRotMat$ is orthogonal, the H-SSBP represents the stable matrices that are orthogonally block-decomposable into H-SSBP scalar and $2\times2$ components. 
Every such matrix has negative-definite symmetric part.
\end{corollary}
The proposition separates block-level spectral support from matrix-space support.
At the block level, the H-SSBP can realize any stable real or complex eigenvalue configuration.
At the matrix level, the unrestricted invertible-basis version has dense, full-measure support in the Hurwitz cone, while excluding longer real Jordan chains and defective complex eigenvalue structures requiring real blocks larger than two.
The proof is presented in \suppref{suppeffOu.ssbp_support}.
The orthogonal H-SSBP trades expressiveness for improved conditioning since $\effOuRotMat^{-1}=\effOuRotMat^{\effOuTranspose}$ \citep{golub_matrix_2013}.
It contains all real normal Hurwitz matrices, obtained when each two-dimensional block satisfies either $\effOuCartNormalizingRate_{\effOuIdxB}=0$ or $\effOuBlkT{\effOuIdxB}=0$.
More generally, allowing both parameters to be nonzero yields nonnormal dynamics within the corresponding invariant subspace.
See the \suppref{suppeffOu.orthogonalSpecialization} for further details on the orthogonal H-SSBP.

\subsection{Block exponential and Lyapunov primitives}
\label{effOu.sec:smbp_computationalAdvantages}

\paragraph{Matrix exponentials}
The H-SSBP factorization reduces computation of the edge actualization matrix to

\vspace{-1.3cm}

\begin{align}\label{effOu.eq:act_factored}
  \effOuActual{}
  = e^{\effOuDriftMatrix\effOuEdgeLength{}} = 
  \effOuRotMat
  e^{\effOuBlkDiag\effOuEdgeLength{}}
  \effOuRotMat^{-1}.
\end{align}
Since $\effOuBlkDiag$ is block diagonal, the exponential decomposes into independent scalar and $2\times2$ block exponentials.
The equal-diagonal form of each $2\times2$ block can be further exploited to derive more efficient closed-form expressions (see \suppref{suppeffOu.sec:forward_exp}).
Evaluating all scalar and $2\times2$ blocks of $\effOuExpMatrix{\effOuBlkDiag\effOuEdgeLength{}}$ requires $\mathcal{O}(\effOuDim)$ work, so the dominant cost, when required, is the change of basis rather than a general dense matrix-exponential computation, such as one based on a Schur decomposition or Pad\'e approximation.

\paragraph{Lyapunov equation}
The same block structure reduces the stationary covariance computation.
Define block-basis quantities for the diffusion covariance and stationary covariance: $\effOuDiffMatRot=\effOuRotMat^{-1}\effOuDiffMat\effOuRotMat^{-\effOuTranspose}$ and $\effOuStatVarRot=\effOuRotMat^{-1}\effOuStatVar\effOuRotMat^{-\effOuTranspose}$.
Substituting $\effOuDriftMatrix=\effOuRotMat\effOuBlkDiag\effOuRotMat^{-1}$ into Equation~\eqref{effOu.eq:lyapunovEquation} yields

\vspace{-1cm}

\begin{align}\label{effOu.eq:lyap_D_basis_main}
\effOuBlkDiag\effOuStatVarRot
+
\effOuStatVarRot\effOuBlkDiag^{\effOuTranspose}
=
-\effOuDiffMatRot.
\end{align}
Because $\effOuBlkDiag$ is block diagonal, this equation decouples over block pairs.
For each pair $(\effOuIdxB,\effOuIdxBTwo)$,
\begin{align}\label{effOu.eq:block_lyap_pair_main}
\effOuBlk{\effOuIdxB}
\effOuStatVarRot_{\effOuIdxB\effOuIdxBTwo}
+
\effOuStatVarRot_{\effOuIdxB\effOuIdxBTwo}
\effOuBlk{\effOuIdxBTwo}^{\effOuTranspose}
=
-\effOuDiffMatRot_{\effOuIdxB\effOuIdxBTwo},
\end{align}
where $\effOuStatVarRot_{\effOuIdxB\effOuIdxBTwo}$ and $\effOuDiffMatRot_{\effOuIdxB\effOuIdxBTwo}$ denote the submatrices with row block $\effOuIdxB$ and column block $\effOuIdxBTwo$.
This is a Sylvester equation with at most four scalar unknowns.
There are $\mathcal{O}(\effOuDim^{2})$ constant-size block pairs, so solving for $\effOuStatVarRot$ costs $\mathcal{O}(\effOuDim^{2})$ and is embarrassingly parallel.
When the stationary covariance is required in the original basis, the dominant cost is the cubic but low-constant change of basis.
This replaces high-constant cubic general-purpose Lyapunov solves with parallel constant-size block computations.
In \suppref{suppeffOu.sec:forward_lyap}, we derive an algorithm to solve these decoupled Lyapunov equations by exploiting the structure of the block parametrization \eqref{effOu.eq:block_def}.

\section{Exact Likelihoods and Gradient-Based Inference}
\label{effOu.sec:adjoints}
Gradient-based maximum-likelihood and Bayesian methods require derivatives of the exact marginal log-likelihood $\effLogLik$ with respect to a potentially large collection of coupled parameters \citep{duane_hybrid_1987,nocedal_numerical_2006,neal_mcmc_2011}.
Reverse-mode differentiation, often called the adjoint method, is particularly appropriate when a scalar objective depends on many parameters \citep{griewank_evaluating_2008,baydin_automatic_2018,margossian_review_2019}.
In this framework, a forward pass first evaluates $\effLogLik$ and the intermediate quantities on which it depends.
A reverse pass then propagates an adjoint for each intermediate quantity.
For a vector-valued intermediate $\mathbf z$ and a matrix-valued intermediate $\mathbf Z$, we define their adjoints by the first-order differential

\vspace{-1.3cm}

\begin{align}
\dd\effLogLik
&=
\overline{\mathbf z}^{\effOuTranspose}\dd\mathbf z
+
\left\langle
\overline{\mathbf Z},
\dd\mathbf Z
\right\rangle
+
\cdots,
&
\left\langle\mathbf U,\mathbf V\right\rangle
&=
\tr\!\left(\mathbf U^{\effOuTranspose}\mathbf V\right).
\end{align}
Thus, an adjoint records the first-order sensitivity of the final scalar objective to a perturbation of the corresponding intermediate quantity.
More generally, suppose that an intermediate quantity is obtained from a differentiable map $\effOuAdjointIntroOutput= f(\effOuAdjointIntroInput)$ whose differential is $\dd\effOuAdjointIntroOutput=\mathrm{D} f(\effOuAdjointIntroInput)[\dd\effOuAdjointIntroInput]$.
Reverse-mode differentiation does not form the full derivative operator.
Instead, given the upstream adjoint $\overline{\effOuAdjointIntroOutput}$, it applies the adjoint of the linearized map under the corresponding Euclidean or Frobenius inner products

\vspace{-1.1cm}

\begin{align}
\left\langle
\overline{\effOuAdjointIntroOutput},
\mathrm D f(\effOuAdjointIntroInput)[\dd\effOuAdjointIntroInput]
\right\rangle
&=
\left\langle
[\mathrm{D} f(\effOuAdjointIntroInput)]^{*}
[\overline{\effOuAdjointIntroOutput}],
\dd\effOuAdjointIntroInput
\right\rangle.
\end{align}
The quantity $[\mathrm{D} f(\effOuAdjointIntroInput)]^{*}[\overline{\effOuAdjointIntroOutput}]$ is the pulled-back contribution to the adjoint of $\effOuAdjointIntroInput$.
We refer to the upstream adjoint $\overline{\effOuAdjointIntroOutput}$ supplied to a reverse-mode primitive as its \textit{adjoint seed}.
Because the objective is scalar, one reverse sweep supplies derivatives with respect to all input parameters, rather than propagating a separate forward sensitivity for each parameter.

In the OU model, the two drift-dependent numerical primitives are the matrix exponential and the Lyapunov equation.
Their reverse-mode adjoints are a major bottleneck in generic automatic-differentiation implementations \citep{baydin_automatic_2018}, because they are cubic operations with heavy constant factors \citep{najfeld_derivatives_1995,golub_matrix_2013}.
Under the H-SSBP, these reverse-mode adjoints decompose after a change of basis into cheap and parallelizable independent constant-size block or block-pair calculations.

\subsection{Gaussian message passing and predictive moments}
\label{effOu.sec:lik_rep}
For fixed OU dimension $\effOuDim$, the marginal likelihood is evaluated in time linear in the number of nodes by Gaussian upward message passing
\citep{lauritzen_local_1988,lauritzen_graphical_1996,sarkka_bayesian_2013,ho_lineartime_2014,bastide_efficient_2021}.
Let $\effOudataVector=(\effOudataBelow{\effOuchildNode},\effOudataAbove{\effOuchildNode})$ partition the observations into those within and outside the subtree rooted at $\effOuchildNode$, respectively. 
Denote by $f_{\effOuEdgeLength{\effOuchildNode}}(\effOustate_{\effOuchildNode} \mid \effOustate_{\effOuparentNode})$ the transition kernel along the edge $\effOuparentNode \to \effOuchildNode$ defined in Equation~\eqref{effOu.eq:transition}.
Let $\effOupostOrder{\effOuchildNode}(\effOustate)$ be the conditional likelihood of the observations in the subtree rooted at $\effOuchildNode$ given $\effOuDataRv{\effOuchildNode}=\effOustate$; it combines the observation factor $\effOuObsFactor{\effOuchildNode}(\effOudata{\effOuchildNode}\mid\effOustate)$ at $\effOuchildNode$ with the messages arriving from its children:
\begin{align}
\effOupostOrder{\effOuchildNode}(\effOustate) = \condprob{\effOudataBelow{\effOuchildNode}}{\effOuDataRv{\effOuchildNode} = \effOustate}
&=
\effOuObsFactor{\effOuchildNode}(\effOudata{\effOuchildNode}\mid\effOustate)
\prod_{\effOuchildNodeOne \in \text{Child}[\effOuchildNode]}
\underbrace{
\Biggl[
\int_{\mathbb{R}^{\effOuDim}}
f_{\effOuEdgeLength{\effOuchildNodeOne}}(\effOustate' \mid \effOustate)
\effOupostOrder{\effOuchildNodeOne}(\effOustate')
\, \dd \effOustate'
\Biggr].
}_{\text{Message (up) } \effOuchildNodeOne \to \effOuchildNode }
\end{align}
Each message in the product accumulates evidence from one child subtree; contributions from distinct children multiply by conditional independence.

The upward pass accounts only for observations within each subtree.
To combine this information with observations outside the subtree rooted at a non-root node $\effOuchildNode$, we need to perform a second, downward pass.
Denote the factor carrying information to the state at $\effOuparentNode$ from the portion of the tree above it as
\begin{align}
\effOuParentSide{\effOuparentNode}(\effOustate)
&=
\begin{cases}
\effOurootDistribution(\effOustate),
& \effOuparentNode=\effRoot,
\\[4pt]
\displaystyle
\smash[b]{
\underbrace{
\int_{\mathbb{R}^{\effOuDim}}
\effOupreOrder{\effOuparentNode}(\effOustate')\,
f_{\effOuEdgeLength{\effOuparentNode}}
(\effOustate\mid\effOustate')\,
\dd \effOustate'
}_{\text{Message (down) to parent} 
}
},
& \effOuparentNode\neq\effRoot,
\end{cases}
\vphantom{
\underbrace{
\int_{\mathbb{R}^{\effOuDim}}
\effOupreOrder{\effOuparentNode}(\effOustate')\,
f_{\effOuEdgeLength{\effOuparentNode}}
(\effOustate\mid\effOustate')\,
\dd \effOustate'
}_{\text{Message (down) to parent} 
}
}
\end{align}

\vspace{0.6cm}

\noindent where $\effOurootDistribution$ is the root distribution. 
The downward-message function $\effOupreOrder{\effOuchildNode}(\effOustate) =
\prob{\effOudataAbove{\effOuchildNode},\,\effOuDataRv{\effOuparentNode} = \effOustate}$ aggregates information from the parent side of the edge ending at node $\effOuchildNode$ with its observation factor and the upward messages from the siblings of $\effOuchildNode$:
\begin{align}
\effOupreOrder{\effOuchildNode}(\effOustate)
&=
\effOuObsFactor{\effOuparentNode}
(\effOudata{\effOuparentNode}\mid\effOustate)\,\, 
\effOuParentSide{\effOuparentNode}(\effOustate)
\prod_{\effOuchildNodeOne\in
\text{Child}[\effOuparentNode]\setminus\{\effOuchildNode\}}
\underbrace{
\left[
\int_{\mathbb{R}^{\effOuDim}}
f_{\effOuEdgeLength{\effOuchildNodeOne}}
(\effOustate'\mid\effOustate)
\effOupostOrder{\effOuchildNodeOne}(\effOustate')
\,\dd\effOustate'
\right]
}_{\text{Message (up) }\effOuchildNodeOne\to\effOuparentNode}.
\label{effOu.eq:preorder_pruning}
\end{align}
For any edge ending at node $\effOuchildNode$, the marginal likelihood can be recovered by the edgewise contraction between the upward and downward messages
\begin{align}
   \effOumyFunc 
&=
\int_{\mathbb{R}^{\effOuDim}}
\effOupostOrder{\effOuchildNode}
\left(
\effOustate_{\effOuchildNode}
\right)
\int_{\mathbb{R}^{\effOuDim}}
f_{\effOuEdgeLength{\effOuchildNode}}
\left(
\effOustate_{\effOuchildNode}
\mid
\effOustate_{\effOuparentNode}
\right)
\effOupreOrder{\effOuchildNode}
\left(
\effOustate_{\effOuparentNode}
\right)
\,\dd\effOustate_{\effOuparentNode}
\,\dd\effOustate_{\effOuchildNode}.
\label{effOu.eq:suppedgeLikelihoodIdentity}
\end{align}
Now consider an edge-specific parameter $\effOuedgeParameter{\effOuchildNode}$ that affects only the transition density on the edge $\effOuparentNode\to\effOuchildNode$.
Under the usual regularity conditions allowing differentiation under the integral sign, differentiating the log-likelihood $\effLogLik = \log \effOumyFunc$ with respect to \(\effOuedgeParameter{\effOuchildNode}\) gives
\begin{align}
\myPartial{\effLogLik}{\effOuedgeParameter{\effOuchildNode}}
&=
\frac{1}{\effOumyFunc}
\int_{\mathbb{R}^{\effOuDim}}
\effOupostOrder{\effOuchildNode}
\left(
\effOustate_{\effOuchildNode}
\right)
\int_{\mathbb{R}^{\effOuDim}}
\myPartial{}{\effOuedgeParameter{\effOuchildNode}}
\Big[
f_{\effOuEdgeLength{\effOuchildNode}}
\left(
\effOustate_{\effOuchildNode}
\mid
\effOustate_{\effOuparentNode}
\right)
\Big]
\effOupreOrder{\effOuchildNode}
\left(
\effOustate_{\effOuparentNode}
\right)
\,\dd\effOustate_{\effOuparentNode}
\,\dd\effOustate_{\effOuchildNode}.
\label{effOu.eq:suppedgeLikelihoodDerivative}
\end{align}
Because the likelihood is linear in each child-to-parent message, the corresponding upstream adjoint of the log-likelihood is the associated downward-message function, divided by $\effOumyFunc$.
Normalizing $\effOupreOrder{\effOuIdx}$ over the parent state gives the edge-specific parent-side distribution
\begin{align}
\effOuDataRv{\effOuParentNodeOperator{\effOuIdx}}
\mid
\effOudataAbove{\effOuIdx}
\sim
\effOuNorm
\left(
\effOuParMean{\effOuIdx},
\effOuParVar{\effOuIdx}
\right).
\end{align}
Propagating this distribution through the OU transition yields the formulas for the child-side predictive moments
\begin{align}
\effOuAlgMean{\effOuIdx}
&=
\effOuActual{\effOuIdx}
\effOuParMean{\effOuIdx}
+
\left(
\effOuIdentity-\effOuActual{\effOuIdx}
\right)
\effOuEquilMeanEdge{\effOuIdx},
\label{effOu.eq:preorder_mean}
\\
\effOuAlgPrecision{\effOuIdx}^{-1}
&=
\effOuStatVarEdge{\effOuIdx}
+
\effOuActual{\effOuIdx}
\left(
\effOuParVar{\effOuIdx}
-
\effOuStatVarEdge{\effOuIdx}
\right)
\effOuActual{\effOuIdx}^{\effOuTranspose}.
\label{effOu.eq:preorder_prop}
\end{align}
These moments provide the finite-dimensional Gaussian representation of the transition integral and the adjoint interface used below.

\subsection{Edge-local adjoint seeds for OU primitives}
\label{effOu.sec:grad_framework}
For an edge ending at node $\effOuIdx$, the Gaussian message-passing reverse pass supplies the upstream adjoints with respect to the child-side predictive mean and covariance.
Denote these adjoints by
$\effOuGradVec{\effOuIdx}
=
\partial\effLogLik/\partial\effOuAlgMean{\effOuIdx}$
and
$\effOuUpAdj{\effOuIdx}
=
\partial\effLogLik/\partial(\effOuAlgPrecision{\effOuIdx}^{-1})$,
respectively.
For an edge-local parameter $\effOuedgeParameter{\effOuIdx}$ entering the OU transition, the chain rule gives
\begin{align}
\myPartial{\effLogLik}{\effOuedgeParameter{\effOuIdx}}
&=
\effOuGradVec{\effOuIdx}^{\effOuTranspose}
\myPartial{\effOuAlgMean{\effOuIdx}}
{\effOuedgeParameter{\effOuIdx}}
+
\left\langle
\effOuUpAdj{\effOuIdx},
\myPartial{}
{\effOuedgeParameter{\effOuIdx}}
\left(
\effOuAlgPrecision{\effOuIdx}^{-1}
\right)
\right\rangle.
\label{effOu.eq:chain_rule_direct}
\end{align}
Within the edge-local pullback, the parent-side moments $\effOuParMean{\effOuIdx}$ and $\effOuParVar{\effOuIdx}$
are held fixed, with their parameter dependence handled by the surrounding Gaussian message-passing reverse pass.
The drift enters Equations~\eqref{effOu.eq:preorder_mean}--\eqref{effOu.eq:preorder_prop} through the actualization matrix
$\effOuActual{\effOuIdx}
=
\effOuExpMatrix{
\effOuDriftMatrixEdge{\effOuIdx}
\effOuEdgeLength{\effOuIdx}
}$
and through the stationary covariance $\effOuStatVarEdge{\effOuIdx}$, which is defined implicitly by the Lyapunov equation.
Holding the parent-side moments fixed within the edge-local differential, the predictive mean and covariance satisfy
\begin{align}
\dd\effOuAlgMean{\effOuIdx}
&=
\left(
\dd\effOuActual{\effOuIdx}
\right)
\left(
\effOuParMean{\effOuIdx}
-
\effOuEquilMeanEdge{\effOuIdx}
\right),
\label{effOu.eq:mean_differential_main}
\\
\dd\left(
\effOuAlgPrecision{\effOuIdx}^{-1}
\right)
&=
\left(
\dd\effOuActual{\effOuIdx}
\right)
\left(
\effOuParVar{\effOuIdx}
-
\effOuStatVarEdge{\effOuIdx}
\right)
\effOuActual{\effOuIdx}^{\effOuTranspose}
\notag\\
&\qquad\qquad+
\effOuActual{\effOuIdx}
\left(
\effOuParVar{\effOuIdx}
-
\effOuStatVarEdge{\effOuIdx}
\right)
\left(
\dd\effOuActual{\effOuIdx}
\right)^{\effOuTranspose}
+
\dd\effOuStatVarEdge{\effOuIdx}
-
\effOuActual{\effOuIdx}
\left(
\dd\effOuStatVarEdge{\effOuIdx}
\right)
\effOuActual{\effOuIdx}^{\effOuTranspose}.
\label{effOu.eq:covariance_differential_main}
\end{align}
Since $\effOuUpAdj{\effOuIdx}$ and
$(\effOuParVar{\effOuIdx}-\effOuStatVarEdge{\effOuIdx})$
are symmetric, contracting Equations~\eqref{effOu.eq:mean_differential_main}--\eqref{effOu.eq:covariance_differential_main} with the upstream adjoints gives the actualization and stationary-covariance seeds,
\begin{align}
\mathbf{G}_{A,\effOuIdx}
&=
\effOuGradVec{\effOuIdx}
\left(
\effOuParMean{\effOuIdx}
-
\effOuEquilMeanEdge{\effOuIdx}
\right)^{\effOuTranspose}
+
2 \,
\effOuUpAdj{\effOuIdx}
\effOuActual{\effOuIdx}
\left(
\effOuParVar{\effOuIdx}
-
\effOuStatVarEdge{\effOuIdx}
\right),
\label{effOu.eq:actualization_seed_main}
\\
\effOuGradVIdx{\effOuIdx}
&=
\effOuUpAdj{\effOuIdx}
-
\effOuActual{\effOuIdx}^{\effOuTranspose}
\effOuUpAdj{\effOuIdx}
\effOuActual{\effOuIdx}.
\label{effOu.eq:edge_seeds_main}
\end{align}
These edge-local seeds provide the inputs to the matrix-exponential and Lyapunov pullbacks developed in the next subsection.

\subsection{Structured reverse-mode primitives}
\label{effOu.sec:grad_frechet_and_lyapunov}
The same block structure used for the forward quantities in Section~\ref{effOu.sec:smbp_computationalAdvantages} applies to reverse-mode differentiation through the drift parameters.
We start from the H-SSBP factorization and propagate the edge-local seeds from Equations~\eqref{effOu.eq:actualization_seed_main}--\eqref{effOu.eq:edge_seeds_main} through the matrix exponential and Lyapunov equation.

\paragraph{Fr\'echet adjoint of the block-diagonal exponential}
For edge $\effOuIdx$,
\begin{align}
\effOuActual{\effOuIdx}
=
\effOuRotMat
\effOuActualRot{\effOuIdx}
\effOuRotMat^{-1},
\qquad
\effOuActualRot{\effOuIdx}
=
\effOuExpMatrix{\effOuBlkDiag\effOuEdgeLength{\effOuIdx}}.
\end{align}
Holding $\effOuRotMat$ fixed, the relevant primal perturbation is
$\dd\effOuActual{\effOuIdx}=\effOuRotMat(\dd\effOuActualRot{\effOuIdx})\effOuRotMat^{-1}$,
so the Frobenius-dual actualization seed is transformed into block coordinates by
\begin{align}
\effOuAdjointVarRotated_{\effOuIdx}
=
\effOuRotMat^{\effOuTranspose}
\mathbf G_{A,\effOuIdx}
\effOuRotMat^{-\effOuTranspose}.
\label{effOu.eq:rotated_actualization_seed_main}
\end{align}
Define the edge-length-specific Fr\'echet operator
\begin{align}\label{effOu.eq:frechet_forward}
\effOuFrechetOp{\effOuBlkDiag}^{(\effOuEdgeLength{\effOuIdx})}
[\effOuFrechetDirection]
&=
\left.
\frac{\dd}{\dd\epsilon}
\effOuExpMatrix{\effOuEdgeLength{\effOuIdx}(\effOuBlkDiag+\epsilon\effOuFrechetDirection)}
\right|_{\epsilon=0}.
\end{align}
The integral representation of its adjoint under the Frobenius inner product is \citep{najfeld_derivatives_1995,monti_nonparametric_2025}
\begin{align}\label{effOu.eq:rotatedIntegralFrechetAdjoint}
\left(\effOuFrechetOpAdj{\effOuBlkDiag}\right)^{(\effOuEdgeLength{\effOuIdx})}
[\effOuAdjointVarRotated_{\effOuIdx}]
&=
\effOuEdgeLength{\effOuIdx}
\int_0^1
\effOuExpMatrix{s\effOuEdgeLength{\effOuIdx}\effOuBlkDiag^{\effOuTranspose}}
\effOuAdjointVarRotated_{\effOuIdx}
\effOuExpMatrix{(1-s)\effOuEdgeLength{\effOuIdx}\effOuBlkDiag^{\effOuTranspose}}
\,\dd s.
\end{align}
When $\effOuBlkDiag$ is shared across edges, these edge-specific Fr\'echet pullbacks are summed before continuing the reverse pass.
Because $\effOuBlkDiag$ is block diagonal, every output block $(\effOuIdxB,\effOuIdxBTwo)$ depends only on $\effOuAdjointVarRotated_{\effOuIdx,\effOuIdxB\effOuIdxBTwo}$ and is obtained from a constant-size kernel \citep{monti_nonparametric_2025}:

\vspace{-0.4in}

\begin{align}\label{effOu.eq:block_frechet_adjoint_pair}
  \left[
    \left(\effOuFrechetOpAdj{\effOuBlkDiag}\right)^{(\effOuEdgeLength{\effOuIdx})}
    [\effOuAdjointVarRotated_{\effOuIdx}]
  \right]_{\effOuIdxB\effOuIdxBTwo}
  &=
  \effOuEdgeLength{\effOuIdx}
  \int_0^1
  e^{s\effOuEdgeLength{\effOuIdx}\effOuBlk{\effOuIdxB}^{\effOuTranspose}}\,
  \effOuAdjointVarRotated_{\effOuIdx,\effOuIdxB\effOuIdxBTwo}\,
  e^{(1-s)\effOuEdgeLength{\effOuIdx}\effOuBlk{\effOuIdxBTwo}^{\effOuTranspose}}\,\dd s,
\end{align}
each involving matrices of size at most $2\times2$.
Since admissible perturbations of $\effOuBlkDiag$ are themselves block diagonal, only the $\mathcal{O}(\effOuDim)$ diagonal kernels are required for the native block gradient.
The off-diagonal output blocks of the full Fr\'echet adjoint would correspond to unrestricted perturbations outside the block-diagonal parameter space and are \emph{not} necessary.
The change-of-basis gradient is instead obtained directly from the standard similarity-transform pullback applied to the transformed actualization seed.
The more general block-pair formulas and their adjoint specialization, including stable computation near repeated eigenvalues, are given in \supprefsrange{suppeffOu.forward_frechet}{suppeffOu.adjoint_exp}.

\paragraph{Adjoint Lyapunov solve}
\label{effOu.sec:grad_lyap_adj}
The stationary covariance enters the likelihood through the rotated Lyapunov solution $\effOuStatVarRot$, defined by
\begin{align}
\effOuBlkDiag\effOuStatVarRot
+
\effOuStatVarRot\effOuBlkDiag^{\effOuTranspose}
=
-\effOuDiffMatRot.
\end{align}
Let $\effOuLyapSeed$ denote the total seed with respect to $\effOuStatVar$ in the original coordinates.
When $\effOuStatVar$ is shared across edges, this seed includes the sum of the corresponding edge-local seeds $\effOuGradVIdx{\effOuIdx}$, together with any direct contribution from a stationary root distribution or observation model.
After changing basis, the corresponding seed is $\effOuLyapSeedRot=\effOuRotMat^{\effOuTranspose}\effOuLyapSeed\effOuRotMat$.
Implicit differentiation of the Lyapunov equation shows that the reverse-mode adjoint is obtained by solving the transposed Lyapunov equation

\vspace{-0.5in}

\begin{align}
\effOuBlkDiag^{\effOuTranspose} \effOuLyapAdjRot
+
\effOuLyapAdjRot\,\effOuBlkDiag
=
-\effOuLyapSeedRot.
\label{effOu.eq:lyap_adjoint_D_basis_main}
\end{align}
Because $\effOuBlkDiag$ is block diagonal, this solve decomposes into independent block-pair Sylvester equations,

\vspace{-0.5in}

\begin{align}
\effOuBlk{\effOuIdxB}^{\effOuTranspose}
\effOuLyapAdjRot_{\effOuIdxB\effOuIdxBTwo}
+
\effOuLyapAdjRot_{\effOuIdxB\effOuIdxBTwo}
\effOuBlk{\effOuIdxBTwo}
=
-\effOuLyapSeedRotBlock.
\end{align}
Thus the reverse Lyapunov step uses the same constant-size block-pair kernels as the forward stationary-covariance computation, rather than a global Bartels--Stewart back-substitution.
The corresponding supplementary derivation is given in \suppref{suppeffOu.adjoint_lyap}.
Once $\effOuLyapAdjRot$ has been computed, the contribution of the stationary-covariance path to the gradient with respect to the admissible block drift is obtained by retaining the diagonal blocks:

\vspace{-0.4in}

\begin{align}
\left[
\frac{\partial \effLogLik_{\mathrm{Lyap}}}{\partial \effOuBlkDiag}
\right]_{\effOuIdxB\effOuIdxB}
&=
\left[
\effOuLyapAdjRot\,\effOuStatVarRot
+
\effOuLyapAdjRot^{\effOuTranspose}\effOuStatVarRot
\right]_{\effOuIdxB\effOuIdxB}.
\label{effOu.eq:lyap_grad_D_main}
\end{align}
The projection onto diagonal blocks is essential because off-diagonal blocks would correspond to perturbations outside the block-diagonal parameter space of $\effOuBlkDiag$.

After adding the Lyapunov drift contribution above to the matrix-exponential contribution, the diagonal blocks are contracted with the derivatives of the H-SSBP block map to obtain gradients with respect to
$(\effOuCartBlockRate_{\effOuIdxB},\effOuCartNormalizingRate_{\effOuIdxB},\effOuBlkT{\effOuIdxB})$.
Separately, the same adjoint variable also gives the seed with respect to the rotated diffusion matrix.
Gradients with respect to the diffusion and change-of-basis parameters follow from standard congruence and similarity pullbacks.
No additional Lyapunov solve is required.

\subsection{Identifiability and unconstrained parametrization} \label{effOu.sec:identifiability_and_unconstraint_parametrization}
We now address the identifiability challenges in inferring the drift matrix and underlying parameters, and present the unconstrained parametrizations used in computation.
To keep notation light, in the rest of this section we omit the edge indices $\effOuIdx$.

\paragraph{Identifiability and gauge conventions}
Spectral parametrizations of the form $\effOuDriftMatrix=\effOuRotMat\effOuBlkDiag\effOuRotMat^{-1}$ are generically nonidentifiable, regardless of the specific model that produces them: block permutations, sign changes, and admissible changes of basis within invariant subspaces can leave $\effOuDriftMatrix$ unchanged, and repeated blocks introduce additional freedom \citep{anderson_introduction_2003,horn_matrix_2012,jolliffe_principal_2016}.
The H-SSBP coordinates $(\effOuRotMat,\effOuBlkDiag)$ inherit this issue, and may therefore remain nonidentifiable even when the reconstructed drift $\effOuDriftMatrix$ itself is identifiable.
We therefore treat $(\effOuRotMat,\effOuBlkDiag)$ as computational coordinates and report inference primarily through $\effOuDriftMatrix$.

\paragraph{Unconstrained parametrization and improved identifiability}
To improve gradient-based routines, we transform constrained parameters to unconstrained Euclidean coordinates using differentiable maps and reduce specific redundancies through ordering and sign conventions as described below.
Starting with the drift block parameters $(\effOuCartBlockRate_{\effOuIdxB},\effOuCartNormalizingRate_{\effOuIdxB},\effOuBlkT{\effOuIdxB})$, the strict stability constraint $|\effOuCartNormalizingRate_{\effOuIdxB}|<1$ on the coupling parameters is guaranteed by setting
\begin{align}
\effOuCartNormalizingRate_{\effOuIdxB}
:=
\tanh\{\tilde{\effOuCartNormalizingRate}_{\effOuIdxB}\}, \qquad \tilde{\effOuCartNormalizingRate}_{\effOuIdxB}\in\mathbb{R}.
\end{align} 
For the contraction rates, we reduce cross-block nonidentifiability and enforce stability ($\effOuCartBlockRate_{\effOuIdxB}<0$) by ordering them $\effOuCartBlockRate_1<\effOuCartBlockRate_2<\cdots<\effOuCartBlockRate_{\effNBlock}<0$ using positive-spacing recursive transformations,

\vspace{-1cm}

\begin{align}
\effOuCartBlockRate_{\effNBlock}
&:=
-\exp\{\tilde{\effOuCartBlockRate}_{\effNBlock}\},
&
\effOuCartBlockRate_{\effOuIdxB}
&:=
\effOuCartBlockRate_{\effOuIdxB+1}
-
\exp\{\tilde{\effOuCartBlockRate}_{\effOuIdxB}\},
&
\effOuIdxB
&=
\effNBlock - 1,\ldots,1,
\end{align}
where $\tilde{\effOuCartBlockRate}_{\cdot} \in \mathbb{R}$.
To prevent block-sign ambiguity in inference, we impose the gauge convention $\effOuBlkT{\effOuIdxB}>0$, and use the unconstrained coordinate $\tilde{\effOuBlkT{}}_{\effOuIdxB}\in\mathbb{R}$ with $\effOuBlkT{\effOuIdxB}=\exp\{\tilde{\effOuBlkT{}}_{\effOuIdxB}\}$.
Alternative reparametrizations are discussed in \suppref{suppeffOu.sec:alternativeParametrizationWithinHssbpFamily}.

For the generic H-SSBP, we parametrize the change-of-basis matrix $\effOuRotMat$ directly by its raw entries and regularize using the penalty 
\begin{align}
  \lambda \, \log\left|\det\effOuRotMat\right| -  \frac{\lambda \effOuDim}{2}\|\effOuRotMat\|_F^2, \quad \lambda>0,
\end{align}
where $\effOuDim$ is the process dimension and $\|\cdot\|_F$ denotes the Frobenius norm. 
This regularization improves conditioning while breaking the global scale invariance of the similarity representation.
For the orthogonal specialization, $\effOuRotMat$ is instead parametrized as a fixed ordered product of Givens rotations, each governed by a single unconstrained angle; discontinuous sign conventions are applied only to final estimates, keeping the quantities driving optimization or sampling continuous throughout.
Finally, the components of the equilibrium mean $\effOuEquilMean$ are already unconstrained, while positive-definite covariance matrices are represented in unconstrained real coordinates via standard SPD parametrizations, such as log-Cholesky factors.

\section{Speed-ups Under Homogeneity} \label{effOu.sec:speedupshomogeneity}

We now consider a homogeneous OU model where every edge shares $\effOuDriftMatrix$, $\effOuDiffMat$, and $\effOuEquilMean$, while $\effOuEdgeLength{\effOuIdx}$ may vary across edges.
In Section~\ref{effOu.sec:homogeneous_block_basis} we show that, after transforming the observation vectors and parameters to the block basis, likelihood and gradient evaluations do not require any further change-of-basis multiplication.
Furthermore, in Section~\ref{effOu.sec:homogeneous_block_basis_simultaneous} we show that the transformed OU process reduces to one- and two-dimensional independent modules if $\effOuDriftMatrix$ and $\effOuDiffMat$ also share a common orthogonal basis.
Table~\ref{effOu.table:homogeneous_block_basis_costs} summarizes the computational complexity results.

\begin{table}[t]
\centering
\small
\setlength{\tabcolsep}{5pt}
\renewcommand{\arraystretch}{1.35}
\begin{tabular}{@{}lccc@{}}
\toprule
\textbf{Operation}
&
\textbf{Standard}
&
\textbf{Change of basis}
&
\textbf{Shared basis}
\\
\midrule
Change of basis for parameters
&
--
&
\(\mathcal{O}(\effOuDim^{3})\)
&
\(\mathcal{O}(\effOuDim^{3})\)
\\
Observation-factor change of basis
&
--
&
\(\mathcal{O}(\effOuNumberOfObs\effOuDim^{2})\)
&
\(\mathcal{O}(\effOuNumberOfObs\effOuDim^{2})\)
\\
Stationary covariance
&
\(\mathcal{O}(\effOuDim^{3})\)
&
\(\mathcal{O}(\effOuDim^{2})\)
&
\(\mathcal{O}(\effOuDim)\)
\\
Edge actualization
&
\(\mathcal{O}(\effOuNumberOfEdges\effOuDim^{3})\)
&
\(\mathcal{O}(\effOuNumberOfEdges\effOuDim)\)
&
\(\mathcal{O}(\effOuNumberOfEdges\effOuDim)\)
\\
Innovation covariance
&
\(\mathcal{O}(\effOuNumberOfEdges\effOuDim^{3})\)
&
\(\mathcal{O}(\effOuNumberOfEdges\effOuDim^{2})\)
&
\(\mathcal{O}(\effOuNumberOfEdges\effOuDim)\)
\\
Actualization adjoint
&
\(\mathcal{O}(\effOuNumberOfEdges\effOuDim^{3})\)
&
\(\mathcal{O}(\effOuNumberOfEdges\effOuDim^{2})\)
&
\(\mathcal{O}(\effOuNumberOfEdges\effOuDim)\)
\\
Innovation-covariance adjoint
&
\(\mathcal{O}(\effOuNumberOfEdges\effOuDim^{3})\)
&
\(\mathcal{O}(\effOuNumberOfEdges\effOuDim^{2})\)
&
\(\mathcal{O}(\effOuNumberOfEdges\effOuDim)\)
\\
Node-level conditioning / message combination
&
\(\mathcal{O}(\effOuNumberOfEdges \effOuDim^{3})\)
&
\(\mathcal{O}(\effOuNumberOfEdges \effOuDim^{3})\)
&
\(\mathcal{O}(\effOuNumberOfEdges \effOuDim)\)
\\
\bottomrule
\noalign{\vskip -3pt}
\multicolumn{4}{@{}r@{}}{
\scriptsize
\(\effOuNumberOfEdges = \lvert\effOuEdgeSet\rvert\) is the number of graph edges, \(\effOuNumberOfObs\) is the number of observed nodes, and \(\effOuDim\) is the state dimension.
}
\end{tabular}
\caption{
Leading cost comparisons for homogeneous Gaussian message passing.
The standard column works in the original coordinates.
The change-of-basis column uses the drift basis, with dense transformed diffusion and stationary covariance.
The shared-basis column assumes that drift and diffusion share the same orthogonal basis, so the OU transition decomposes into $\left\lceil \effOuDim/2 \right\rceil$ independent one- or two-dimensional block transitions in the rotated coordinates.
The counts assume distinct edge lengths. The observation model is linear--Gaussian (Section~\ref{effOu.sec:model}), possibly partial or noisy.
For the shared-basis column, the observation operators and observation-error covariances are additionally assumed to be (block) diagonal in the same rotated basis.
} \label{effOu.table:homogeneous_block_basis_costs}
\end{table}

\subsection{Homogeneous block-basis implementation}
\label{effOu.sec:homogeneous_block_basis}

In a homogeneous OU model, the same drift factorization
$\effOuDriftMatrix
=
\effOuRotMat\effOuBlkDiag\effOuRotMat^{-1}$
is shared by every edge of the graph.
This makes the block basis useful beyond a single matrix-function evaluation: the Gaussian message-passing recursions can be expressed in one common coordinate system, so that repeated edge-level OU operations do not require forming dense transition matrices in the original basis.
Define the centered block-basis state by
\begin{align}
\effOuDataRvRotated{\effOuIdx}
=
\effOuRotMat^{-1}
\left(
\effOuDataRv{\effOuIdx}
-
\effOuEquilMean
\right),
\qquad
\effOuDataRv{\effOuIdx}
=
\effOuEquilMean
+
\effOuRotMat
\effOuDataRvRotated{\effOuIdx}.
\end{align}
We denote by
$\effOuDiffMatRot
=
\effOuRotMat^{-1}
\effOuDiffMat
\effOuRotMat^{-\effOuTranspose}$
and 
$\effOuStatVarRot
=
\effOuRotMat^{-1}
\effOuStatVar
\effOuRotMat^{-\effOuTranspose}$ the diffusion and stationary covariance in this basis.
For an edge of length \(\effOuEdgeLength{\effOuIdx}\), the transition becomes
\begin{align}
\effOuDataRvRotated{\effOuIdx}
\mid
\effOuDataRvRotated{\effOuParentNodeOperator{\effOuIdx}}
\sim
\effOuNorm
\left(
\effOuActualRot{\effOuIdx}
\effOuDataRvRotated{\effOuParentNodeOperator{\effOuIdx}},
\effOuEdgeCovRot{\effOuIdx}
\right),
\end{align}
where
$\effOuActualRot{\effOuIdx}
=
\effOuExpMatrix{\effOuBlkDiag\effOuEdgeLength{\effOuIdx}}$,
and 
$\effOuEdgeCovRot{\effOuIdx}
=
\effOuStatVarRot
-
\effOuActualRot{\effOuIdx}
\effOuStatVarRot
\effOuActualRot{\effOuIdx}^{\effOuTranspose}$.
Since \(\effOuBlkDiag\) is block diagonal, \(\effOuActualRot{\effOuIdx}\) is block diagonal for every edge.
Thus products of the form
\(\effOuActualRot{\effOuIdx}\effOuStatVarRot\effOuActualRot{\effOuIdx}^{\effOuTranspose}\) cost only \(\mathcal{O}(\effOuDim^{2})\) when \(\effOuStatVarRot\) is dense, rather than \(\mathcal{O}(\effOuDim^{3})\).
The same reduction applies in reverse mode, where the adjoints associated with edge actualization and innovation covariance are propagated by block-diagonal left and right multiplications before being accumulated across edges.

The gain from homogeneity is graph-wide.
A natural-basis implementation repeatedly forms dense edge transitions and dense innovation covariances, whereas the block-basis implementation keeps propagation in the shared transformed coordinates.
Dense dependence is not removed in general, since it remains in \(\effOuStatVarRot\), in Gaussian message covariances, and in local evidence factors.

\subsection{Fully decoupled shared-basis case}
\label{effOu.sec:homogeneous_block_basis_simultaneous}

A particularly simple special case arises when the drift and diffusion share a common orthogonal basis, that is,

\vspace{-0.5in}

\begin{align} \label{effOu.eq:sharedOrthogonalBasis}
\effOuDriftMatrix
=
\effOuRotMat\effOuBlkDiag\effOuRotMat^{\effOuTranspose},
\qquad
\effOuDiffMat
=
\effOuRotMat\mathbf{D}\effOuRotMat^{\effOuTranspose},
\qquad
\effOuRotMat^{\effOuTranspose}\effOuRotMat
=
\effOuIdentity,
\end{align}
$\mathbf{D}=\operatorname{diag}(d_{1},\ldots,d_{\effOuDim})$ with $d_{j}> 0$.
This framework improves both computational efficiency and interpretability.
Computationally, the block-basis stationary covariance is obtained from
$\effOuBlkDiag\effOuStatVarRot
+
\effOuStatVarRot\effOuBlkDiag^{\effOuTranspose}
+
\mathbf{D}
=
0$.
Since $\mathbf{D}$ is diagonal and $\effOuBlkDiag$ is block diagonal, $\effOuStatVarRot$ is also block diagonal.
Consequently, the edge innovation covariance and the actualization matrices are also block diagonal for every interval.

Under the shared-basis assumption, the OU process decomposes in the rotated coordinates into independent one- and two-dimensional blocks.
Actualization matrices, stationary and innovation covariances, and their adjoints are therefore computed blockwise, so interval-level OU operations scale linearly in $\effOuDim$, apart from the initial transformations through $\effOuRotMat$.
The columns of $\effOuRotMat$ define orthogonal combinations of the original coordinates, while $\mathbf{D}$ and $\effOuBlkDiag$ govern, respectively, stochastic variation and mean-reversion dynamics within each block.
Equation~\eqref{effOu.eq:sharedOrthogonalBasis} can thus be viewed as a dynamic analogue of common principal components \citep{flury_algorithm_1986,flury_common_1988,phillips_hierarchical_1999}, in which the same latent blocks govern both deterministic relaxation and diffusion-driven variation.

\section{Computational and Numerical Performance}\label{effOu.sec:results}
%
%
We benchmark the computational primitives that determine the state-dimensional cost of OU likelihood and gradient evaluation.
We compare the H-SSBP with Schur- and Lyapunov-based parametrizations in terms of computational efficiency (below) and numerical robustness (\suppref{suppeffOu.sec:numerical_robustness}).
The experiments isolate the drift-dependent computations for which the H-SSBP changes the dominant constant factors: matrix exponentials and their adjoints, Lyapunov solves, and edge-local messages and gradients.
All benchmarks are implemented in \texttt{Julia}~1.10.3 using optimized linear algebra and matrix-function routines, and timings were collected with ILP64 OpenBLAS/LAPACK on a MacBook Pro with an Apple M1 using one BLAS thread.


The efficiency benchmarks isolate the drift-dependent kernels that are repeated across edges in likelihood and gradient evaluation.
We compare three implementations: H-SSBP block-basis kernels, known-Schur kernels in a cached quasi-triangular basis, and dense kernels.
The dense kernels represent the Lyapunov-based parametrization, assuming that the dense drift has already been formed from its primitive parameters; hence the benchmarks focus on the downstream matrix-function costs rather than on parameter construction.
For the Schur and H-SSBP implementations, basis changes are included only when they scale with the number of edge or observation contributions in a homogeneous setting.
In particular, exponential adjoints include the change of basis for the upstream seed, but exclude the final pullback applied after edgewise adjoints have been accumulated.

Figure~\ref{fig:ou_exp_adjoint_cached_benchmarks} reports median per-call timings for dimensions $\effOuDim\in\{4,8,16,32,64\}$.
The block exponential is already tens of times faster than the dense and known-Schur alternatives, and the advantage is even larger for the adjoint Fr\'echet exponential, where H-SSBP avoids dense matrix-function differentiation.
The Lyapunov panel shows a smaller but still substantial gain because both H-SSBP and Schur implementations require basis changes, whereas the H-SSBP solve itself decomposes into block-pair operations.
When these kernels are embedded in edge-message construction, dense Gaussian algebra and covariance multiplications common to all methods reduce the visible speed-up.
Nevertheless, H-SSBP remains substantially faster for both forward edge messages and reverse-mode covariance adjoints.

\begin{figure}[!ht]
\centering
\includegraphics[width=0.98\textwidth]{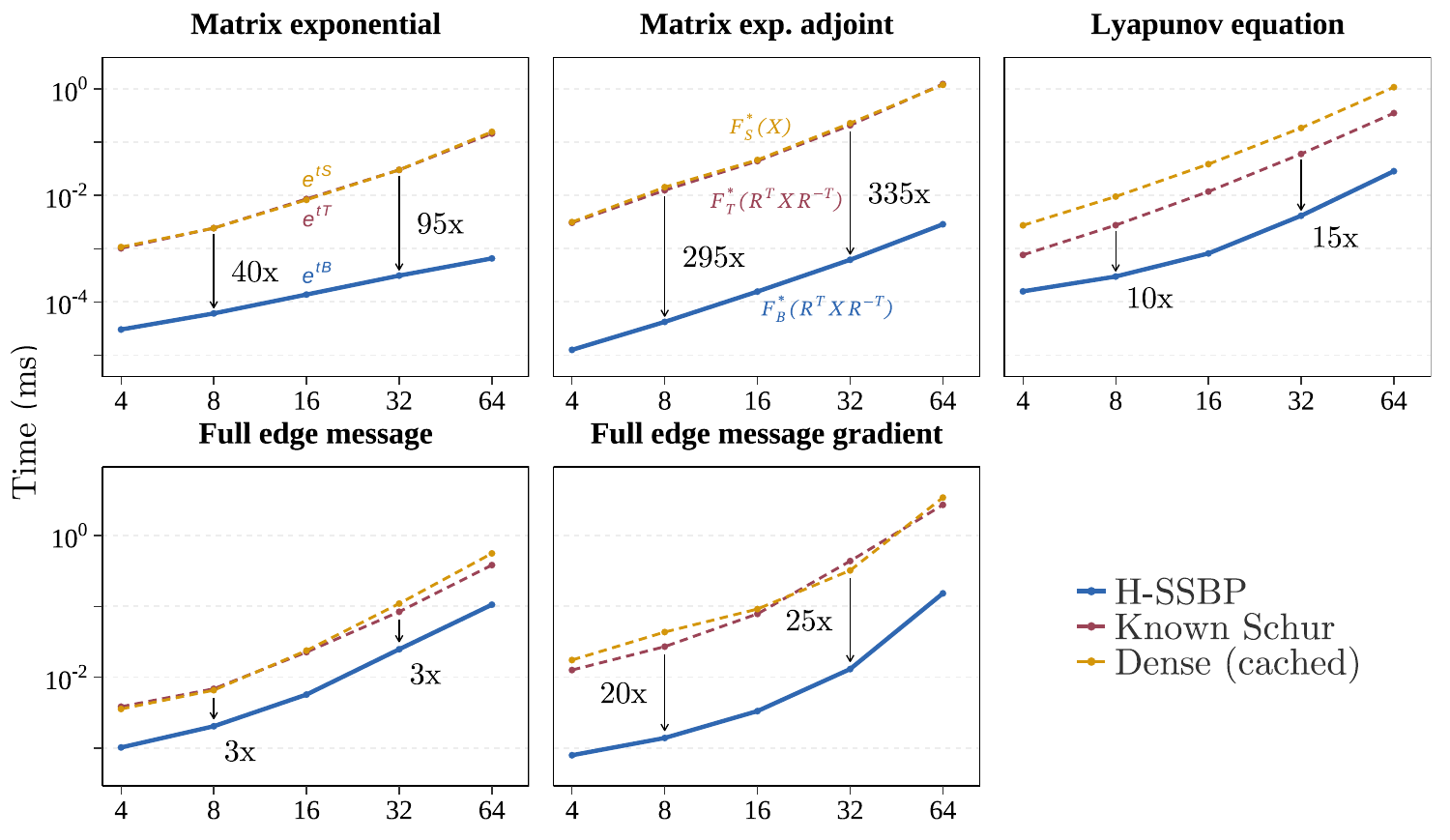}
\caption{
Computational benchmarks for drift-dependent OU kernels and edge-level message operations.
The top row compares the matrix exponential, the adjoint Fr\'echet exponential, and the Lyapunov solve under H-SSBP, known-Schur, and dense cached implementations.
The bottom row embeds these kernels in forward edge-message construction and in the corresponding reverse-mode covariance adjoint.
H-SSBP uses block-basis formulas; known-Schur uses cached quasi-triangular-basis operations; dense cached applies dense routines to a precomputed drift matrix.
These benchmarks are developed with a homogeneous OU process as the primary reference setting.
For a conservative comparison, however, we do \emph{not} assume that the observations and parameters have already been transformed into the block basis as detailed in Section~\ref{effOu.sec:homogeneous_block_basis}.
We include basis changes when they scale with the repeated edge-level computation, while final adjoint pullbacks performed once after edge accumulation are excluded.
Annotations report speed-ups of H-SSBP relative to the fastest competing implementation at the marked dimensions.
}
\label{fig:ou_exp_adjoint_cached_benchmarks}
\end{figure}

\section{Applications}
\label{effOu.sec:applications}

We implemented the OU message-passing algorithms and H-SSBP gradient calculations in the open-source software \textsc{Beast X} \citep{baele_beast_2025} and performed joint Bayesian inference either by maximizing the log posterior to obtain maximum a posteriori (MAP) estimates or by Hamiltonian Monte Carlo (HMC) \citep{neal_mcmc_2011} with adaptive step size.
Section~\ref{effOu.sec:kalmanSmoother} applies the proposed framework to simulated time series and an asynchronous financial data example, while Section~\ref{effOu.sec:phylogeneticExample} presents a phylogenetic comparative analysis problem.

\paragraph{Common priors}
When unknown, the diffusion covariance was inferred through a Cholesky correlation representation, with independent \(\mathrm{Gamma}(0.5,0.5)\) shape--scale priors on the diagonal diffusion parameters and an LKJ\((1)\) prior on the correlation matrix \citep{lewandowski_generating_2009}.
When inferred, the entries of the equilibrium mean \(\effOuEquilMean\) were assigned independent \(N(0,1)\) priors.
For the H-SSBP block parameters, the unconstrained coordinates corresponding to the scalar rate, block-rate increments \(\tilde{\effOuCartBlockRate}_{\effOuIdxB}\), and oscillatory parameters \(\tilde{\effOuBlkT{}}_{\effOuIdxB}\) were assigned independent \(N(0,1)\) priors.
The coupling parameters \(\effOuCartNormalizingRate_{\effOuIdxB}\) were assigned independent uniform priors on \((-1,1)\).
The orthogonal H-SSBP used independent \(N(0,0.25^2)\) priors on the Givens angles.
For the dense-basis H-SSBP, we used the regularizing prior described in Section~\ref{effOu.sec:identifiability_and_unconstraint_parametrization}, with $\lambda=0.1$.

\subsection{Time-series state-space analyses}
\label{effOu.sec:kalmanSmoother}
The time-series setting is the simplest graphical specialization of the framework presented in Section~\ref{effOu.sec:model}.
The graph is a chain of latent states at observation times \(t_1<\cdots<t_n\), each connected to the data through a Gaussian observation factor.
The general Gaussian message-passing recursions then reduce to Kalman filtering for the marginal likelihood and Rauch--Tung--Striebel smoothing for latent-state moments \citep{kalman_new_1960,rauch_maximum_1965,sarkka_bayesian_2013}.
MAP estimates were obtained by limited-memory BFGS (L-BFGS) optimization \citep{liu_limited_1989} of the log posterior in unconstrained coordinates.

\subsubsection{Simulations} \label{effOu.sec:simulations}
We evaluated drift recovery under the orthogonal and dense-basis H-SSBP parametrizations in three five-dimensional ($\effOuDim=5$) simulation experiments designed to probe behavior across the real--complex spectral boundary, under departures from orthogonal block decomposability, and under model misspecification.
Further details on the simulation design and inference procedure are provided in \suppref{suppeffOu.sec:simulation_material}.
For each value of the deformation parameter in each experiment, we generated $25$ independent stationary OU trajectories with $\effNObs=800$ irregularly spaced observations with inter-observation gaps drawn independently from a uniform distribution on $[0,0.1]$ and all five coordinates observed with Gaussian measurement error.
MAP estimates were obtained under both parametrizations by analytic-gradient L-BFGS from five independent random initializations, retaining the run with the largest final log posterior; all $950$ selected fits satisfied the prespecified optimization-stabilization criterion.

The first experiment varied the rotational parameter $\effOuBlkT{1}$ of one two-dimensional block across the real--complex eigenvalue boundary while remaining within the correctly specified orthogonal H-SSBP family.
The second introduced increasing shear, producing diagonalizable but nonnormal drifts representable by the dense-basis H-SSBP but not, in general, by the orthogonal specialization.
The third introduced a Jordan block of size four lying outside both fitted H-SSBP families.
Figure~\ref{effOu.fig:sim_combined_rmse} shows essentially no loss of recovery accuracy at the real--complex boundary, with median root mean square error (RMSE) ranging from $0.152$ to $0.156$ for the orthogonal fit and from $0.260$ to $0.262$ for the dense-basis fit.
Under increasing shear, median RMSE for the orthogonal fit increased from $0.648$ to $1.803$, whereas the dense-basis fit remained between $0.340$ and $0.358$, consistent with the difference in matrix-space support.
Under the defective Jordan-block truth, orthogonal-fit RMSE ranged from $0.285$ to $0.444$, whereas dense-basis RMSE remained between $0.310$ and $0.329$, indicating greater robustness of the dense basis as the Jordan coupling increased.

\begin{figure}[t]
\centering
\includegraphics[width=0.98\textwidth]{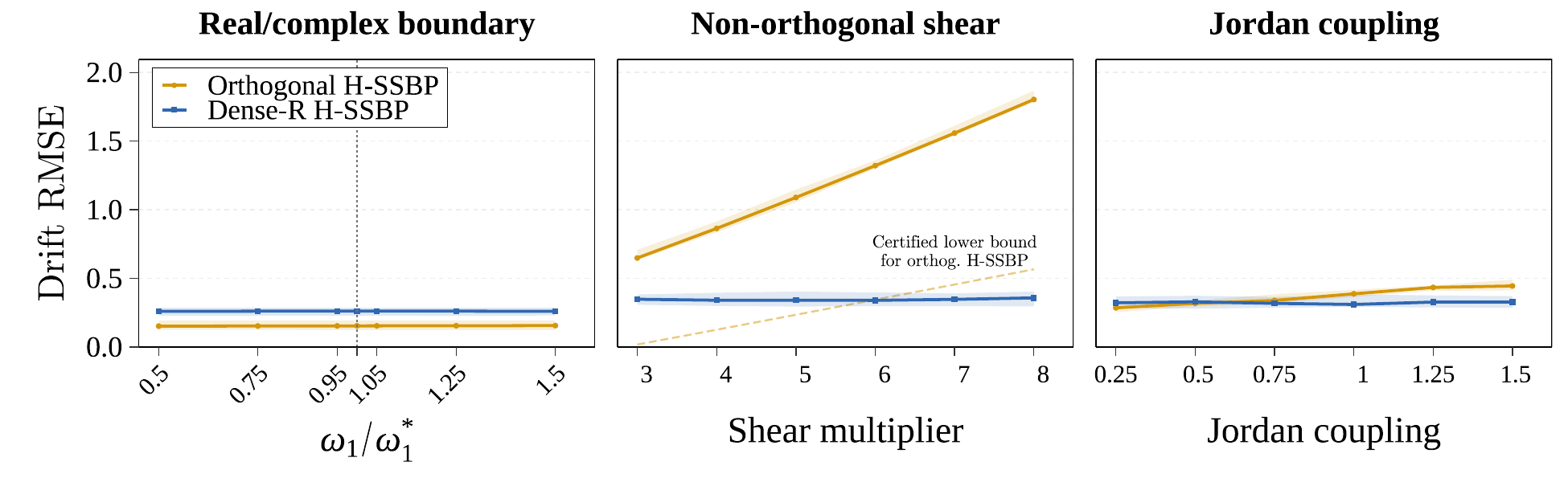}
\caption{Drift reconstruction RMSE under three deformations.
Points and ribbons show the median and interquartile range across $25$ simulated trajectories.
The panels vary, respectively, the position relative to the real--complex eigenvalue boundary, non-orthogonal shear, and the coupling of a $4$-dimensional Jordan-block truth.
In the first panel, $\effOuBlkT{1}$ changes while $\effOuBlkT{1}^{*}$ is the fixed boundary value separating real and complex block eigenvalues: 
$\effOuBlkT{1}^{*}=|\effOuCartBlockRate_1\effOuCartNormalizingRate_1|$.
The dashed curve in the middle panel gives the lower bound on RMSE attainable by any orthogonal H-SSBP fit.
Further details are provided in \suppref{suppeffOu.sec:simulation_material}.}
\label{effOu.fig:sim_combined_rmse}
\end{figure}

\subsubsection{BitMEX trade data example}
\label{effOu.sec:bitmex_trade_data}
We next applied the H-SSBP model to asynchronous high-frequency BitMEX trade data retrieved through the public BitMEX REST API \citep{bitmex_api_2026}.
The records consisted of timestamped trades for actively listed cryptocurrency instruments, giving an irregularly observed multivariate time series with different coordinates observed at different times.
We treated each trade time as a partial observation of the shared latent OU state and analytically integrated the unobserved coordinates following \citet{bastide_efficient_2021} and \citet{hassler_inferring_2022}.
We analyzed five instruments: \texttt{XBTUSDT}, \texttt{ETHUSDT}, \texttt{LTCUSDT}, \texttt{BCHUSDT}, and \texttt{DOGEUSDT}.
For these instruments, we analyzed log prices after size-weighted aggregation and within-window centering over the post-shock recovery window from 13:00 to 20:00 UTC on January~24, 2022, beginning at the one-minute \texttt{XBTUSDT} interval containing the local trough.
We fitted the general invertible-basis H-SSBP with one scalar block and two ordered two-dimensional blocks.
We also inferred the equilibrium mean $\effOuEquilMean$ and diffusion matrix $\effOuDiffMat$, while fixing the observation-error variance to $10^{-5}$.
The MAP search from ten random starts produced a reproducible optimum, with negative diagonal self-reversion coefficients and a pronounced positive \texttt{DOGEUSDT}-to-\texttt{BCHUSDT} conditional coupling, while most other cross-instrument terms remained close to zero.
The reconstructed drift was asymmetric but had an entirely real, negative spectrum, illustrating that the H-SSBP can capture directional coupling without requiring oscillatory modes.
Time-series trajectories, the reconstructed drift, and additional diagnostics are provided in \suppref{suppeffOu.sec:empirical_details}.

\subsection{Diffusion over phylogenetic trees} \label{effOu.sec:phylogeneticExample}
The OU process is a standard model for quantitative-trait evolution on phylogenies, extending Brownian-motion comparative models by allowing stabilizing selection toward adaptive optima \citep{felsenstein_phylogenies_1985,hansen_stabilizing_1997,butler_phylogenetic_2004}.
Multivariate OU models extend this framework to the joint evolution of several traits, allowing correlated trait dynamics and interactions among traits in their responses to selection \citep{bartoszek_phylogenetic_2012,bastide_efficient_2021}.
In this setting, a general nonsymmetric drift matrix can represent asymmetric cross-trait effects in the dynamics of adaptation.
As a real-data illustration, we analyzed the \emph{Anolis} lizard comparative data distributed with \texttt{phytools} \citep{revell_phytools_2012}. 
The analysis used a fixed 100-tip phylogeny (left panel of Figure~\ref{fig:anole_drift_matrix}) and included six centered, unit-variance morphological traits: snout--vent length, head length, hindlimb length, forelimb length, lamella number, and tail length.
We fitted a multivariate phylogenetic OU model using an H-SSBP drift matrix with an orthogonal basis.
In addition to the drift, the equilibrium mean $\effOuEquilMean$ and the diffusion covariance $\effOuDiffMat$ were inferred.
The root distribution used a conjugate Gaussian prior with mean $\effOuEquilMean$ and prior sample size $1$.
We ran four HMC chains from different starting values, each for $2{,}000{,}000$ iterations, logged every $100$ iterations, and discarded the first $20\%$ of each chain as burn-in.
Posterior summaries were computed from the resulting $64{,}000$ retained draws pooled across the four chains.
Convergence diagnostics indicated good mixing and between-chain agreement: the maximum rank-normalized \(\widehat R\) across reconstructed drift-matrix entries was \(1.002\); the minimum (median) bulk ESS was approximately \(2{,}300\) (\(7{,}900\)), and the minimum (median) tail ESS was approximately \(2{,}100\) (\(10{,}400\)).

\begin{figure}[H]
\centering
\includegraphics[width=\textwidth]{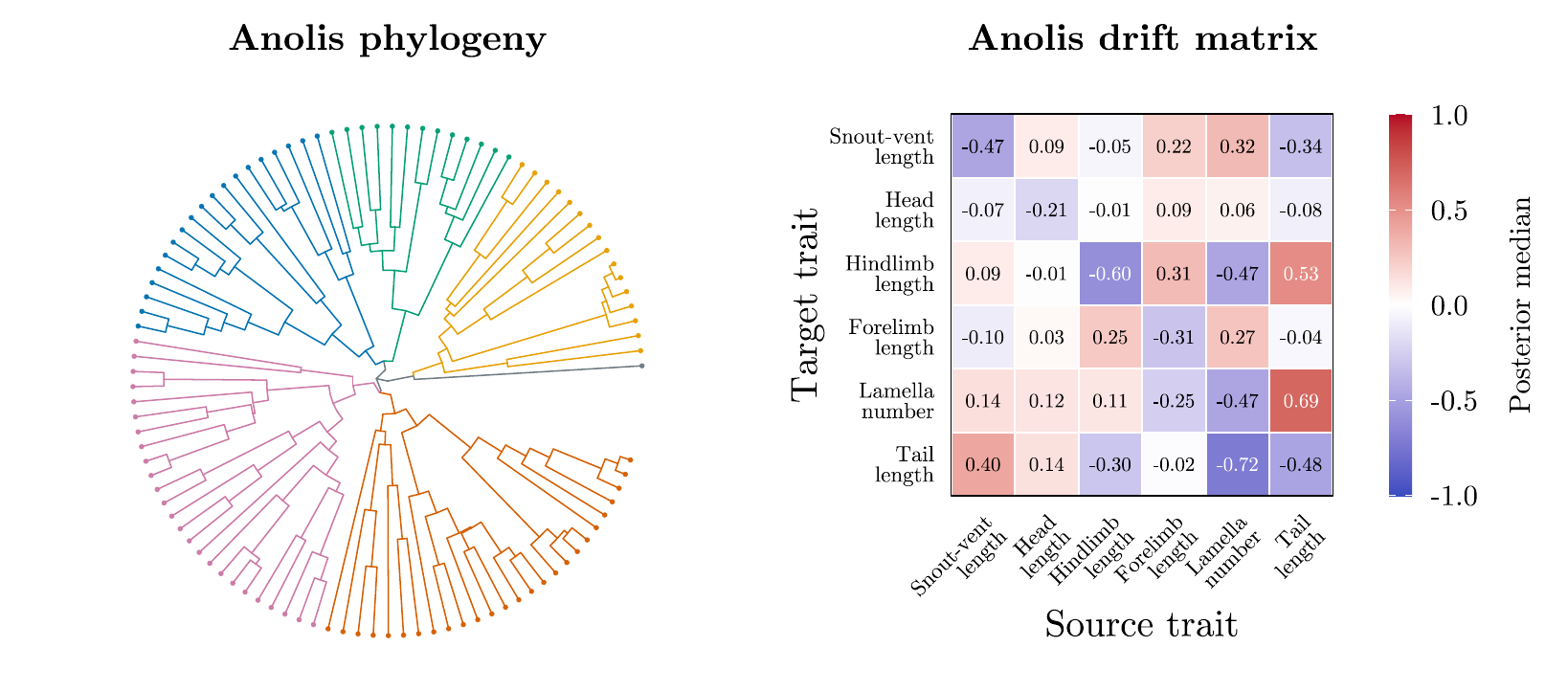}
\caption{
Phylogeny and posterior median \emph{Anolis} drift matrix.
Left: fixed 100-tip phylogeny used for the comparative analysis.
Edge colors distinguish major topological clades for visualization and do not encode trait or ecomorph states.
Right: posterior median drift matrix after $20\%$ burn-in.
Rows are target traits and columns are source traits, so the off-diagonal entry in row $i$, column $j$, quantifies the effect of the deviation of trait $j$ on the drift of trait $i$.
}
\label{fig:anole_drift_matrix}
\end{figure}
The posterior median drift matrix has negative diagonal entries, as implied by the orthogonal H-SSBP parametrization.
Most off-diagonal entries have $95\%$ highest posterior density (HPD) intervals that include zero, while four have intervals excluding zero.
The tail-length-to-snout--vent-length and lamella-number-to-hindlimb-length entries are negative, with posterior medians $-0.34$ and $-0.47$ and $95\%$ HPD intervals $[-0.69,-0.02]$ and $[-0.97,-0.07]$, respectively.
The tail-length-to-hindlimb-length and tail-length-to-lamella-number entries are positive, with posterior medians $0.53$ and $0.69$ and $95\%$ HPD intervals $[0.12,0.99]$ and $[0.18,1.40]$, respectively.
Overall, the posterior provides evidence for directional cross-trait drift coupling, with several off-diagonal effects supported away from zero.

The H-SSBP block representation allowed the posterior to move across real and complex eigenvalue regimes.
Using the criterion $\effOuBlkT{\effOuIdxB}^{2} > \effOuCartBlockRate_{\effOuIdxB}^{2}\effOuCartNormalizingRate_{\effOuIdxB}^{2}$ for each block, the three ordered blocks were in the complex-eigenvalue regime in $60\%$, $47\%$, and $80\%$ of the pooled post-burn-in draws, respectively.
Across the four chains, the corresponding numbers of transitions between real and complex regimes ranged from $4{,}500$ to $5{,}300$, $5{,}500$ to $6{,}200$, and $1{,}400$ to $1{,}700$, respectively.
Although these counts depend on the logging interval, they show that every chain repeatedly crossed the real--complex boundary rather than remaining confined to a single spectral regime.

\section{Discussion}\label{effOu.sec:discussion}

This work develops the H-SSBP, a stable drift parametrization for multivariate OU models designed around likelihood-based inference.
By representing Hurwitz drift matrices through low-dimensional stable blocks and a global change of basis, the H-SSBP combines broad spectral flexibility with structured computation of the matrix exponential, stationary Lyapunov equation, and their reverse-mode derivatives.
For OU processes on rooted directed trees, we exploit this structure within exact Gaussian message passing to substantially reduce the cost of likelihood and gradient evaluation.

The computational gains arise from replacing general dense matrix-function and Lyapunov calculations with independent constant-size block or block-pair operations together with change-of-basis multiplications.
This avoids much of the computational overhead of standard dense algorithms, including Schur-based matrix-function methods and Bartels–Stewart-type Lyapunov solvers, while retaining exact likelihoods and gradients.
Under homogeneous dynamics, a common block basis can be used throughout the graph, eliminating repeated edge-level changes of basis and further reducing dense computation.

Some limitations remain.
The method does not remove all cubic costs from Gaussian inference: dense rotations, Gaussian conditioning, and message algebra may still dominate in some regimes.
The generic H-SSBP is not an identifiable coordinate system without gauge conventions, because equivalent changes of basis can represent the same drift.
The orthogonal specialization improves conditioning and removes some gauge freedom, but restricts matrix-space support to orthogonally block-decomposable dynamics.

As a future direction, the same block structure could also be exploited for other matrix functions that arise
in stable linear dynamics \citep{higham_functions_2008,trefethen_spectra_2005}, such as resolvents, time-integrated semigroups, and trace integrals:

\vspace{-0.5in}

\begin{align}
(s\effOuIdentity-\effOuDriftMatrix)^{-1}
& \effOuEqualityWithHSSBP
\effOuRotMat
(s\effOuIdentity-\effOuBlkDiag)^{-1}
\effOuRotMat^{-1} \qquad s\notin\operatorname{spec}(\effOuDriftMatrix).\\
\int_{0}^{\effOuTime}
\effOuExpMatrix{\effOuDriftMatrix u}
\,\dd u
=
\effOuDriftMatrix^{-1}
\left(
\effOuExpMatrix{\effOuDriftMatrix\effOuTime}
-
\effOuIdentity
\right) 
&\effOuEqualityWithHSSBP
\effOuRotMat
\effOuBlkDiag^{-1}
\left(
\effOuExpMatrix{\effOuBlkDiag\effOuTime}
-
\effOuIdentity
\right)
\effOuRotMat^{-1}
\; \overset{\effOuTime \to \infty}{\longrightarrow}\;
-
\effOuRotMat
\effOuBlkDiag^{-1}
\effOuRotMat^{-1}, \\
\int_{0}^{\infty}
\tr\!\left(
\effOuExpMatrix{\effOuDriftMatrix u}
\right)
\,\dd u
=
-
\tr\!\left(
\effOuDriftMatrix^{-1}
\right)
&\effOuEqualityWithHSSBP
-
\tr\!\left(
\effOuBlkDiag^{-1}
\right).
\end{align}
Moreover, this work has focused on chains and rooted trees. Extending the structured likelihood and adjoint calculations to more general graphical structures is a natural direction, with potential applications to phylogenetic networks \citep{bastide_phylogenetic_2018, teo_leveraging_2025} and dynamic biological networks \citep{zou_new_2005,sanchez-castillo_bayesian_2018}.

\noindent Overall, the H-SSBP offers a way to retain expressive stable linear dynamics while making exact gradient-based inference more transparent, computationally structured, and numerically stable.

\if\depIndAnonymous1
\section*{Data availability statement} 
The data and code required to reproduce the numerical experiments, empirical analyses, and figures are available in the \href{https://github.com/suchard-group/stableMatrixParametrizationAndStructuredAdjoints}{project repository}.
\else
\section*{Data availability statement} 
The data and code required to reproduce the numerical experiments, empirical analyses, and figures are available in the public GitHub repository at \url{https://github.com/suchard-group/stableMatrixParametrizationAndStructuredAdjoints}.

\section*{Disclosure statement}
The authors report no conflicts of interest.

\section*{Acknowledgments}
We gratefully acknowledge Advanced Micro Devices, Inc.\ for donating parallel computing resources used for this research.

\section*{Funding}
MAS was supported by National Institutes of Health grants U19 AI135995, R01 AI153044, and R01 AI162611.
AJH was supported by NSF DMS 2236854.
NEGH was supported by NSF DMS 2510856.
\fi
\onehalfspacing

\bibliographystyle{apalike}
\bibliography{references}


\end{document}


\doublespacing

\if\depIndAnonymous1
\else
\begin{flushright}
Version dated: August 17, 2026\\
\end{flushright}

\bigskip
\fi

\begin{center}
\noindent{\LARGE \bf
\begin{singlespace}
Supplementary Material\\[4pt]
for ``Stable Matrix Parametrizations \\
and Structured Adjoints\\
for Ornstein--Uhlenbeck Processes''
\end{singlespace}
}

\bigskip
\if\depIndAnonymous1

\else
\noindent{\normalsize \sc
Filippo Monti, Andrew Holbrook, Nathan E. Glatt-Holtz, and Marc A.~Suchard}
\fi
\end{center}
\bigskip

\section{SSBP Block Structure and Support}
\label{suppeffOu.overview}

This supplementary section collects the block-level algebra used by the smooth spectral
block parametrization (SSBP) and its Hurwitz specialization (H-SSBP) presented in the main text.
Throughout, let $\effOuDriftMatrix$ be a $\effOuDim$-dimensional square matrix which can be written as
\begin{align}
  \effOuDriftMatrix = \effOuRotMat \effOuBlkDiag \effOuRotMat^{-1},
  \qquad
  \effOuBlkDiag = \operatorname{bdiag}
  (\effOuBlk{1},\ldots,\effOuBlk{\effNBlock}),
\end{align}
where each block has size one or two.
We also discuss the orthogonal specialization when
\(\effOuRotMat^{-1}=\effOuRotMat^{\effOuTranspose}\).

\subsection{Notation and block identities}
\label{suppeffOu.notation}
For any matrix \(\mathbf{X}\), \(\mathbf{X}_{\effOuIdxB \effOuIdxBTwo}\) denotes its block with rows in block \(\effOuIdxB\) and columns in block \(\effOuIdxBTwo\).
When $\effOuIdxB = \effOuIdxBTwo$, we contract the notation to \(\mathbf{X}_{\effOuIdxB}\).

\begin{definition}[SSBP blocks] The block-diagonal matrix in the SSBP has two-dimensional blocks of the form
\begin{align} \label{suppeffOu.eq:block_definition}
  \effOuBlk{\effOuIdxB}
  =
  \effOuCartBlockRate_{\effOuIdxB}\effOuIdentity_2+\effOuBlkN{\effOuIdxB},
\end{align}
with hollow residual
\begin{align}
  \effOuBlkN{\effOuIdxB}
  =
  \begin{pmatrix}
    0 &
    \effOuCartBlockRate_{\effOuIdxB}\effOuCartNormalizingRate_{\effOuIdxB}
    +\effOuBlkT{\effOuIdxB}
    \\
    \effOuCartBlockRate_{\effOuIdxB}\effOuCartNormalizingRate_{\effOuIdxB}
    -\effOuBlkT{\effOuIdxB}
    & 0
  \end{pmatrix}.
\end{align}
If the matrix dimension is odd, then one scalar block $\effOuBlk{\effOuIdxB} = \effOuCartBlockRate_{\effOuIdxB}$ is also included.
\end{definition}
\begin{remark}[Hurwitz-SSBP] To restrict the SSBP to be Hurwitz stable, we impose the following two constraints on each two-dimensional block:
  \begin{align}
  \effOuCartBlockRate_{\effOuIdxB}<0, \quad | \effOuCartNormalizingRate_{\effOuIdxB} | < 1.
  \end{align}
When the dimension is odd, the scalar block is additionally required to be negative.
\end{remark}

\begin{remark}[Equal-diagonal structure and algebraic closure]
  \label{suppeffOu.rem:equal_diag}
Because the two diagonal entries of \(\effOuBlk{\effOuIdxB}\) are both
\(\effOuCartBlockRate_{\effOuIdxB}\), the centered residual
\(\effOuBlkN{\effOuIdxB}
=
\effOuBlk{\effOuIdxB}
-
\effOuCartBlockRate_{\effOuIdxB}\effOuIdentity_2\)
is hollow and satisfies
\begin{align}
  \effOuBlkN{\effOuIdxB}^{2}
  =
  \effOuBlkDelta{\effOuIdxB}\effOuIdentity_2,
  \qquad
  \effOuBlkDelta{\effOuIdxB}
  =
  \left(
  \effOuCartBlockRate_{\effOuIdxB}\effOuCartNormalizingRate_{\effOuIdxB}
  +
  \effOuBlkT{\effOuIdxB}
  \right)
  \left(
  \effOuCartBlockRate_{\effOuIdxB}\effOuCartNormalizingRate_{\effOuIdxB}
  -
  \effOuBlkT{\effOuIdxB}
  \right)
  =
  \effOuCartBlockRate_{\effOuIdxB}^{2}
  \effOuCartNormalizingRate_{\effOuIdxB}^{2}
  -
  \effOuBlkT{\effOuIdxB}^{2}.
\end{align}
Thus even powers of \(\effOuBlkN{\effOuIdxB}\) are scalar multiples of
\(\effOuIdentity_2\), while odd powers are scalar multiples of
\(\effOuBlkN{\effOuIdxB}\).
\end{remark}

\subsection{Support}
\label{suppeffOu.ssbp_support}

\begin{proposition}[Full spectral support]\label{suppeffOu.prop:full_support}
  Let $\effOuBlk{\effOuIdxB} = \effOuBlk{\effOuIdxB}(\effOuCartBlockRate, \effOuCartNormalizingRate, \effOuBlkT{})$ be a $2\times 2$ block as in Equation~\eqref{suppeffOu.eq:block_definition} subject to the stability constraints $\effOuCartBlockRate < 0$ and $|\effOuCartNormalizingRate| < 1$.
  Then $\effOuBlk{\effOuIdxB}$
  achieves full spectral support over all Hurwitz eigenvalue configurations, that is:
\begin{enumerate}[label=\roman*)]
  \item For any ordered real pair $\effOuEigenvalue^- \le \effOuEigenvalue^+ < 0$, there exists a $\effOuBlk{\effOuIdxB}(\effOuCartBlockRate, \effOuCartNormalizingRate, \effOuBlkT{})$ with eigenvalues $\effOuEigenvalue^\pm$.
  \item For any complex-conjugate pair $\effOuEigenvalueRealPart \pm \mathrm{i}\effOuEigenvalueImaginaryPart$ with $\effOuEigenvalueRealPart < 0$, $\effOuEigenvalueImaginaryPart \in \mathbb{R}$, there exists a $\effOuBlk{\effOuIdxB}(\effOuCartBlockRate, \effOuCartNormalizingRate, \effOuBlkT{})$ with eigenvalues $\effOuEigenvalueRealPart \pm \mathrm{i}\effOuEigenvalueImaginaryPart$.
\end{enumerate}
In particular, the imaginary part of the eigenvalues is unbounded and the parametrization covers the entire open left half-plane.
\end{proposition}
\begin{proof}
\textit{(i) Real pair.}
Set
\[
\effOuCartBlockRate = \frac{\effOuEigenvalue^- + \effOuEigenvalue^+}{2},
\qquad
\effOuCartNormalizingRate =
\frac{\effOuEigenvalue^+ - \effOuEigenvalue^-}{\effOuEigenvalue^- + \effOuEigenvalue^+},
\qquad
\effOuBlkT{}=0.
\]
Since $\effOuEigenvalue^- \le \effOuEigenvalue^+ < 0$, we have $\effOuCartBlockRate<0$.
If $\effOuEigenvalue^- = \effOuEigenvalue^+$ then $\effOuCartNormalizingRate=0$.
Otherwise \(0 < \effOuEigenvalue^+ - \effOuEigenvalue^- < |\effOuEigenvalue^- + \effOuEigenvalue^+|\), so $|\effOuCartNormalizingRate|<1$.
The discriminant is $\effOuBlkDelta{}=\effOuCartBlockRate^2\effOuCartNormalizingRate^2\ge0$ with $\sqrt{\effOuBlkDelta{}}=(\effOuEigenvalue^+ - \effOuEigenvalue^-)/2$, giving eigenvalues $\effOuCartBlockRate\pm\sqrt{\effOuBlkDelta{}}=\effOuEigenvalue^\pm$.
\textit{(ii) Complex conjugate pair.}
Set $\effOuCartBlockRate = \effOuEigenvalueRealPart < 0$, $\effOuCartNormalizingRate = 0$, and $\effOuBlkT{} = \effOuEigenvalueImaginaryPart$.
Then $\effOuBlkDelta{}=-\effOuEigenvalueImaginaryPart^2\le0$, with eigenvalues $\effOuEigenvalueRealPart\pm\mathrm{i}\effOuEigenvalueImaginaryPart$.
\end{proof}
\noindent The proposition above concerns the spectral support of individual blocks.
Matrix-space support follows by allowing an arbitrary real similarity basis.
\begin{proposition} \label{suppeffOu.prop:supportSSBP}
With the generic form $\effOuDriftMatrix=\effOuRotMat\effOuBlkDiag\effOuRotMat^{-1}$, $\effOuRotMat\in GL_p(\mathbb{R})$,
\begin{itemize}
\item[1] the H-SSBP represents every real Hurwitz matrix diagonalizable over $\mathbb{C}$, and more generally every real Hurwitz matrix
whose real Jordan decomposition can be organized into scalar real blocks, two-dimensional real blocks for complex-conjugate eigenvalue pairs, and length-two Jordan blocks for defective repeated real eigenvalues.
\item[2] the image of the H-SSBP is dense in the set of real Hurwitz matrices endowed with the Euclidean topology on
\(\mathbb{R}^{\effOuDim \times \effOuDim}\).
It also has full Lebesgue measure within the Hurwitz cone.
\end{itemize} 
\end{proposition}

\noindent This result follows from Proposition~\ref{suppeffOu.prop:full_support}.
For the diagonalizable part, a real matrix can be decomposed over real invariant subspaces into scalar blocks for real eigenvalues and two-dimensional blocks for complex-conjugate eigenvalue pairs \citep{gantmakher_theory_2000, horn_matrix_2012, golub_matrix_2013}.
Proposition~\ref{suppeffOu.prop:full_support} shows that the H-SSBP realizes all Hurwitz configurations of these scalar and two-dimensional spectral blocks.
The remaining size-two defective real case is represented by an H-SSBP block with \(\effOuBlkDelta{\effOuIdxB}=0\) and nonzero nilpotent residual \(\effOuBlkN{\effOuIdxB}\).
Jordan chains of length greater than two, and defective complex eigenvalue pairs requiring real blocks of dimension larger than two, are not represented by the strict block-diagonal form.
Since matrices with distinct eigenvalues are diagonalizable over
\(\mathbb{C}\) and dense in the ambient matrix space
\citep{horn_matrix_2012}, and since the Hurwitz set is open, every real Hurwitz
matrix can be approximated arbitrarily well by a real Hurwitz matrix with
distinct eigenvalues.
Such an approximating matrix is diagonalizable over \(\mathbb{C}\), and hence is
represented by the H-SSBP by part~1 of the proposition.
Therefore, the H-SSBP image is dense in the set of real Hurwitz matrices.
Moreover, matrices with repeated eigenvalues lie in the zero set of the characteristic-polynomial discriminant, a proper algebraic subset of \(\mathbb{R}^{\effOuDim\times\effOuDim}\), and hence form a Lebesgue-null set.
Thus simple-spectrum matrices form a full-measure subset of the Hurwitz cone, and all such Hurwitz matrices are represented by the H-SSBP.

\subsection{Orthogonal specialization} \label{suppeffOu.orthogonalSpecialization}
To improve conditioning \citep{stewart_error_1973, golub_matrix_2013}, we can constrain the change-of-basis matrix to be orthogonal, \(\effOuRotMat^{\effOuTranspose} \effOuRotMat = \effOuIdentity\).
\begin{proposition}[Support under the orthogonal specialization]
\label{suppeffOu.prop:orthogonal_support}
If \(\effOuRotMat\) is restricted to be orthogonal, the H-SSBP represents the stable matrices that are orthogonally block-decomposable into H-SSBP scalar and \(2\times2\) components.
In particular, every represented orthogonal H-SSBP matrix is strictly dissipative, with negative-definite symmetric part.
\end{proposition}

\noindent This restriction is related to, but weaker than, normality.
Real normal matrices are orthogonally block diagonalizable into scalar blocks
and \(2\times2\) rotation-dilation blocks, and hence form a subclass of the
orthogonal H-SSBP support.
The converse need not hold, because a two-dimensional H-SSBP block need not be
normal.

For inference, several smooth parametrizations of orthogonal matrices are
available, including exponential maps of skew-symmetric matrices, Cayley
transforms, Householder products, and products of Givens rotations \citep{golub_matrix_2013}.
In our analyses, we use a fixed ordered product of Givens rotations
to parametrize \(\effOuRotMat\).
For each coordinate pair $1\leq i<j\leq \effOuDim$, let $\mathbf{G}_{ij}(\effOuGivensAngles_{ij})$ be the identity matrix except on the $(i,j)$ coordinate plane, where it has the $2\times2$ rotation
\begin{align}
  \begin{pmatrix}
  \cos \effOuGivensAngles_{ij} & -\sin \effOuGivensAngles_{ij}\\
  \sin \effOuGivensAngles_{ij} & \cos \effOuGivensAngles_{ij}
  \end{pmatrix}.
\end{align}
We then set
\begin{align}
  \effOuRotMat(\boldsymbol{\effOuGivensAngles})
  =
  \prod_{1\leq i<j\leq \effOuDim}
  \mathbf{G}_{ij}(\effOuGivensAngles_{ij}),
\end{align}
using a fixed lexicographic order for the product.
The $\effOuDim(\effOuDim-1)/2$ angles $\effOuGivensAngles_{ij}\in\mathbb{R}$ are unconstrained, so the parametrization is smooth and automatically orthogonal, with $\effOuRotMat^{-1}=\effOuRotMat^{\effOuTranspose}$.
Each Givens rotation has determinant \(+1\), so this product parametrizes the determinant-one component of the orthogonal group.
The other component could be obtained by prepending one fixed reflection, but the sign of the determinant is not identifiable in $\effOuRotMat \effOuBlkDiag \effOuRotMat^{\effOuTranspose}$, so we consider only scenarios with determinant equal to \(+1\).

\section{Alternative Parametrizations for $2\times 2$ Blocks} \label{suppeffOu.sec:alternativeParametrizationsForTwodimensionalBlocks}

The key observation that motivated the structure of the H-SSBP is that $2\times 2$ blocks are the smallest components that allow us to represent any Hurwitz real eigenvalue pair and any Hurwitz complex-conjugate pair.
The main value of the proposed parametrization lies in balancing three objectives:
\begin{itemize}
  \item[(i)] \emph{parsimony}, reducing nonidentifiability;
  \item[(ii)] \emph{stability constraints} that factor across parameters and can be enforced through unconstrained transformations;
  \item[(iii)] \emph{smooth matrix-valued transitions} between real- and complex-eigenvalue regimes.
\end{itemize}

\noindent We define the H-SSBP two-dimensional block family as
\begin{align}
\mathcal{H}
&=
\left\{
\begin{pmatrix}
\,a\, & \,b\, \\
\,c\, & \,a\,
\end{pmatrix}
:
a<0,\;
|b+c|<-2a
\right\}.
\end{align}
The coordinate system introduced in Equation~\eqref{suppeffOu.eq:block_definition} and used in the main manuscript parametrizes exactly this family.
For the remainder of this section, we refer to these coordinates as the \emph{relative Cartesian coordinates}.
In Section~\ref{suppeffOu.sec:alternativeParametrizationWithinHssbpFamily}, we compare them with two alternative coordinate systems that parametrize the same H-SSBP block family: a fully Cartesian coordinate system and a polar coordinate system.
In Section~\ref{suppeffOu.sec:alternativeParametrizationOutsideHssbpFamily}, we then discuss three alternative parametrizations for $2\times2$ blocks that are not equivalent to the H-SSBP block family and evaluate them under the three objectives above.

\subsection{Equivalent coordinate systems for H-SSBP blocks}
\label{suppeffOu.sec:alternativeParametrizationWithinHssbpFamily}

\paragraph{Blocks in fully Cartesian coordinates}
An alternative coordinate system for the same H-SSBP $2\times2$ block family is obtained by replacing the relative coordinate $\effOuCartNormalizingRate$ with the Cartesian coordinate $\sigma=\effOuCartBlockRate\effOuCartNormalizingRate\in\mathbb{R}$, while keeping $\effOuCartBlockRate<0$, so that the $2\times2$ block becomes
\begin{align}
\effOuBlkDiag(\effOuCartBlockRate, \sigma, \effOuBlkT{})
&=
\effOuCartBlockRate\,\effOuIdentity_2
+
\begin{pmatrix}
0 & \sigma + \effOuBlkT{} \\[4pt]
\sigma - \effOuBlkT{} & 0
\end{pmatrix}.
\end{align}
Its discriminant is
\(\effOuBlkDelta{}=\sigma^2-\effOuBlkT{}^2\).
The stability constraint $|\effOuCartNormalizingRate|<1$ becomes $|\sigma|<|\effOuCartBlockRate|$.
The relative Cartesian coordinates are recovered via $\effOuCartNormalizingRate=\sigma/\effOuCartBlockRate$, confirming that the two coordinate systems parametrize exactly the same H-SSBP block family.
Compared with the relative Cartesian coordinates, however, this fully Cartesian form introduces a coupled stability constraint.

\paragraph{Blocks in polar coordinates}
Another equivalent coordinate system for the same H-SSBP $2\times2$ block family is obtained by replacing $(\effOuCartBlockRate,\effOuCartNormalizingRate)$ with polar coordinates $\effOuBlkRho{}<0$ and $\effOuBlkTheta{}\in(-\pi/4,\pi/4)$ satisfying
\begin{align}
\effOuCartBlockRate
&=
\effOuBlkRho{}\cos\effOuBlkTheta{},
&
\effOuCartNormalizingRate
&=
\tan\effOuBlkTheta{}.
\end{align}
Substituting these expressions into the H-SSBP block gives
\begin{align}
  \effOuBlkDiag(\effOuBlkRho{}, \effOuBlkTheta{}, \effOuBlkT{})
  =
  \begin{pmatrix}
    \effOuBlkRho{}\cos\effOuBlkTheta{} &
    \effOuBlkRho{}\sin\effOuBlkTheta{}+\effOuBlkT{} \\[4pt]
    \effOuBlkRho{}\sin\effOuBlkTheta{}-\effOuBlkT{} &
    \effOuBlkRho{}\cos\effOuBlkTheta{}
  \end{pmatrix},
\end{align}
with discriminant $\effOuBlkDelta{}=\effOuBlkRho{}^2\sin^2\effOuBlkTheta{}-\effOuBlkT{}^2$.
The relative Cartesian coordinates are recovered via $\effOuCartBlockRate=\effOuBlkRho{}\cos\effOuBlkTheta{}$ and $\effOuCartNormalizingRate=\tan\effOuBlkTheta{}$, while the polar coordinates are recovered via $\effOuBlkTheta{}=\arctan\effOuCartNormalizingRate$ and $\effOuBlkRho{}=\effOuCartBlockRate\sqrt{1+\effOuCartNormalizingRate^2}$.
Thus the two coordinate systems parametrize exactly the same H-SSBP block family.
Compared with the relative Cartesian coordinates, however, the polar form makes the contraction parameter less directly interpretable: the common diagonal entry and the real part of complex block eigenvalues are $\effOuBlkRho{}\cos\effOuBlkTheta{}$, not $\effOuBlkRho{}$ itself.

\subsection{Alternative parametrizations beyond the H-SSBP block family} \label{suppeffOu.sec:alternativeParametrizationOutsideHssbpFamily}

\paragraph{Unconstrained entrywise parametrization} The most general parametrization assigns an independent parameter to each matrix entry:
\begingroup
\setlength{\arraycolsep}{9pt}
\begin{align}
\begin{bmatrix}
a & b \\
c & d
\end{bmatrix}.
\end{align}
\endgroup
Its eigenvalues are
\begin{align}
\lambda_{\pm}
=
\frac{a+d\pm\sqrt{(a-d)^2+4bc}}{2}.
\end{align}
Both eigenvalues have negative real part if and only if
\begin{align} \label{suppeffOu.eq:unconstrained2x2StabilityConstraints}
a+d<0,
\qquad
ad-bc>0.
\end{align}
Thus, stability requires coupled constraints on the matrix entries.
Moreover, when this block is embedded in a similarity-based parametrization, its unrestricted form can introduce redundant and nonidentifiable parameters.

\paragraph{Equal-diagonal parametrization} The H-SSBP blocks belong to the simpler class of equal-diagonal matrices with independently parametrized off-diagonal entries:
\begingroup
\setlength{\arraycolsep}{9pt}
\begin{align} \label{suppeffOu.eq:equalDiagonalParametrization}
\begin{bmatrix}
a & b \\
c & a
\end{bmatrix}.
\end{align}
\endgroup
Its eigenvalues are
\begin{align}
\lambda_{\pm}
=
a\pm\sqrt{bc}.
\end{align}
Both eigenvalues have negative real part if and only if
\begin{align}
a<0,
\qquad
a^2-bc>0.
\end{align}
Although these constraints are simpler than those in \eqref{suppeffOu.eq:unconstrained2x2StabilityConstraints}, they remain coupled.
The Hurwitz-stable portion of this equal-diagonal family strictly contains the H-SSBP block family: every H-SSBP block satisfies the two inequalities above, but the equal-diagonal family also includes stable blocks that violate the H-SSBP dissipativity constraint \(|b+c|<-2a\).

The tridiagonal parametrization of \citet{mckelvey_state_1996} is built from these equal-diagonal $2\times2$ substructures with independently parametrized off-diagonal entries.
Unlike H-SSBP, however, it retains superdiagonal couplings between adjacent substructures and parametrizes arbitrary, rather than necessarily stable, state matrices.
It therefore neither provides separable stability constraints nor generally permits matrix functions and Lyapunov equations to decompose into independent blockwise computations.

\paragraph{Relationship with the real Schur decomposition}
For $bc<0$, the equal-diagonal matrix in~\eqref{suppeffOu.eq:equalDiagonalParametrization} coincides with the standardized $2\times2$ block used by LAPACK to represent a complex-conjugate eigenvalue pair in the real Schur form \citep{anderson_lapack_1999}.
More generally, the real Schur decomposition writes a real matrix as $\mathbf{R}\mathbf{T}\mathbf{R}^{\top}$, where $\mathbf{R}$ is orthogonal and $\mathbf{T}$ is upper quasi-triangular, with $1\times1$ diagonal blocks for real eigenvalues and $2\times2$ diagonal blocks for complex-conjugate pairs \citep{schur_ueber_1909,golub_matrix_2013}.
Under the canonical Schur convention, two distinct real eigenvalues are represented by two scalar blocks rather than by a single $2\times2$ block.
Consequently, the canonical block partition changes when two real eigenvalues coalesce and become a complex-conjugate pair.

The notable restriction to $c=-b$ gives the normal rotation--dilation block
\begingroup
\setlength{\arraycolsep}{9pt}
\begin{align} \label{suppeffOu.eq:equalDiagonalParametrizationSkewSymmetric}
\begin{bmatrix}
a & b \\
-b & a
\end{bmatrix}.
\end{align}
\endgroup
With an orthogonal change of basis and no upper-block couplings in $\mathbf{T}$, Hurwitz-stable blocks of this form (\(a<0\)) represent only real normal Hurwitz matrices.
This Hurwitz-stable class is a proper subset of the matrices representable by the orthogonal H-SSBP (Section~\ref{suppeffOu.orthogonalSpecialization}).
Allowing upper-block couplings in $\mathbf{T}$ also permits nonnormality between distinct invariant subspaces.

\paragraph{Second-order companion/damping parametrization}
A classical two-parameter alternative arises from the homogeneous damped second-order equation
\begin{align}
\ddot{x}(t)
+
2\zeta\effOuCartOscillatoryRate\dot{x}(t)
+
\effOuCartOscillatoryRate^2x(t)
&=
0,
\end{align}
where $\effOuCartOscillatoryRate>0$ is the undamped natural frequency and $\zeta>0$ is the damping ratio \citep{inman_engineering_2022}.
Writing this equation as a first-order system for $(x(t),\dot{x}(t))^{\top}$ gives the companion block
\begin{align}
\begin{bmatrix}
0 & 1 \\
-\effOuCartOscillatoryRate^2 & -2\zeta\effOuCartOscillatoryRate
\end{bmatrix},
\end{align}
with eigenvalues
\begin{align}
\lambda_{\pm}
&=
-\zeta\effOuCartOscillatoryRate
\pm
\effOuCartOscillatoryRate\sqrt{\zeta^2-1}.
\end{align}
The regimes $0<\zeta<1$, $\zeta=1$, and $\zeta>1$ correspond respectively to underdamped complex eigenvalues, a critically damped repeated eigenvalue, and overdamped distinct real eigenvalues.
The independent positivity constraints on $\effOuCartOscillatoryRate$ and $\zeta$ guarantee Hurwitz stability and can be enforced through separate unconstrained transformations.
Moreover, every Hurwitz spectral pair is represented by taking
\begin{align}
\effOuCartOscillatoryRate
&=
\sqrt{\lambda_1\lambda_2},
&
\zeta
&=
-\frac{\lambda_1+\lambda_2}{2\sqrt{\lambda_1\lambda_2}}.
\end{align}

This spectral completeness translates differently into matrix-space support depending on the change of basis.
Under a generic invertible change of basis, every non-scalar $2\times2$ Hurwitz matrix is similar to a companion block with the same characteristic polynomial.
The exception occurs at a semisimple repeated eigenvalue: when $\zeta=1$, the companion block is necessarily defective and therefore cannot represent $-\effOuCartOscillatoryRate\effOuIdentity_2$ exactly.
Moreover, as a family of diagonalizable matrices approaches such a semisimple repeated-eigenvalue limit, its representation in companion form may require an increasingly ill-conditioned change-of-basis matrix.
Thus, the companion parametrization has generic matrix-space support under an unrestricted similarity basis, but can transfer numerical difficulty to the similarity transformation near this boundary.

The restriction is substantially stronger under an orthogonal change of basis.
The symmetric part of the companion block has determinant
\begin{align}
-\frac{\left(1-\effOuCartOscillatoryRate^2\right)^2}{4}
&\leq
0,
\end{align}
and is therefore never negative definite and is generically indefinite.
Since orthogonal similarity preserves the eigenvalues of the symmetric part, no orthogonal change of basis can recover the strictly dissipative stable matrices included in the orthogonal H-SSBP.
Thus, although the companion form is parsimonious, stability-constrained, and spectrally complete, its limitations differ by basis choice: under a generic invertible basis they arise primarily near semisimple repeated eigenvalues, whereas under an orthogonal basis its canonical position--velocity geometry imposes a substantially stronger restriction on matrix-space support.

\section{Forward Kernels for Block Solves}

\subsection{Block matrix exponential} \label{suppeffOu.sec:forward_exp}
For an edge length \(\effOuTime>0\), the OU process actualization matrix is computed as the matrix exponential of the drift matrix:
$  e^{\effOuTime\effOuDriftMatrix}
  =
  \effOuRotMat e^{\effOuTime\effOuBlkDiag}\effOuRotMat^{-1}$.
Since \(\effOuBlkDiag\) is block diagonal, the only nontrivial task is to
evaluate \(e^{\effOuTime\effOuBlk{\effOuIdxB}}\) for each block.
For a scalar block, $e^{\effOuTime\effOuBlk{\effOuIdxB}} = e^{\effOuTime\effOuCartBlockRate_{\effOuIdxB}}$.
For two-dimensional blocks, it is convenient to introduce the scalar functions
\begin{align}
  \cfunc(\effOuBlkDelta{},\effOuTime)
  =
  \begin{cases}
    \cosh(\effOuTime\sqrt{\effOuBlkDelta{}}), & \effOuBlkDelta{}>0,\\
    1, & \effOuBlkDelta{}=0,\\
    \cos(\effOuTime\sqrt{-\effOuBlkDelta{}}), & \effOuBlkDelta{}<0,
  \end{cases}
  \qquad
  \sfunc(\effOuBlkDelta{},\effOuTime)
  =
  \begin{cases}
    \dfrac{\sinh(\effOuTime\sqrt{\effOuBlkDelta{}})}{\sqrt{\effOuBlkDelta{}}}, & \effOuBlkDelta{}>0,\\[6pt]
    \effOuTime, & \effOuBlkDelta{}=0,\\[4pt]
    \dfrac{\sin(\effOuTime\sqrt{-\effOuBlkDelta{}})}{\sqrt{-\effOuBlkDelta{}}}, & \effOuBlkDelta{}<0.
  \end{cases}
\end{align}
These are entire functions of \(\effOuBlkDelta{}\), with the apparent singularity of \(\sfunc\) at \(\effOuBlkDelta{}=0\) being removable.
Using the representation $\effOuBlk{\effOuIdxB}
  =
  \effOuBlkDiagScalar{\effOuIdxB}\effOuIdentity_2+\effOuBlkN{\effOuIdxB}$ \citep{putzer_avoiding_1966, leonard_matrix_1996}, the matrix exponential is
\begin{align}
  e^{\effOuTime\effOuBlk{\effOuIdxB}}
  =
  e^{\effOuTime\effOuBlkDiagScalar{\effOuIdxB}}
  \left[
    \cfunc(\effOuBlkDelta{\effOuIdxB},\effOuTime)\effOuIdentity_2
    +
    \sfunc(\effOuBlkDelta{\effOuIdxB},\effOuTime)\effOuBlkN{\effOuIdxB}
  \right].
\end{align}
To derive this, factor \(e^{\effOuTime\effOuBlk{\effOuIdxB}}
= e^{\effOuTime\effOuBlkDiagScalar{\effOuIdxB}}e^{\effOuTime\effOuBlkN{\effOuIdxB}}\)
and expand the second factor as a power series.
Using \(\effOuBlkN{\effOuIdxB}^{2k}=\effOuBlkDelta{\effOuIdxB}^k\effOuIdentity_2\)
and \(\effOuBlkN{\effOuIdxB}^{2k+1}=\effOuBlkDelta{\effOuIdxB}^k\effOuBlkN{\effOuIdxB}\)
(which follow from the equal-diagonal property, Remark~\ref{suppeffOu.rem:equal_diag}),
the even and odd terms separate:
\begin{align}
  e^{\effOuTime\effOuBlkN{\effOuIdxB}}
  =
  \underbrace{
    \sum_{k=0}^{\infty}\frac{\effOuTime^{2k}\effOuBlkDelta{\effOuIdxB}^k}{(2k)!}
  }_{=\,\cfunc(\effOuBlkDelta{\effOuIdxB},\effOuTime)}\effOuIdentity_2
  +
  \underbrace{
    \sum_{k=0}^{\infty}\frac{\effOuTime^{2k+1}\effOuBlkDelta{\effOuIdxB}^k}{(2k+1)!}
  }_{=\,\sfunc(\effOuBlkDelta{\effOuIdxB},\effOuTime)}\effOuBlkN{\effOuIdxB}.
\end{align}
The three cases of \(\cfunc\) and \(\sfunc\) (hyperbolic, linear, trigonometric) arise directly from the two power series above, which sum to hyperbolic functions when \(\effOuBlkDelta{}>0\), reduce to \(1\) and \(\effOuTime\) when \(\effOuBlkDelta{}=0\), and sum to trigonometric functions when \(\effOuBlkDelta{}<0\).
For a fully general \(2\times2\) block without the equal-diagonal constraint, the same two-dimensional closure is
available only after subtracting half the trace; the H-SSBP equal-diagonal
structure makes this centering automatic and keeps the residual in the hollow
form used throughout these formulas.

\paragraph{Stable evaluation near repeated eigenvalues}
Near \(\effOuBlkDelta{}=0\), direct formulas involving
\(\sqrt{\effOuBlkDelta{}}\) should be evaluated numerically by their series expansions,
\begin{align}
  \cfunc(\effOuBlkDelta{},\effOuTime)
  =
  \sum_{m=0}^{\infty}
  \frac{\effOuTime^{2m}\effOuBlkDelta{}^m}{(2m)!},
  \qquad
  \sfunc(\effOuBlkDelta{},\effOuTime)
  =
  \sum_{m=0}^{\infty}
  \frac{\effOuTime^{2m+1}\effOuBlkDelta{}^m}{(2m+1)!}.
\end{align}
These expansions cover the transition between two real eigenvalues, a repeated
real eigenvalue, and a complex-conjugate pair without changing formulas.

\subsection{Forward Lyapunov block solves}
\label{suppeffOu.sec:forward_lyap}

Let $\tilde{\effOuDiffMat} = \effOuRotMat^{-1}\effOuDiffMat\effOuRotMat^{-\effOuTranspose}$
and $\tilde{\effOuStatVar} = \effOuRotMat^{-1}\effOuStatVar\effOuRotMat^{-\effOuTranspose}$ be the diffusion covariance and stationary covariance in the drift block basis.
The Lyapunov equation can be rewritten as
$  \effOuBlkDiag\tilde{\effOuStatVar}
  +
  \tilde{\effOuStatVar}\effOuBlkDiag^{\effOuTranspose}
  =
  -\tilde{\effOuDiffMat}$.
Since $\effOuBlkDiag$ is block diagonal, this equation reduces to a collection of low-dimensional Sylvester equations for each block pair indexed by $(\effOuIdxB, \effOuIdxBTwo)$:
\begin{align}
  \effOuBlk{\effOuIdxB}
  \tilde{\effOuStatVar}_{\effOuIdxB\effOuIdxBTwo}
  +
  \tilde{\effOuStatVar}_{\effOuIdxB\effOuIdxBTwo}
  \effOuBlk{\effOuIdxBTwo}^{\effOuTranspose}
  =
  -\tilde{\effOuDiffMat}_{\effOuIdxB\effOuIdxBTwo}.
  \label{suppeffOu.eq:block_forward_lyap}
\end{align}
Only the upper-triangular block-pair equations need to be solved, since
\(\tilde{\effOuDiffMat}\) is symmetric.

\subsubsection{A reusable Sylvester inverse}
\label{suppeffOu.sylvester_inverse}

The forward and adjoint Lyapunov computations both reduce to
\begin{align}
  \effOuBlk{\effOuIdxB}\effOuLyapunovSolution+\effOuLyapunovSolution\effOuBlk{\effOuIdxBTwo}=\mathbf{C},
  \label{suppeffOu.eq:sylvester_kernel}
\end{align}
where \(\effOuBlk{\effOuIdxB}\) and \(\effOuBlk{\effOuIdxBTwo}\) are blocks of size at most two.
For scalar \(\effOuBlk{\effOuIdxB}=\effOuBlockScalarOne\) and
\(\effOuBlk{\effOuIdxBTwo}=\effOuBlockScalarTwo\), the solution is 
\(\effOuLyapunovSolution=\mathbf{C}/(\effOuBlockScalarOne+\effOuBlockScalarTwo)\).
For \(\effOuBlk{\effOuIdxB}=\effOuBlockScalarOne\) scalar and
\(\effOuBlk{\effOuIdxBTwo}=\effOuBlockScalarTwo\effOuIdentity+\effOuBlkN{\effOuIdxBTwo}\), with
\(\effOuBlkN{\effOuIdxBTwo}^2=\effOuBlkDelta{\effOuIdxBTwo}\effOuIdentity\),
\begin{align}
  \effOuLyapunovSolution
  =
  \mathbf{C}\{(\effOuBlockScalarOne+\effOuBlockScalarTwo)\effOuIdentity+\effOuBlkN{\effOuIdxBTwo}\}^{-1}
  =
  \frac{(\effOuBlockScalarOne+\effOuBlockScalarTwo)\mathbf{C}-\mathbf{C}\effOuBlkN{\effOuIdxBTwo}}
       {(\effOuBlockScalarOne+\effOuBlockScalarTwo)^2-\effOuBlkDelta{\effOuIdxBTwo}}.
\end{align}
For \(\effOuBlk{\effOuIdxB}=\effOuBlockScalarOne\effOuIdentity+\effOuBlkN{\effOuIdxB}\) and
\(\effOuBlk{\effOuIdxBTwo}=\effOuBlockScalarTwo\) scalar,
\begin{align}
  \effOuLyapunovSolution
  =
  \{(\effOuBlockScalarOne+\effOuBlockScalarTwo)\effOuIdentity+\effOuBlkN{\effOuIdxB}\}^{-1}\mathbf{C}
  =
  \frac{(\effOuBlockScalarOne+\effOuBlockScalarTwo)\mathbf{C}-\effOuBlkN{\effOuIdxB}\mathbf{C}}
       {(\effOuBlockScalarOne+\effOuBlockScalarTwo)^2-\effOuBlkDelta{\effOuIdxB}}.
\end{align}
The next subsection deals with the remaining combination where both blocks are two-dimensional.

\subsubsection{Two-dimensional--two-dimensional block pair}
\label{suppeffOu.twobytwo_sylvester}

Let
\(\effOuBlk{\effOuIdxB}=\effOuBlockScalarOne\effOuIdentity+\effOuBlkN{\effOuIdxB}\),
\(\effOuBlk{\effOuIdxBTwo}=\effOuBlockScalarTwo\effOuIdentity+\effOuBlkN{\effOuIdxBTwo}\), with
\(\effOuBlkN{\effOuIdxB}^2=\effOuBlkDelta{\effOuIdxB}\effOuIdentity\),
\(\effOuBlkN{\effOuIdxBTwo}^2=\effOuBlkDelta{\effOuIdxBTwo}\effOuIdentity\), and
\(\omegaAB=\effOuBlockScalarOne+\effOuBlockScalarTwo\).
Define left and right multiplication operators
\(\mathcal L(\mathbf C)=\effOuBlkN{\effOuIdxB}\mathbf C\) and
\(\mathcal R(\mathbf C)=\mathbf C\effOuBlkN{\effOuIdxBTwo}\).
They commute and satisfy
\(\mathcal L^2=\effOuBlkDelta{\effOuIdxB}\mathcal I\) and
\(\mathcal R^2=\effOuBlkDelta{\effOuIdxBTwo}\mathcal I\), so the algebra they generate is contained in
\(\operatorname{span}\{\mathcal I,\mathcal L,\mathcal R,\mathcal L\mathcal R\}\).
We therefore seek the inverse Sylvester action on \(\mathbf C\) in the form
\begin{align}
  \effOuLyapunovSolution
  =
  \beta_1 \mathbf{C}
  +
  \beta_2 \effOuBlkN{\effOuIdxB}\mathbf{C}
  +
  \beta_3 \mathbf{C}\effOuBlkN{\effOuIdxBTwo}
  +
  \beta_4 \effOuBlkN{\effOuIdxB}\mathbf{C}\effOuBlkN{\effOuIdxBTwo}
  \label{suppeffOu.eq:sylvester_span}
\end{align}
which is equivalent to applying the operator
\(\beta_1\mathcal I+\beta_2\mathcal L+\beta_3\mathcal R+\beta_4\mathcal L\mathcal R\) to \(\mathbf C\).
Substituting into
\(\effOuBlk{\effOuIdxB}\effOuLyapunovSolution+\effOuLyapunovSolution\effOuBlk{\effOuIdxBTwo}=\mathbf{C}\)
and expanding with
\(\effOuBlk{\effOuIdxB}\effOuLyapunovSolution+\effOuLyapunovSolution\effOuBlk{\effOuIdxBTwo}
=\omegaAB\effOuLyapunovSolution
+\effOuBlkN{\effOuIdxB}\effOuLyapunovSolution
+\effOuLyapunovSolution\effOuBlkN{\effOuIdxBTwo}\)
gives
\begin{align*}
  \effOuBlk{\effOuIdxB}\effOuLyapunovSolution+\effOuLyapunovSolution\effOuBlk{\effOuIdxBTwo}
  &= (\omegaAB\beta_1 + \effOuBlkDelta{\effOuIdxB}\beta_2 + \effOuBlkDelta{\effOuIdxBTwo}\beta_3)\,\mathbf{C}
   + (\omegaAB\beta_2 + \beta_1 + \effOuBlkDelta{\effOuIdxBTwo}\beta_4)\,\effOuBlkN{\effOuIdxB}\mathbf{C} \\
  &\quad
   + (\omegaAB\beta_3 + \beta_1 + \effOuBlkDelta{\effOuIdxB}\beta_4)\,\mathbf{C}\effOuBlkN{\effOuIdxBTwo}
   + (\omegaAB\beta_4 + \beta_2 + \beta_3)\,\effOuBlkN{\effOuIdxB}\mathbf{C}\effOuBlkN{\effOuIdxBTwo},
\end{align*}
where \(\effOuBlkN{\effOuIdxB}^2=\effOuBlkDelta{\effOuIdxB}\effOuIdentity\) and
\(\effOuBlkN{\effOuIdxBTwo}^2=\effOuBlkDelta{\effOuIdxBTwo}\effOuIdentity\) reduce all higher-order operator products back into
\(\operatorname{span}\{\mathcal I,\mathcal L,\mathcal R,\mathcal L\mathcal R\}\).
Requiring this operator identity to act as the identity on \(\mathbf C\) gives the scalar system
\begin{align}
  \begin{pmatrix}
    \omegaAB & \effOuBlkDelta{\effOuIdxB} & \effOuBlkDelta{\effOuIdxBTwo} & 0 \\
    1 & \omegaAB & 0 & \effOuBlkDelta{\effOuIdxBTwo} \\
    1 & 0 & \omegaAB & \effOuBlkDelta{\effOuIdxB} \\
    0 & 1 & 1 & \omegaAB
  \end{pmatrix}
  \begin{pmatrix}
    \beta_1\\ \beta_2\\ \beta_3\\ \beta_4
  \end{pmatrix}
  =
  \begin{pmatrix}
    1\\0\\0\\0
  \end{pmatrix}.
  \label{suppeffOu.eq:sylvester_coeff_system}
\end{align}
Equivalently, provided the Sylvester operator is nonsingular,
\begin{align}
  \beta_1
  &= \phantom{-}
  \frac{\omegaAB(\omegaAB^2-\effOuBlkDelta{\effOuIdxB}-\effOuBlkDelta{\effOuIdxBTwo})}
       {(\omegaAB^2-\effOuBlkDelta{\effOuIdxB}-\effOuBlkDelta{\effOuIdxBTwo})^2-4\effOuBlkDelta{\effOuIdxB}\effOuBlkDelta{\effOuIdxBTwo}},
         &\beta_2
  =
  -\frac{\omegaAB^2-\effOuBlkDelta{\effOuIdxB}+\effOuBlkDelta{\effOuIdxBTwo}}
        {(\omegaAB^2-\effOuBlkDelta{\effOuIdxB}-\effOuBlkDelta{\effOuIdxBTwo})^2-4\effOuBlkDelta{\effOuIdxB}\effOuBlkDelta{\effOuIdxBTwo}}, \\
  \beta_3
  &=
  -\frac{\omegaAB^2+\effOuBlkDelta{\effOuIdxB}-\effOuBlkDelta{\effOuIdxBTwo}}
        {(\omegaAB^2-\effOuBlkDelta{\effOuIdxB}-\effOuBlkDelta{\effOuIdxBTwo})^2-4\effOuBlkDelta{\effOuIdxB}\effOuBlkDelta{\effOuIdxBTwo}},
  &\beta_4
  = \phantom{-}
  \frac{2\omegaAB}
       {(\omegaAB^2-\effOuBlkDelta{\effOuIdxB}-\effOuBlkDelta{\effOuIdxBTwo})^2-4\effOuBlkDelta{\effOuIdxB}\effOuBlkDelta{\effOuIdxBTwo}}.
\end{align}
For the stationary covariance block
\eqref{suppeffOu.eq:block_forward_lyap}, use
\(\effOuBlk{\effOuIdxB}\), \(\effOuBlk{\effOuIdxBTwo}^{\effOuTranspose}\), and
\(\mathbf{C}=-\tilde{\effOuDiffMat}_{\effOuIdxB\effOuIdxBTwo}\).

\paragraph{Uniqueness and nonsingularity}
The Hurwitz condition ensures that the global Lyapunov operator is nonsingular,
and hence that the stationary Lyapunov equation has a \emph{unique} solution.
The denominator in the coefficient formulas vanishes exactly when the
Sylvester operator
\(\effOuLyapunovSolution\mapsto \effOuBlk{\effOuIdxB}\effOuLyapunovSolution
+\effOuLyapunovSolution\effOuBlk{\effOuIdxBTwo}\)
is singular.
Equivalently, this occurs when
\(\lambda+\lambda'=0\) for some
\(\lambda\in\operatorname{spec}(\effOuBlk{\effOuIdxB})\) and
\(\lambda'\in\operatorname{spec}(\effOuBlk{\effOuIdxBTwo})\).
In the stationary Lyapunov applications considered here, the blocks are Hurwitz and so all such sums have strictly negative real part, so the block-pair
Sylvester operator is nonsingular.

\paragraph{Cached equal-diagonal form}
For implementation, we specialize the reusable Sylvester kernel to the forward Lyapunov equation, whose right block is transposed.
For a pair of two-dimensional blocks, write
\begin{align}
\effOuBlk{\effOuIdxB}
&=
\begin{pmatrix}
\effOuEqualDiagOne & \effOuUpperRightOne\\
\effOuLowerLeftOne & \effOuEqualDiagOne
\end{pmatrix},
\qquad
\effOuBlk{\effOuIdxBTwo}
=
\begin{pmatrix}
\effOuEqualDiagTwo & \effOuUpperRightTwo\\
\effOuLowerLeftTwo & \effOuEqualDiagTwo
\end{pmatrix}.
\notag
\end{align}
For an off-diagonal block pair \(\effOuIdxB<\effOuIdxBTwo\), let
\(\effOuLyapunovWeight=(\effOuLyapunovWeightEntry{rs})\) and
\(\effOuLyapunovSolution=(\effOuLyapunovSolutionEntry{rs})\), and consider
\begin{align}
\effOuBlk{\effOuIdxB}\effOuLyapunovSolution
+
\effOuLyapunovSolution\effOuBlk{\effOuIdxBTwo}^{\effOuTranspose}
&=
\effOuLyapunovWeight.
\label{suppeffOu.eq:cached_equal_diag_equation}
\end{align}
Using the row-wise vectorizations
\begin{align}
\effOuLyapunovWeightColumn
&=
\begin{pmatrix}
\effOuLyapunovWeightEntry{00}\\
\effOuLyapunovWeightEntry{01}\\
\effOuLyapunovWeightEntry{10}\\
\effOuLyapunovWeightEntry{11}
\end{pmatrix},
\qquad
\effOuLyapunovSolutionColumn
=
\begin{pmatrix}
\effOuLyapunovSolutionEntry{00}\\
\effOuLyapunovSolutionEntry{01}\\
\effOuLyapunovSolutionEntry{10}\\
\effOuLyapunovSolutionEntry{11}
\end{pmatrix},
\end{align}
Equation~\eqref{suppeffOu.eq:cached_equal_diag_equation} becomes
\begin{align}
\begin{pmatrix}
\omegaAB
&
\effOuUpperRightTwo
&
\effOuUpperRightOne
&
0
\\
\effOuLowerLeftTwo
&
\omegaAB
&
0
&
\effOuUpperRightOne
\\
\effOuLowerLeftOne
&
0
&
\omegaAB
&
\effOuUpperRightTwo
\\
0
&
\effOuLowerLeftOne
&
\effOuLowerLeftTwo
&
\omegaAB
\end{pmatrix}
\effOuLyapunovSolutionColumn
&=
\effOuLyapunovWeightColumn,
\qquad
\omegaAB
=
\effOuEqualDiagOne+\; \effOuEqualDiagTwo.
\label{suppeffOu.eq:cached_equal_diag_system}
\end{align}
The placement of \(\effOuUpperRightTwo\) and \(\effOuLowerLeftTwo\) in this system follows from the transpose on the second block.
The inverse of the coefficient matrix in~\eqref{suppeffOu.eq:cached_equal_diag_system} admits a convenient factorization.
Define
\begin{align}
c_1
&=
\omegaAB^2
+
\effOuUpperRightTwo\effOuLowerLeftTwo
-
\effOuUpperRightOne\effOuLowerLeftOne,
\\
c_2
&=
2\omegaAB\effOuUpperRightTwo,
\qquad
c_3
=
2\omegaAB\effOuLowerLeftTwo,
\\
c_0
&=
c_1^2-c_2c_3.
\end{align}
The scalar \(c_0\) is the determinant of the coefficient matrix in~\eqref{suppeffOu.eq:cached_equal_diag_system}.
Equivalently,
\begin{align}
c_0
&=
\left(
\omegaAB^2
-
\effOuBlkDelta{\effOuIdxB}
-
\effOuBlkDelta{\effOuIdxBTwo}
\right)^2
-
4\effOuBlkDelta{\effOuIdxB}
\effOuBlkDelta{\effOuIdxBTwo},
\end{align}
where
\(\effOuBlkDelta{\effOuIdxB}
=
\effOuUpperRightOne\effOuLowerLeftOne\) and
\(\effOuBlkDelta{\effOuIdxBTwo}
=
\effOuUpperRightTwo\effOuLowerLeftTwo\).
Thus \(c_0\neq0\) precisely when the corresponding block-pair Sylvester operator is nonsingular.
The solution can be evaluated by first computing
\begin{align}
\mathbf{b}
&=
\begin{pmatrix}
\omegaAB
&
\effOuUpperRightTwo
&
-\effOuUpperRightOne
&
0
\\
\effOuLowerLeftTwo
&
\omegaAB
&
0
&
-\effOuUpperRightOne
\\
-\effOuLowerLeftOne
&
0
&
\omegaAB
&
\effOuUpperRightTwo
\\
0
&
-\effOuLowerLeftOne
&
\effOuLowerLeftTwo
&
\omegaAB
\end{pmatrix}
\effOuLyapunovWeightColumn,
\end{align}
and then applying
\begin{align}
\effOuLyapunovSolutionColumn
&=
\frac{1}{c_0}
\begin{pmatrix}
c_1 & -c_2 & 0 & 0\\
-c_3 & c_1 & 0 & 0\\
0 & 0 & c_1 & -c_2\\
0 & 0 & -c_3 & c_1
\end{pmatrix}
\mathbf{b}.
\label{suppeffOu.eq:cached_equal_diag_offdiag_lyap}
\end{align}
For fixed drift blocks, both sparse factors in this expression, or equivalently their \(4\times4\) product, depend only on the block parameters and can therefore be cached and reused for different right-hand sides.
For a diagonal block pair with symmetric \(\effOuLyapunovWeight\), uniqueness of the Sylvester solution implies that \(\effOuLyapunovSolution\) is also symmetric.
Consequently, only
\((\effOuLyapunovSolutionEntry{00},
\effOuLyapunovSolutionEntry{01},
\effOuLyapunovSolutionEntry{11})\)
must be computed.
The reduced system is
\begin{align}
\begin{pmatrix}
2\effOuEqualDiagOne
&
2\effOuUpperRightOne
&
0
\\
\effOuLowerLeftOne
&
2\effOuEqualDiagOne
&
\effOuUpperRightOne
\\
0
&
2\effOuLowerLeftOne
&
2\effOuEqualDiagOne
\end{pmatrix}
\begin{pmatrix}
\effOuLyapunovSolutionEntry{00}\\
\effOuLyapunovSolutionEntry{01}\\
\effOuLyapunovSolutionEntry{11}
\end{pmatrix}
&=
\begin{pmatrix}
\effOuLyapunovWeightEntry{00}\\
\effOuLyapunovWeightEntry{01}\\
\effOuLyapunovWeightEntry{11}
\end{pmatrix}.
\label{suppeffOu.eq:cached_equal_diag_diag_lyap}
\end{align}
A fixed block drift therefore defines a small reusable solve plan: one \(3\times3\) inverse map for each diagonal block and one \(4\times4\) inverse map for each off-diagonal block pair.
When the full right-hand side is symmetric, the solve is required only for
\(\effOuIdxB\leq\effOuIdxBTwo\), and the remaining blocks are recovered as
\begin{align}
\effOuLyapunovSolution_{\effOuIdxBTwo\effOuIdxB}
&=
\effOuLyapunovSolution_{\effOuIdxB\effOuIdxBTwo}^{\effOuTranspose}.
\end{align}

\section{Fr\'echet Derivative and Adjoint Block Kernels}
\subsection{Matrix exponential block Fr\'echet derivatives}
\label{suppeffOu.forward_frechet}

For a perturbation \(\effOuFrechetDirection\), the Fr\'echet derivative of
\(\effOuBlkDiag\mapsto e^{\effOuTime\effOuBlkDiag}\) can be written using the following integral representation \citep{najfeld_derivatives_1995}
\begin{align}
  \DexpBlk[\effOuFrechetDirection]
  =
  \effOuTime\int_0^1
  e^{(1-s)\effOuTime\effOuBlkDiag}\effOuFrechetDirection\,e^{s\effOuTime\effOuBlkDiag}\dd s.
\end{align}
Because \(\effOuBlkDiag\) is block diagonal, we can evaluate this integral block by block \citep{monti_nonparametric_2025}:
\begin{align}
  [\DexpBlk[\effOuFrechetDirection]]_{\effOuIdxB\effOuIdxBTwo}
  =
  \effOuTime\int_0^1
  e^{(1-s)\effOuTime\effOuBlk{\effOuIdxB}}
  {\effOuFrechetDirection}_{\effOuIdxB\effOuIdxBTwo}
  e^{s\effOuTime\effOuBlk{\effOuIdxBTwo}}\dd s.
  \label{suppeffOu.eq:block_frechet}
\end{align}
Thus the derivative is a collection of independent block-pair kernels.

As a prerequisite for evaluating these kernels, we first define the \emph{divided difference} function as
\begin{align}
  \varphi(\effOusuppDummyIndexOne,\effOusuppDummyIndexTwo)
  =
  \int_0^1
  e^{(1-s)\effOusuppDummyIndexOne}
  e^{s\effOusuppDummyIndexTwo}\dd s
  =
  \begin{cases}
    \dfrac{e^{\effOusuppDummyIndexOne}-e^{\effOusuppDummyIndexTwo}}
      {\effOusuppDummyIndexOne-\effOusuppDummyIndexTwo},
      & \effOusuppDummyIndexOne\ne \effOusuppDummyIndexTwo,\\[8pt]
    e^{\effOusuppDummyIndexOne},
      & \effOusuppDummyIndexOne=\effOusuppDummyIndexTwo.
  \end{cases}
\end{align}
For numerical evaluation near coincident arguments, we use the equivalent
\(\varphi(x,y)=e^y\operatorname{exprel}(x-y)\), where
\(\operatorname{exprel}(z)=\operatorname{expm1}(z)/z\) with limiting value \(1\) at \(z=0\).
We then provide explicit formulas for Equation~\eqref{suppeffOu.eq:block_frechet} when the two blocks are scalar, when one is scalar and one is two-dimensional, and when both are two-dimensional.
\paragraph{Scalar blocks}
For scalar blocks,
\(\effOuBlk{\effOuIdxB}=\effOuBlockScalarOne\) and
\(\effOuBlk{\effOuIdxBTwo}=\effOuBlockScalarTwo\), we simply have:
\begin{align}
[\DexpBlk[\effOuFrechetDirection]]_{\effOuIdxB\effOuIdxBTwo} = \effOuTime\,\varphi(\effOuTime\effOuBlockScalarOne,\effOuTime\effOuBlockScalarTwo)\,{\effOuFrechetDirection}_{\effOuIdxB\effOuIdxBTwo}.
\end{align}

\paragraph{Mixed block pairs}

If \(\effOuBlk{\effOuIdxB}=\effOuBlockScalarOne\) is scalar and
\(\effOuBlk{\effOuIdxBTwo}\) is two-dimensional, then
\begin{align}
  [\DexpBlk[\effOuFrechetDirection]]_{\effOuIdxB\effOuIdxBTwo}
  =
  \effOuTime\, e^{\effOuTime\effOuBlockScalarOne}
  {\effOuFrechetDirection}_{\effOuIdxB\effOuIdxBTwo}
  \phi_1\!\left(\effOuTime\effOuBlk{\effOuIdxBTwo}-\effOuTime\effOuBlockScalarOne\effOuIdentity\right),
\end{align}
where $\phi_1(\mathbf{M})=\int_0^1 e^{s\mathbf{M}}\dd s$.
Similarly, when the left block is two-dimensional and the right block is
scalar, we have
\begin{align}
  [\DexpBlk[\effOuFrechetDirection]]_{\effOuIdxB\effOuIdxBTwo}
  =
  \effOuTime\, e^{\effOuTime\effOuBlockScalarTwo}
  \phi_1\!\left(\effOuTime\effOuBlk{\effOuIdxB}-\effOuTime\effOuBlockScalarTwo\effOuIdentity\right)
  {\effOuFrechetDirection}_{\effOuIdxB\effOuIdxBTwo}.
\end{align}
The same two-dimensional exponential formula from
Section~\ref{suppeffOu.sec:forward_exp} evaluates \(\phi_1\), either by
\(\mathbf{M}^{-1}(e^{\mathbf{M}}-\effOuIdentity)\) when well conditioned or by its power series.

\paragraph{Two-dimensional--two-dimensional block pairs} 

Let
\(\effOuBlk{\effOuIdxB}=\effOuBlockScalarOne\effOuIdentity+\effOuBlkN{\effOuIdxB}\) and
\(\effOuBlk{\effOuIdxBTwo}=\effOuBlockScalarTwo\effOuIdentity+\effOuBlkN{\effOuIdxBTwo}\), with
\(\effOuBlkN{\effOuIdxB}^2=\effOuBlkDelta{\effOuIdxB}\effOuIdentity\) and
\(\effOuBlkN{\effOuIdxBTwo}^2=\effOuBlkDelta{\effOuIdxBTwo}\effOuIdentity\).
Substituting the single-block exponentials from
Section~\ref{suppeffOu.sec:forward_exp} into the block integral
\eqref{suppeffOu.eq:block_frechet} and using the closure
\(\effOuBlkN{\effOuIdxB}^2=\effOuBlkDelta{\effOuIdxB}\effOuIdentity\) and
\(\effOuBlkN{\effOuIdxBTwo}^2=\effOuBlkDelta{\effOuIdxBTwo}\effOuIdentity\)
shows that no products beyond left and right multiplication by the two hollow
residuals are needed.
With the convention that exponent \(0\) denotes the identity and exponent \(1\)
the corresponding residual, the block Fr\'echet kernel can therefore be written as
\begin{align}
  [\DexpBlk[\effOuFrechetDirection]]_{\effOuIdxB\effOuIdxBTwo}
  =
  \sum_{u,v\in\{0,1\}}
  f_{uv}^{\effOuIdxB\effOuIdxBTwo}(\effOuTime)\,
  \effOuBlkN{\effOuIdxB}^{u}
  \effOuFrechetDirection_{\effOuIdxB\effOuIdxBTwo}
  \effOuBlkN{\effOuIdxBTwo}^{v},
  \label{suppeffOu.eq:frechet_span}
\end{align}
where the stable coefficient formula is
\begin{align}
  f_{uv}^{\effOuIdxB\effOuIdxBTwo}(\effOuTime)
  =
  \effOuTime^{1+u+v}
  \sum_{i,j\ge 0}
  \frac{(\effOuTime^2\effOuBlkDelta{\effOuIdxB})^i}{(2i+u)!}
  \frac{(\effOuTime^2\effOuBlkDelta{\effOuIdxBTwo})^j}{(2j+v)!}
  M_{2i+u,\,2j+v}^{\effOuIdxB\effOuIdxBTwo}(\effOuTime).
  \label{suppeffOu.eq:fuv_series}
\end{align}
The moment in \eqref{suppeffOu.eq:fuv_series} is
\begin{align}
  M_{m,n}^{\effOuIdxB\effOuIdxBTwo}(\effOuTime)
  =
  \int_0^1
  (1-s)^m s^n
  \exp\!\left[
    \effOuTime\{(1-s)\effOuBlockScalarOne+s\effOuBlockScalarTwo\}
  \right]\dd s.
  \label{suppeffOu.eq:mn_moment}
\end{align}
This single formula covers real-pair, repeated, and oscillatory blocks.
For exact repeated blocks, set the corresponding \(\effOuBlkDelta{}\) to zero;
then all higher powers in that block vanish automatically.
If both blocks are repeated, we have
\begin{align}
  f_{uv}^{\effOuIdxB\effOuIdxBTwo}(\effOuTime)
  &=
  \effOuTime^{1+u+v}
  M_{u,v}^{\effOuIdxB\effOuIdxBTwo}(\effOuTime),
  \qquad u,v\in\{0,1\}.
  \notag
\end{align}

Away from repeated blocks, the same coefficient can also be evaluated through a
four-point divided difference when the square-root gaps are well conditioned:
\begin{align}
  f_{uv}^{\effOuIdxB\effOuIdxBTwo}(\effOuTime)
  =
  \frac{\effOuTime}
       {4(\sqrt{\effOuBlkDelta{\effOuIdxB}})^u
          (\sqrt{\effOuBlkDelta{\effOuIdxBTwo}})^v}
  \sum_{\sigma,\tau\in\{\pm1\}}
  \sigma^u\tau^v
  \varphi\!\left(
    \effOuTime(\effOuBlockScalarOne+\sigma\sqrt{\effOuBlkDelta{\effOuIdxB}}),
    \effOuTime(\effOuBlockScalarTwo+\tau\sqrt{\effOuBlkDelta{\effOuIdxBTwo}})
  \right).
  \label{suppeffOu.eq:fuv_divided_difference}
\end{align}
Direct evaluation of \eqref{suppeffOu.eq:fuv_divided_difference} can suffer from
cancellation when one or both blocks are close to repeated.
In implementation, \eqref{suppeffOu.eq:fuv_series} is used whenever
\(|\effOuTime^2\effOuBlkDelta{\effOuIdxB}|\) or
\(|\effOuTime^2\effOuBlkDelta{\effOuIdxBTwo}|\) is small, truncating the double
series once the remaining terms are below the desired numerical tolerance.

\paragraph{Stable moment evaluation}
\label{suppeffOu.stable_divided_difference}

The moments \(M_{m,n}^{\effOuIdxB\effOuIdxBTwo}\) in
\eqref{suppeffOu.eq:fuv_series} are reduced to
one-dimensional exponential moments.
Define
\begin{align}
  m_k(c)=\int_0^1 s^k e^{sc}\dd s.
\end{align}
Let \(a=\effOuTime\effOuBlockScalarOne\) and \(b=\effOuTime\effOuBlockScalarTwo\).
When \(a\ge b\), factor out \(e^a\); since
\((1-s)^m=\sum_{h=0}^{m}(-1)^h\binom{m}{h}s^h\),
\begin{align}
  M_{m,n}^{\effOuIdxB\effOuIdxBTwo}(\effOuTime)
  =
  e^a
  \sum_{h=0}^{m}
  (-1)^h
  \binom{m}{h}
  m_{n+h}(b-a).
  \label{suppeffOu.eq:mn_from_mk}
\end{align}
When \(b>a\), applying the change of variables \(s\mapsto1-s\) instead gives
\begin{align}
  M_{m,n}^{\effOuIdxB\effOuIdxBTwo}(\effOuTime)
  =
  e^b
  \sum_{h=0}^{n}
  (-1)^h
  \binom{n}{h}
  m_{m+h}(a-b).
\end{align}
The residual exponential argument is therefore always nonpositive in Hurwitz H-SSBP applications.
The repeated-eigenvalue limits require only the scalar moments \(m_k\),
not finite differences of nearly equal quantities.
Away from zero, evaluate \(m_0(c)=\mathrm{expm1}(c)/c\) to avoid cancellation
and then use the recurrence
$m_k(c)=e^c/c-(k/c)m_{k-1}(c)$.
When \(|c|\) is below the chosen tolerance, use the convergent series
\begin{align}
  m_k(c)
  =
  \sum_{n=0}^{\infty}
  \frac{c^n}{n!(n+k+1)}.
\end{align}

\subsection{Adjoint block kernels for matrix exponentials}
\label{suppeffOu.adjoint_exp}

From this point onward, adjoints are taken with respect to the Frobenius pairing
\(\langle \mathbf{X},\mathbf{Y}\rangle =
\operatorname{tr}(\mathbf{X}^{\effOuTranspose}\mathbf{Y})\).
Let \(\effOuAdjointVarSupp{}\) be the adjoint seed with respect to \(e^{\effOuTime\effOuBlkDiag}\).
The adjoint of the Fr\'echet derivative is defined by
$\langle \DexpBlk[\effOuFrechetDirection],\effOuAdjointVarSupp{}\rangle = \langle \effOuFrechetDirection,\DexpBlk^*[\effOuAdjointVarSupp{}]\rangle$.
Using cyclicity of the trace in the integral representation gives
\begin{align}
  \DexpBlk^*[\effOuAdjointVarSupp{}]
  =
  \effOuTime\int_0^1
  e^{s\effOuTime\effOuBlkDiag^{\effOuTranspose}}
  \effOuAdjointVarSupp{}
  e^{(1-s)\effOuTime\effOuBlkDiag^{\effOuTranspose}}
  \dd s.
  \label{suppeffOu.eq:adjoint_frechet}
\end{align}
Therefore the adjoint exponential uses the same kernels as in
Section~\ref{suppeffOu.forward_frechet}, with \(\effOuBlkDiag\) replaced by
\(\effOuBlkDiag^{\effOuTranspose}\).
For $\effOuActual{}=\effOuRotMat e^{\effOuTime\effOuBlkDiag}\effOuRotMat^{-1}$,
an incoming adjoint \(\effOuAdjointVarSupp{A}\) with respect to \(\effOuActual{}\) induces the
block-basis seed $\effOuAdjointVarSuppRotated{E} = \effOuRotMat^{\effOuTranspose} \effOuAdjointVarSupp{A} \effOuRotMat^{-\effOuTranspose}$.
The block-basis adjoint with respect to \(\effOuBlkDiag\) is then
$\effOuAdjointVarSupp{B}^{(\exp)} = \DexpBlk^*[\effOuAdjointVarSuppRotated{E}]$,
the quantity whose diagonal blocks are contracted with the native H-SSBP
block parameters.

\subsection{Adjoint Lyapunov block solves}
\label{suppeffOu.adjoint_lyap}

Suppose the log-likelihood contributes a seed \(\effOuAdjointVarSuppRotated{\effOuStatVarIndex}\) with respect to
\(\tilde{\effOuStatVar}\), so that $\dd\effOuLogLikelihood = \langle \effOuAdjointVarSuppRotated{\effOuStatVarIndex},\dd\tilde{\effOuStatVar}\rangle$.
Differentiating the Lyapunov equation in the block basis
$  \effOuBlkDiag\tilde{\effOuStatVar}
  +
  \tilde{\effOuStatVar}\effOuBlkDiag^{\effOuTranspose}
  =
  -\tilde{\effOuDiffMat}$
at fixed \(\tilde{\effOuDiffMat}\) gives
\begin{align}
  \effOuBlkDiag\dd\tilde{\effOuStatVar}
  +
  \dd\tilde{\effOuStatVar}\effOuBlkDiag^{\effOuTranspose}
  =
  -(\dd\effOuBlkDiag)\tilde{\effOuStatVar}
  -
  \tilde{\effOuStatVar}(\dd\effOuBlkDiag)^{\effOuTranspose}.
\end{align}
Define the block-basis Lyapunov adjoint \(\effOuLyapAdjRot\) and the induced block-drift adjoint by
\begin{align}
  \effOuBlkDiag^{\effOuTranspose}\effOuLyapAdjRot
  +\effOuLyapAdjRot\effOuBlkDiag
  &=
  -\effOuAdjointVarSuppRotated{\effOuStatVarIndex},
  \label{suppeffOu.eq:adjoint_lyap}\\
  \effOuAdjointVarSupp{B}^{(\mathrm{Lyap})}
  &=
  \effOuLyapAdjRot\tilde{\effOuStatVar}
  +
  \effOuLyapAdjRot^{\effOuTranspose}\tilde{\effOuStatVar}.
  \label{suppeffOu.eq:lyap_grad_D}
\end{align}
Here the second line uses symmetry of \(\tilde{\effOuStatVar}\) and follows by
cyclically pairing the differential above with \(\effOuLyapAdjRot\).

\paragraph{Block-pair adjoint solve}
Equation~\eqref{suppeffOu.eq:adjoint_lyap} decomposes over block pairs as
\begin{align}
  \effOuBlk{\effOuIdxB}^{\effOuTranspose}
  \effOuLyapAdjRot_{\effOuIdxB\effOuIdxBTwo}
  +
  \effOuLyapAdjRot_{\effOuIdxB\effOuIdxBTwo}
  \effOuBlk{\effOuIdxBTwo}
  &=
  -(\effOuAdjointVarSuppRotated{\effOuStatVarIndex})_{\effOuIdxB\effOuIdxBTwo},
  \notag
\end{align}
which is the same Sylvester kernel as in Section~\ref{suppeffOu.sylvester_inverse}.
In implementation, this is the cached equal-diagonal solve
\eqref{suppeffOu.eq:cached_equal_diag_offdiag_lyap}--\eqref{suppeffOu.eq:cached_equal_diag_diag_lyap}
with the first block's off-diagonal entries swapped, since
\(\effOuBlk{\effOuIdxB}^{\effOuTranspose}\) exchanges
\((\effOuUpperRightOne,\effOuLowerLeftOne)\).
For a symmetric seed, only the upper block-pair equations are solved, and the resulting blocks are then mirrored.
After the adjoint solve, the Lyapunov contribution to the block-drift gradient
is obtained from \eqref{suppeffOu.eq:lyap_grad_D} and then contracted with the native
H-SSBP block parameters.

\paragraph{Diffusion pullback}
If the diffusion matrix is also parametrized, the same adjoint variable \(\effOuLyapAdjRot\)
gives the seed with respect to \(\tilde{\effOuDiffMat}\).
When \(\effOuBlkDiag\) is fixed,
\(\effOuBlkDiag\dd\tilde{\effOuStatVar}
+\dd\tilde{\effOuStatVar}\effOuBlkDiag^{\effOuTranspose}
=-\dd\tilde{\effOuDiffMat}\), so the convention
\eqref{suppeffOu.eq:adjoint_lyap} yields
$\effOuAdjointVarSuppRotated{\effOuDiffMatIndex} = \effOuLyapAdjRot$
and
$\dd\effOuLogLikelihood =
\langle \effOuLyapAdjRot,\dd\tilde{\effOuDiffMat}\rangle$.
This seed is then pulled back through
\(\tilde{\effOuDiffMat}
=\effOuRotMat^{-1}\effOuDiffMat\effOuRotMat^{-\effOuTranspose}\).

\section{Supplementary Material for Numerical Robustness}
\label{suppeffOu.sec:numerical_robustness}

This section reports numerical-accuracy stress tests for the drift-dependent kernels used by the OU likelihood and its adjoints.
The tests compare the H-SSBP with the Lyapunov-based parametrization.
The tests were implemented in \texttt{Julia}~1.10.3 with ILP64 OpenBLAS/LAPACK and one BLAS thread.
Each reported error compares a double-precision \texttt{Float64} computation with a 256-bit \texttt{BigFloat} reference using relative Frobenius-norm error.
In the benchmarks below, we consider a homogeneous OU process evolving along an edge of length $0.2$.
Except for the first real--complex boundary benchmark, the process is set to have dimension \(\effOuDim=8\).

For the H-SSBP, we stress the repeated-root boundary between real and complex block eigenvalues and the conditioning of the change-of-basis matrix $\effOuRotMat$ when $\effOuRotMat$ is not orthogonal.
The repeated-root experiment uses one \(2\times2\) block with diagonal entry \(-2\), upper off-diagonal entry \(1\), and lower off-diagonal entry \(\pm d\), where \(d\in\{10^{-16},10^{-15},\ldots,10^{0}\}\).
Errors in this panel are medians over \(16\) random seeds and over the real and complex sides of the boundary.
The basis-conditioning experiment uses random bases with \(\kappa(\effOuRotMat)\in\{10^{0},10^{2},\ldots,10^{10}\}\).
For the Lyapunov-based parametrization, we stress the conditioning of the prescribed stationary covariance $\effOuStatVar$ and the magnitude of the skew-symmetric component $\effOuSkewSymmatrixMat$ in $\effOuDriftMatrix=\left(-\tfrac{1}{2}\effOuDiffMat+\effOuSkewSymmatrixMat\right)\effOuStatVar^{-1}$.
These panels use \(\kappa(\effOuStatVar)\in\{10^{0},10^{2},\ldots,10^{10}\}\), and skew strengths over \(\{10^{-6},10^{-4},\ldots,10^{6}\}\).

Figure~\ref{suppeffOu.fig:ou_ssbp_vs_lyapunov_param_stability_summary} shows that the H-SSBP block formulas remain accurate at the real--complex transition: exponential, Lyapunov, covariance, and adjoint errors stay close to the high-precision reference as the block discriminant approaches zero.
The main H-SSBP sensitivity is instead the conditioning of $\effOuRotMat$, which affects dense reconstruction and basis-change operations.
The block-basis kernels themselves remain stable.
The Lyapunov-based parametrization shows the corresponding behavior for its own dense construction: errors increase when $\effOuStatVar$ is ill-conditioned or when the skew-symmetric component is large.
Together, these panels show that the numerical loss of accuracy is substantially stronger for the Lyapunov-based parametrization, confirming that the H-SSBP improves both speed and numerical stability.
\begin{figure}[!ht]
\centering
\includegraphics[width=1\textwidth]{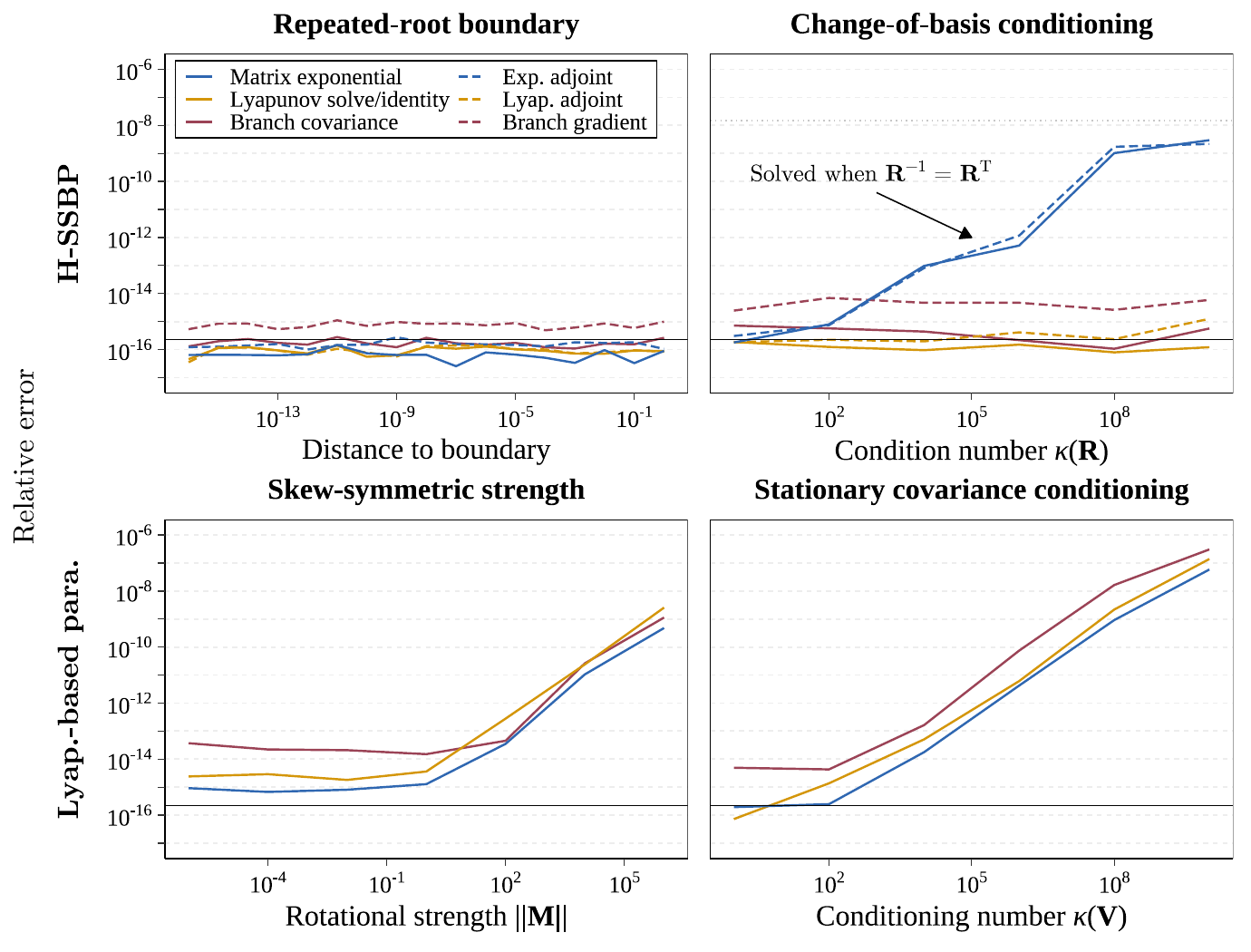}
\caption{
Numerical-stability stress tests for the H-SSBP and the Lyapunov-based parametrization.
For the H-SSBP, the panels vary the distance to the repeated-root boundary and the condition number $\kappa(\effOuRotMat)$ of the change-of-basis matrix.
For the Lyapunov-based parametrization, the panels vary the Frobenius norm of the skew-symmetric component $\effOuSkewSymmatrixMat$ and the condition number $\kappa(\effOuStatVar)$ of the prescribed stationary covariance.
Errors compare \texttt{Float64} computations with 256-bit \texttt{BigFloat} references.
The solid black line marks machine precision and the dotted gray line marks its square root.
}
\label{suppeffOu.fig:ou_ssbp_vs_lyapunov_param_stability_summary}
\end{figure}
\section{Supplementary Material for Simulations}
\label{suppeffOu.sec:simulation_material}

\paragraph{Simulation design}
We simulated data under a \(5\)-dimensional homogeneous Ornstein--Uhlenbeck process:
\begin{align}
\dd\effOuDataRv{}(\effOuTime)
&=
\effOuDriftMatrix
\left(
\effOuDataRv{}(\effOuTime)
-
\effOuEquilMean
\right)
\,\dd \effOuTime
+ \effOuDiffMat^{1/2}\,\dd\effOuWienerProcess{}(\effOuTime).
\end{align}
We used three grids of drift matrices: (1) an orthogonal real--complex boundary grid, (2) a diagonalizable nonnormal shear grid, and (3) a defective Jordan-coupling grid.
For each grid value in each experiment, we generated \(25\) independent trajectories observed at \(\effNObs=800\) time points with inter-observation gaps drawn independently from a uniform distribution on $[0,0.1]$.
All \(\effOuDim\) coordinates were observed at every time point, giving \(\effNObs\effOuDim=4000\) scalar observations per data set.
The three grids are defined below.

\paragraph{Common simulation settings}
The equilibrium mean was fixed at \(\effOuEquilMean=\mathbf 0\), and the initial state was drawn from the stationary distribution implied by \(\effOuDriftMatrix\) and \(\effOuDiffMat\).
The observation model was
\begin{align}
\effOuDataRv{m}^*
&=
\effOuDataRv{}(\effOuTime_m)+\boldsymbol{\epsilon}_m,
\qquad
\boldsymbol{\epsilon}_m
\overset{\mathrm{ind}}{\sim}
\effOuNorm\!\left(\mathbf 0,\mathbf \Sigma_{\mathrm{obs}}\right),
\end{align}
independently over \(m\).
The diffusion matrix and observation-error covariance matrix used to generate the data were
\begin{align}
\effOuDiffMat
&=
\diag(0.16,0.12,0.14,0.11,0.13),
\\
\mathbf \Sigma_{\mathrm{obs}}
&=
\diag(0.012,0.010,0.011,0.009,0.0105),
\end{align}
respectively.
In the MAP fits, the equilibrium mean, initial covariance, and observation-error covariance were fixed at their generating values.
The diffusion covariance was inferred through a Cholesky correlation representation, with independent \(\mathrm{Gamma}(0.5,0.5)\) shape--scale priors on the diagonal diffusion parameters and an LKJ\((1)\) prior on the correlation matrix.

\paragraph{Inference settings} For each simulated data set, we computed MAP estimates under two parametrizations of the drift: the orthogonal H-SSBP, \(\effOuDriftMatrix=\effOuRotMat\effOuBlkDiag\effOuRotMat^{\effOuTranspose}\), and the generic dense-basis H-SSBP, \(\effOuDriftMatrix=\effOuRotMat\effOuBlkDiag\effOuRotMat^{-1}\).
For both parametrizations, the unconstrained coordinates for the block parameters corresponding to the scalar rate, block-rate increments \(\tilde{\effOuCartBlockRate}_{\effOuIdxB}\), and oscillatory parameters \(\tilde{\effOuBlkT{}}_{\effOuIdxB}\) were assigned independent \(N(0,1)\) priors.
The coupling parameters \(\effOuCartNormalizingRate_{\effOuIdxB}\) were assigned independent uniform priors on \((-1,1)\).
The orthogonal H-SSBP used independent \(N(0,0.25^2)\) priors on the Givens angles.
For the dense-basis H-SSBP, \(\effOuRotMat\) was represented by its raw entries and regularized by the prior
\begin{align}
\log p(\effOuRotMat)
&=
C
+\lambda\log\left|\det\effOuRotMat\right| - \frac{\lambda \effOuDim}{2}\|\effOuRotMat \|_F^2,
\qquad
\lambda=0.1.
\end{align}

For each data set and fitted parametrization, the MAP optimization was run from five independent random initializations with a common budget of \(5000\) iterations.
Optimization used limited-memory BFGS (L-BFGS) \citep{liu_limited_1989} on the unconstrained inference coordinates with analytic log-posterior gradients.
Orthogonal starts used small Givens-angle perturbations, while dense-basis starts used a raw basis matrix \(\effOuRotMat=\effOuIdentity+\mathbf E\), with entries of \(\mathbf E\) drawn independently from a centered normal distribution with standard deviation \(0.05\), followed by a global normalization to achieve Frobenius norm equal to one.
The MAP estimate reported for a data set and parametrization is the optimization run with the largest final log posterior among the five starts.
We report the drift reconstruction error
\begin{align}
\mathrm{RMSE}
\left(
\widehat{\effOuDriftMatrix},
\effOuDriftMatrix
\right)
&=
\left\{
\frac{1}{\effOuDim^2}
\left\|
\widehat{\effOuDriftMatrix}
-
\effOuDriftMatrix
\right\|_F^2
\right\}^{1/2}.
\end{align}
Each point in Figure~2 of the main manuscript is the median RMSE across the \(25\) selected MAP estimates, with the interquartile range shown as a ribbon.
Optimization stabilization was assessed using the relative log-posterior improvement between the last two reported L-BFGS iterations.
A selected MAP estimate was treated as stabilized when this relative improvement was at most \(10^{-6}\).
All \(950\) selected MAP estimates in Figure~2 satisfied this rule.

\paragraph{Generating drift matrices} We now present the process used to generate the drift matrices for each of the three grids.
\begin{enumerate}
  \item 

The baseline orthogonal truth consisted of one scalar block and two \(2\times2\) H-SSBP blocks,
\begin{align}
\effOuBlk{\effOuIdxB}^{\mathrm{orth}}
&=
\begin{pmatrix}
\effOuCartBlockRate_{\effOuIdxB}
&
\effOuCartBlockRate_{\effOuIdxB}\effOuCartNormalizingRate_{\effOuIdxB}
+
\effOuBlkT{\effOuIdxB}
\\
\effOuCartBlockRate_{\effOuIdxB}\effOuCartNormalizingRate_{\effOuIdxB}
-
\effOuBlkT{\effOuIdxB}
&
\effOuCartBlockRate_{\effOuIdxB}
\end{pmatrix},
\qquad
\effOuDriftMatrix_{\mathrm{orth}}
=
\effOuRotMat_{\mathrm{orth}}
\operatorname{bdiag}
\left(
-0.72,
\effOuBlk{1}^{\mathrm{orth}},
\effOuBlk{2}^{\mathrm{orth}}
\right)
\effOuRotMat_{\mathrm{orth}}^{\effOuTranspose}.
\end{align}
The block parameters were \((\effOuCartBlockRate_{1},\effOuCartBlockRate_{2})=(-0.54604475,-0.89635746)\), \((\effOuCartNormalizingRate_{1},\effOuCartNormalizingRate_{2})=(0.12057934,-0.09024379)\), and \((\effOuBlkT{1},\effOuBlkT{2})=(0.24,0.16)\).
The orthogonal basis \(\effOuRotMat_{\mathrm{orth}}\) was generated by the fixed Givens-angle vector
\begin{align}
\boldsymbol{\effOuGivensAngles}_{\mathrm{orth}}
&=
(0.16,-0.12,0.08,0.05,-0.09,
0.11,-0.06,0.07,-0.04,0.10).
\end{align}
The first experiment stayed within this correctly specified orthogonal family and varied the rotational parameter $\effOuBlkT{1}$ of the first two-dimensional block.
We set
\begin{align}
&\effOuBlkT{1}
=
u\,\effOuBlkT{1}^*,
\qquad
\effOuBlkT{1}^*
=
\left|\effOuCartBlockRate_{1}\effOuCartNormalizingRate_{1}\right|,
\\
&u
\in
\{0.5,0.75,0.95,1.0,1.05,1.25,1.5\}.
\end{align}
The resulting \(2\times2\) block has two real eigenvalues for \(u<1\), a repeated real eigenvalue at \(u=1\), and a complex-conjugate pair for \(u>1\).
This experiment therefore isolates behavior near the real--complex eigenvalue boundary while leaving the orthogonal H-SSBP correctly specified.

\item The second experiment generated diagonalizable but nonnormal drifts by introducing shear in block coordinates.
For shear multiplier \(\effOuShear\), define
\begin{align}
\effOuDriftMatrix_0(\effOuShear)
&=
\operatorname{bdiag}
\left\{
\effOuBlk{1}^{(\effOuShear)},
\effOuBlk{2}^{(\effOuShear)},
-1.5
\right\},
\qquad
\effOuBlk{\effOuIdxB}^{(\effOuShear)}
=
\begin{pmatrix}
-r_{\effOuIdxB}
&
s_{\effOuIdxB}\effOuShear r_{\effOuIdxB}
\\
0
&
-2r_{\effOuIdxB}
\end{pmatrix},
\end{align}
where \(r_1=0.5\), \(r_2=1.0\), the signs \(s_{\effOuIdxB}\in\{-1,1\}\) were fixed by the simulation seed, and the shear multiplier was \(\effOuShear\in\{3,4,5,6,7,8\}\).
The observed-coordinate drift was obtained by applying a fixed dense orthogonal rotation,
\begin{align}
\effOuDriftMatrix_{\effOuShear}
&=
\effOuSchurRot
\effOuDriftMatrix_0(\effOuShear)
\effOuSchurRot^{\effOuTranspose}.
\end{align}
The dense-basis H-SSBP is correctly specified for this experiment, whereas the generating drifts are not representable by the orthogonal H-SSBP.

\item The third experiment used a larger defective block to create a controlled model-misspecification setting.
For coupling \(c\), the truth was
\begin{align}
\effOuDriftMatrix_{\mathrm J}(c)
&=
\effOuSchurRot_{\mathrm J}
\operatorname{bdiag}
\left\{-2,
-\effOuIdentity_4+c\effOuBlkN{\mathrm J}
\right\}
\effOuSchurRot_{\mathrm J}^{\effOuTranspose},
\end{align}
where \(\effOuBlkN{\mathrm J}\) has ones on the first superdiagonal and zeros elsewhere.
The coupling grid was
\begin{align}
c
&\in
\{0.25,0.5,0.75,1.0,1.25,1.5\}.
\end{align}
For \(c>0\), the drift contains a real Jordan chain of length four.
Neither fitted H-SSBP parametrization can represent such a block, so this experiment probes the effect of approximating a larger defective invariant subspace with scalar and two-dimensional blocks.
\end{enumerate}

\paragraph{Results}
Gold curves in Figure~2 of the main manuscript correspond to the orthogonal H-SSBP, and blue curves correspond to the dense-basis H-SSBP.
\begin{enumerate}

\item In the real--complex boundary experiment, the orthogonal truth remained correctly specified throughout the grid.
The orthogonal fit was essentially flat, with median RMSE between \(0.152\) and \(0.156\), while the dense-basis fit was also flat but higher, between \(0.260\) and \(0.262\).
This indicates that the eigenvalue-regime boundary itself did not create a visible loss of accuracy at this sample size.

\item In the diagonalizable non-orthogonal experiment, the orthogonal fit degraded monotonically as shear increased, with median RMSE rising from \(0.648\) to \(1.803\).
The dense-basis fit remained much lower, between \(0.340\) and \(0.358\), consistent with the fact that the generating drifts are representable by the dense-basis H-SSBP but not by the orthogonal specialization.
The dashed gold curve in this panel is a certified lower bound on the drift RMSE attainable by any orthogonal H-SSBP fit to \(\effOuDriftMatrix_{\effOuShear}\).
It is computed as \(d_{\mathrm{diss}}(\effOuDriftMatrix_{\effOuShear})/\effOuDim\), where
\begin{align}
d_{\mathrm{diss}}(\effOuDriftMatrix_{\effOuShear})
&=
\left\|
\left[
\operatorname{sym}(\effOuDriftMatrix_{\effOuShear})
\right]_+
\right\|_F,
\qquad
\operatorname{sym}(\effOuDriftMatrix_{\effOuShear})
=
\frac{
\effOuDriftMatrix_{\effOuShear}
+
\effOuDriftMatrix_{\effOuShear}^{\effOuTranspose}
}{2}.
\end{align}
Here \([\cdot]_+\) denotes the positive-semidefinite part of a symmetric matrix, obtained by retaining its positive eigenvalues and setting all nonpositive eigenvalues to zero.
Because every orthogonal H-SSBP matrix has negative-definite symmetric part under the stability constraints used here, no orthogonal fit can attain a lower RMSE than this floor even with unlimited data.
Indeed, for any orthogonal H-SSBP candidate \(\effOuDriftMatrix \), the Frobenius projection onto symmetric matrices gives
\(\|\effOuDriftMatrix_{\effOuShear}-\effOuDriftMatrix\|_F
\geq
\|\operatorname{sym}(\effOuDriftMatrix_{\effOuShear})-\operatorname{sym}(\effOuDriftMatrix )\|_F\).
Since \(\operatorname{sym}(\effOuDriftMatrix)\prec0\), the distance from \(\operatorname{sym}(\effOuDriftMatrix_{\effOuShear})\) to the negative-semidefinite cone is
\(\|[\operatorname{sym}(\effOuDriftMatrix_{\effOuShear})]_+\|_F\), yielding the plotted RMSE bound after division by \(\effOuDim\).
The dashed curve therefore certifies an irreducible lower bound.
\item In the Jordan experiment, both fitted parametrizations were misspecified.
The orthogonal fit had median RMSE between \(0.285\) and \(0.444\), while the dense-basis fit ranged from \(0.310\) to \(0.329\).
The dense-basis fit was more stable, consistent with the additional flexibility of a general rather than orthogonal change-of-basis matrix.
\end{enumerate}

\section{Supplementary Material for BitMEX Trade-Data Analysis}
\label{suppeffOu.sec:empirical_details}

This section records the BitMEX post-shock analysis summarized in Section~8.1.2 of the main manuscript.
The records were timestamped trades for actively listed cryptocurrency instruments retrieved through the public BitMEX REST API \citep{bitmex_api_2026}.
We focused on the recovery following the market drawdown observed on January~24,~2022.
We analyzed the following instruments: \texttt{XBTUSDT}, \texttt{ETHUSDT}, \texttt{LTCUSDT}, \texttt{BCHUSDT}, and \texttt{DOGEUSDT}.
The analysis window spans seven hours, from 13:00 to 20:00 UTC.
The common post-shock start time was defined by the one-minute \texttt{XBTUSDT} interval containing the local price trough following the drawdown.
For each instrument and timestamp, trades sharing that instrument and timestamp were aggregated as
\begin{align}
\bar p_{j,t}
&=
\frac{
\sum_{k\in \mathcal I_{j,t}} w_k p_k
}{
\sum_{k\in \mathcal I_{j,t}} w_k
},
\\
y_{j,t}
&=
\log \bar p_{j,t}
-
\frac{1}{n_j}\sum_{s=1}^{n_j}\log \bar p_{j,s},
\end{align}
where \(p_k\) and \(w_k\) are the trade price and size, \(\mathcal I_{j,t}\) is the set of trades for instrument \(j\) at timestamp \(t\), and \(n_j\) is the number of aggregated timestamps for instrument \(j\) in the analyzed window.
The analyzed window contained $10{,}096$ raw trades and $7{,}521$ scalar modeled observation events after same-symbol duplicate timestamps were aggregated.
Time was measured in hours. 
The left plot in Figure~\ref{suppeffOu.fig:bitmex_trade_drift_matrix} shows the centered log-price trends.

We represented the data as five scalar irregular observation series that select natural coordinates of a shared latent OU process.
We fit the general invertible-basis H-SSBP with one scalar block and two ordered two-dimensional blocks, an inferred equilibrium mean $\effOuEquilMean$, and a diffusion matrix $\effOuDiffMat$.
The common scalar observation-error variance was fixed at $10^{-5}$.
The dense change-of-basis matrix was represented by its raw entries and assigned the regularization from Section~5.4 of the main manuscript, with $\lambda=0.1$.
MAP estimates were obtained by L-BFGS optimization from ten random starts.
All ten starts converged to a tight high-posterior cluster within $0.25$ log-posterior units of each other, with final gradient norms at most $1.21$ and change-of-basis condition numbers $\operatorname{cond}(\effOuRotMat)$ between $6.81$ and $12.55$.
The displayed H-SSBP MAP had the highest log posterior among these starts, with final gradient norm $0.14$ and $\operatorname{cond}(\effOuRotMat)=6.81$.

\begin{figure}[H]
\centering
\includegraphics[width=\textwidth]{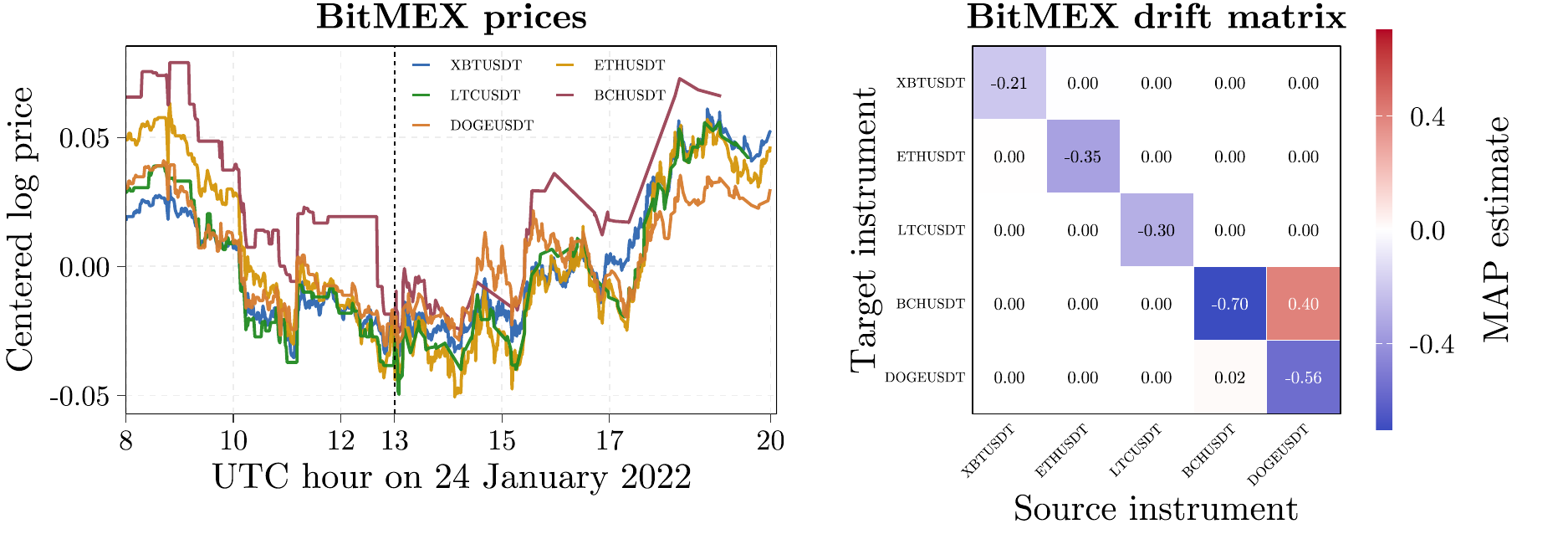}
\caption{
BitMEX January 2022 post-shock trade-data example.
Left: centered log prices from 08:00 to 20:00 UTC, with the dashed vertical line marking the 13:00 UTC start of the analyzed seven-hour recovery window.
The analyzed window uses size-weighted aggregation of trades sharing the same instrument and timestamp.
Right: drift matrix from the general invertible-basis H-SSBP MAP estimate with the highest log posterior among the ten starts.
Rows are target instruments and columns are source instruments, with red/blue indicating positive/negative estimated coupling and white indicating values near zero.
Both diagonal self-reversion entries and off-diagonal couplings are shown.
Color limits are symmetric around zero and include the displayed diagonal entries.
Drift coefficients are per-hour conditional drift couplings for centered log prices.
}
\label{suppeffOu.fig:bitmex_trade_drift_matrix}
\end{figure}

For the displayed highest-posterior MAP in Figure~\ref{suppeffOu.fig:bitmex_trade_drift_matrix}, the fitted diagonal entries ranged from about $-0.21$ to $-0.70$ per hour.
Most off-diagonal entries were close to zero, with a localized positive \texttt{DOGEUSDT}-to-\texttt{BCHUSDT} conditional coupling of comparable order to the self-reversion rates.
Across the ten starts, this was the only off-diagonal pattern with stable sign and a magnitude of comparable order to the diagonal entries.
The fitted two-dimensional block rates satisfied the imposed ordering, and the reconstructed drift spectrum was real and negative.
The corresponding absolute decay rates ranged from about $0.21$ to $0.74$ per hour.
This shows that the H-SSBP is capable of identifying drift matrices with an exclusively real spectrum.

\newpage

\bibliographystyle{apalike}
\bibliography{references}